%% file: URO_and_SIUR-published_version2.tex
\documentclass{article}
\input{ur_preamble.tex}

\begin{document}

\title{\large\bf Uncertainty Regions of Observables and State-Independent Uncertainty Relations}

\author{Lin Zhang$^1$\footnote{E-mail: godyalin@163.com},\,\ Shunlong Luo$^2$,\,\ Shao-Ming Fei$^3$,\,\ and Junde Wu$^4$\\
  {\it\small $^1$Institute of Mathematics, Hangzhou Dianzi University, Hangzhou 310018, PR~China}\\
  {\it\small $^2$Academy of Mathematics and Systems Science, Chinese
Academy of Sciences, Beijing 100190, China } \\
{\it\small $^3$School of Mathematical Sciences, Capital Normal
University, Beijing 100048, China}\\
{\it\small $^4$School of Mathematical Sciences, Zhejiang University,
Hangzhou 310027, China}}
\date{}
\maketitle \tableofcontents

\begin{abstract}

The optimal state-independent lower bounds for the sum of variances
or deviations of observables are of significance for the growing
number of experiments that reach the uncertainty limited regime. We
present a framework for computing the tight uncertainty relations of
variance or deviation via determining the uncertainty regions, which
are formed by the tuples of two or more of quantum observables in
random quantum states induced from the uniform Haar measure on the
purified states. From the analytical formulae of these uncertainty
regions, we present state-independent uncertainty inequalities
satisfied by the sum of variances or deviations of two, three and
arbitrary many observables, from which experimentally friend
entanglement detection criteria are derived for bipartite and
tripartite systems.

\end{abstract}

\newpage

\section{Introduction}

The uncertainty principle, apart from serving as a hallmark of the
quantum world, has wide applications and implications in both
theoretical and practical investigations of quantum mechanics. Ever
since its birth in 1927 \cite{Heisenberg1927}, various uncertainty
relations, as concrete realizations of the uncertainty principle,
have been extensively and intensively studied. In particular,
recently, the state-independent uncertainty relations have attracted
a lot of attentions
\cite{Dammeier2015njp,Li2015,Guise2018pra,Giorda2019pra,Xiao2019pra,Sponar2020pra}.
Whether deeper principles underlie quantum uncertainty and
nonlocality has been listed as one of the challenging scientific
problems on the occasion of celebrating the 125th anniversary of the
academical journal \emph{Science} \cite{Seife2005}. Thus it is of
fundamental significance to explore the intrinsic uncertainty of
given quantum mechanical observables due to its connections with
entanglement detection
\cite{Hofman2003pra,Guhne2004prl,Guhne2009pra,Schwonnek2017prl,Qian2018qip,Zhao2019prl}
and quantum nonlocality \cite{Oppenheim2010}.

The most celebrated uncertainty relation was initially conceived for
position and momentum observables by Heisenberg
\cite{Heisenberg1927}. A general form was the
Robertson-Schr\"{o}dinger uncertainty relation
\cite{Kennard1927,Weyl1928,Robertson1929,Schrodinger1930}:
\begin{eqnarray*}
(\Delta_\rho\bsA )^2 (\Delta _\rho\bsB)^2 \geqslant \frac 14(
\langle\set{\bsA _0,\bsB _0}\rangle_\rho
^2+\langle[\bsA,\bsB]\rangle_\rho^2 ),
\end{eqnarray*}
where $\langle\bsX\rangle_\rho=\Tr{\bsX\rho}$ is the expectation
value of $\bsX$ with respect to the state $\rho$, $(\Delta_\rho\bsA
)^2 =(\langle \bsA^2 \rangle _\rho -\langle \bsA \rangle _\rho ^2)$
is the corresponding variance of $\bsA$, $\bsA_0=\bsA -\langle \bsA
\rangle_\rho,$ $\bsB_0=\bsB -\langle \bsB \rangle_\rho$, $\set{\bsA,
\bsB}=\bsA\bsB+\bsB\bsA$ and $[\bsA, \bsB]=\bsA\bsB-\bsB\bsA$ denote
the anti-commutator (symmetric Jordan product) and commutator
(anti-symmetric Lie product), respectively.

Although the above uncertainty relation captures certain features of
the uncertainty principle in a very appealing, intuitive and
succinct way, the state-dependent lower bound and the variances
capture limited information of uncertainty for given pair of
observables. Recently, Busch and Reardon-Smitha proposed to consider
the uncertainty region of two observables $\bsA$ and $\bsB$ instead
of finding bounds on some particular choice of uncertainty
functionals \cite{Busch2019}, which provides apparently more
information about the uncertainty of the two observables $\bsA$ and
$\bsB$. Later some Vasudevrao \emph{et al} also followed this line
and conducted specific computations about uncertainty regions
\cite{Vasu2021arxiv}.

In this paper, we use the probability theory and random matrices to
study the uncertainty regions of two, three and multiple qubit
observables. The probabilistic method and random matrix theory has
many useful applications such as in evaluating the average entropy
of a subsystem \cite{R0,R1,R2,R3}, studying the non-additivity of
quantum channel capacity \cite{Hastings2009} and random quantum pure
states
\cite{Christandl2014,Dartois2020,Zhang2018pla,Zhang2019jpa,Venuti2013}.
Motivated by these works, we derive analytical formulas concerning
the expectation and uncertainty (variance) of quantum observables in
random quantum states. We identify the uncertainty regions as the
supports of such probability distribution functions. From these
analytical results on uncertainty regions, we present the optimal
state-independent lower bounds for the sum of variances or the
deviations, which are just the optimization problems over the
uncertainty regions. The paper is organized as follows. In
Section~\ref{sec2}, we present our main results, their proofs are
developed in Section~\ref{sec3}.

\section{Main results}\label{sec2}

A qubit observable can be parameterized as
$\bsA=a_0\I+\bsa\cdot\boldsymbol{\sigma}$, $(a_0,\bsa)\in\real^4$,
where $\I$ is the identity matrix on the qubit Hilbert space
$\complex^2$, and $\boldsymbol{\sigma}=(\sigma_1,\sigma_2,\sigma_3)$
is the vector of the standard Pauli matrices. The two eigenvalues of
$\bsA$ are $\lambda_k(\bsA)=a_0+(-1)^k\abs{\bsa}$ $(k=1,2)$ with
$\abs{\bsa}=(a_1^2+a_2^2+a_3^2)^{1/2} >0$ the length of the vector
$\bsa =(a_1,a_2,a_3)\in \real^3$.

Any qubit state $\rho$ can be purified to a pure state on
$\complex^2\ot \complex^2$. The set of  pure states on
$\complex^2\ot \complex^2$ can be represented as
$\Set{\bsU\ket{\Phi}: \bsU\in \U(\complex^2\ot \complex^2)}$ with
$\ket{\Phi} \in \complex^2\ot \complex^2$ any fixed pure state, and
$\U(\complex^2\ot \complex^2)$ the full unitary group on
$\complex^2\ot \complex^2$, endowed with the standard Haar measure.
Denote $\rD(\complex^2)$ the set of all quantum (pure or mixed)
states on $\complex^2$. The probability measure $\dif\mu(\rho)$ can
be derived from this Haar measure by taking partial trace over
$\complex^2$ of the pure states on $\complex^2\ot\complex^2$.

Let $\bsA_k$ $(k=1,\ldots,n)$ be a set of qubit observables. The
uncertainty region $\cU_{\Delta\bsA_1,\ldots,\Delta\bsA_n}$ of an
$n$-tuple of uncertainties $\Delta\bsA_k$ is defined by
\begin{eqnarray}\label{un}
\Set{(\Delta_\rho\bsA_1,\ldots,\Delta_\rho\bsA_n)\in\real^n_{\geqslant0}:\rho\in\rD(\complex^2)}.
\end{eqnarray}
Here $\real_{\geqslant0}=[0,+\infty)$. From \eqref{un}, tight
state-independent variance and deviation uncertainty relations can
be obtained,
\begin{eqnarray}
\sum^n_{k=1}(\Delta_\rho\bsA_k)^2&\geqslant&
\min_{(x_1,\ldots,x_n)\in\cU_{\Delta\bsA_1,\ldots,\Delta\bsA_n}}\sum^n_{k=1}x^2_k,\\
\sum^n_{k=1}\Delta_\rho\bsA_k&\geqslant&
\min_{(x_1,\ldots,x_n)\in\cU_{\Delta\bsA_1,\ldots,\Delta\bsA_n}}\sum^n_{k=1}x_k.
\end{eqnarray}

We first investigate the uncertainty regions for a pair of
observables $\bsA =a_0\I+\bsa\cdot\boldsymbol{\sigma}$ and $\bsB
=b_0\I+\bsb\cdot\boldsymbol{\sigma}$ with $(a_0,\bsa),~(b_0,\bsb)\in
\real^4$. Denote
\begin{eqnarray*}
\bsT_{\bsa,\bsb} = \Pa{\begin{array}{cc}
                         \Inner{\bsa}{\bsa} & \Inner{\bsa}{\bsb} \\
                         \Inner{\bsb}{\bsa} & \Inner{\bsb}{\bsb}
                       \end{array}}.
\end{eqnarray*}

\begin{lem}\label{lem1}
The uncertainty region
$\cU_{\Delta\bsA,\Delta\bsB}=\Set{(\Delta_\rho \bsA,\Delta_\rho
\bsB)\in\real^2_{\geqslant0}: \rho\in\density{\complex^2}}$ of
$\bsA$ and $\bsB$, is determined by the following inequality:
\begin{eqnarray}
\abs{\bsb}^2x^2+\abs{\bsa}^2y^2+2\abs{\Inner{\bsa}{\bsb}}
\sqrt{(\abs{\bsa}^2-x^2)(\abs{\bsb}^2-y^2)}\geqslant
\abs{\bsa}^2\abs{\bsb}^2+\Inner{\bsa}{\bsb}^2,
\end{eqnarray}
where $x\in[0,\abs{\bsa}]$, $y\in[0,\abs{\bsb}]$ and
$\set{\bsa,\bsb}$ is linearly independent.
\end{lem}

From Lemma~\ref{lem1}, the boundary curve of the uncertainty region
$\cU_{\Delta\bsA,\Delta\bsB}$ is easily obtained. Clearly, the
boundary curve equations are independent of $(a_0,b_0)$. Denote
$\theta$ the angle between $\bsa$ and $\bsb$. We see that, for
$\abs{\bsa}=\abs{\bsb}=1$,
\begin{eqnarray}
x^2+y^2+2\abs{\cos\theta}\sqrt{(1-x^2)(1-y^2)}\geqslant
1+\cos^2\theta.
\end{eqnarray}
In particular, for
$\theta\in\Set{\frac\pi{16},\frac\pi8,\frac\pi4,\frac{3\pi}8,\frac{7\pi}{16},\frac\pi2}$,
our result covers that in \cite{Busch2019,Li2015} perfectly.
Obviously, $\cU_{\Delta\bsA,\Delta\bsB}\equiv \cU(\theta)$ is
determined only by the angle $\theta$ between $\bsa$ and $\bsb$ if
$\abs{\bsa}=\abs{\bsb}=1$, and $\cU(\theta)=\cU(\pi-\theta)$ for
$\theta\in[0,\pi]$. Moreover, we can also get the volume (the area
for the 2-dimensional domain) of the uncertainty region
$\cU(\theta)(\theta\in[0,\pi/2])$,
\begin{eqnarray}\label{vol2}
\vol(\cU(\theta))=\frac12(\pi-3\theta) \sin\theta-\cos\theta+1.
\end{eqnarray}
Hence in general,
$\vol(\cU_{\Delta\bsA,\Delta\bsB})=\abs{\bsa}\abs{\bsb}\vol(\cU(\theta))$.
From \eqref{vol2}, we have that the maximal uncertainty region is
attained at $\theta_0\simeq0.741758<\frac\pi4$,
$$
\max_{\theta\in[0,\pi/2]}\vol[\cU(\theta)]=\vol[\cU(\theta_0)]\simeq0.572244.
$$

From the uncertainty region given in Lemma~\ref{lem1}, we can derive
the optimal \emph{state-independent} lower bound for the sum of
variances $(\Delta_\rho\bsA)^2+(\Delta_\rho\bsB)^2$ and the sum of
standard deviation $\Delta_\rho\bsA+\Delta_\rho\bsB$.

\begin{thrm}\label{th1}
The sum of the variances and the standard
deviations satisfy the following tight inequalities with
state-independent lower bounds,
\begin{eqnarray}\label{th21}
(\Delta_\rho\bsA)^2+(\Delta_\rho\bsB)^2
\geqslant\min\Set{x^2+y^2:(x,y)\in\cU_{\Delta\bsA,\Delta\bsB}}=\lambda_{\min}(\bsT_{\bsa,\bsb})
\end{eqnarray}
and
\begin{eqnarray}\label{th211}
\Delta_\rho\bsA+\Delta_\rho\bsB\geqslant\min\Set{x+y:(x,y)\in\cU_{\Delta\bsA,\Delta\bsB}}=\frac{\abs{\bsa\times\bsb}}{\max(\abs{\bsa},\abs{\bsb})},
\end{eqnarray}
where $\lambda_{\min}(\bsT_{\bsa,\bsb})$ stands for the minimal
eigenvalue of $\bsT_{\bsa,\bsb}$.
\end{thrm}

In particular, if $\abs{\bsa}=\abs{\bsb}=1$, \eqref{th21} and
\eqref{th211} are reduced to
\begin{eqnarray*}
(\Delta_\rho\bsA)^2+(\Delta_\rho\bsB)^2 &\geqslant&
1-\abs{\inner{\bsa}{\bsb}},\\
\Delta_\rho\bsA+\Delta_\rho\bsB &\geqslant&\abs{\bsa\times\bsb}.
\end{eqnarray*}
Note that
$[\bsA,\bsB]=\bsA\bsB-\bsB\bsA=2\mathrm{i}(\bsa\times\bsb)\cdot\boldsymbol{\sigma}$,
$\abs{\Inner{\bsa}{\bsb}}=\abs{\cos\theta}$ and
$\abs{\bsa\times\bsb}=\abs{\sin\theta}$. One recovers the
uncertainty relations \cite{Busch2019},
\begin{eqnarray*}
(\Delta_\rho\bsA)^2+(\Delta_\rho\bsB)^2 &\geqslant& 1 -
\sqrt{1-\frac14\norm{[\bsA,\bsB]}^2},\\\nonumber
\Delta_\rho\bsA+\Delta_\rho\bsB&\geqslant&\frac12\norm{[\bsA,\bsB]}.
\end{eqnarray*}
As an illustration, in Fig. \ref{theta} we plot for any fixed
angle $\theta\in(0,\tfrac\pi2)$ the uncertainty region
$\cU_{\Delta\bsA,\Delta\bsB}$ given in Lemma~\ref{lem1} and those points where
the minimizations in Theorem~\ref{th1} are achieved, where the coordinate of the red point is
$(\sqrt{(1-\cos\theta)/2},\sqrt{(1-\cos\theta)/2})$, and the
coordinates of the two purple points are $(\sin\theta,0)$ and
$(0,\sin\theta)$, respectively.
\begin{figure}[ht]\centering
{\begin{minipage}[b]{0.5\linewidth}
\includegraphics[width=1\textwidth]{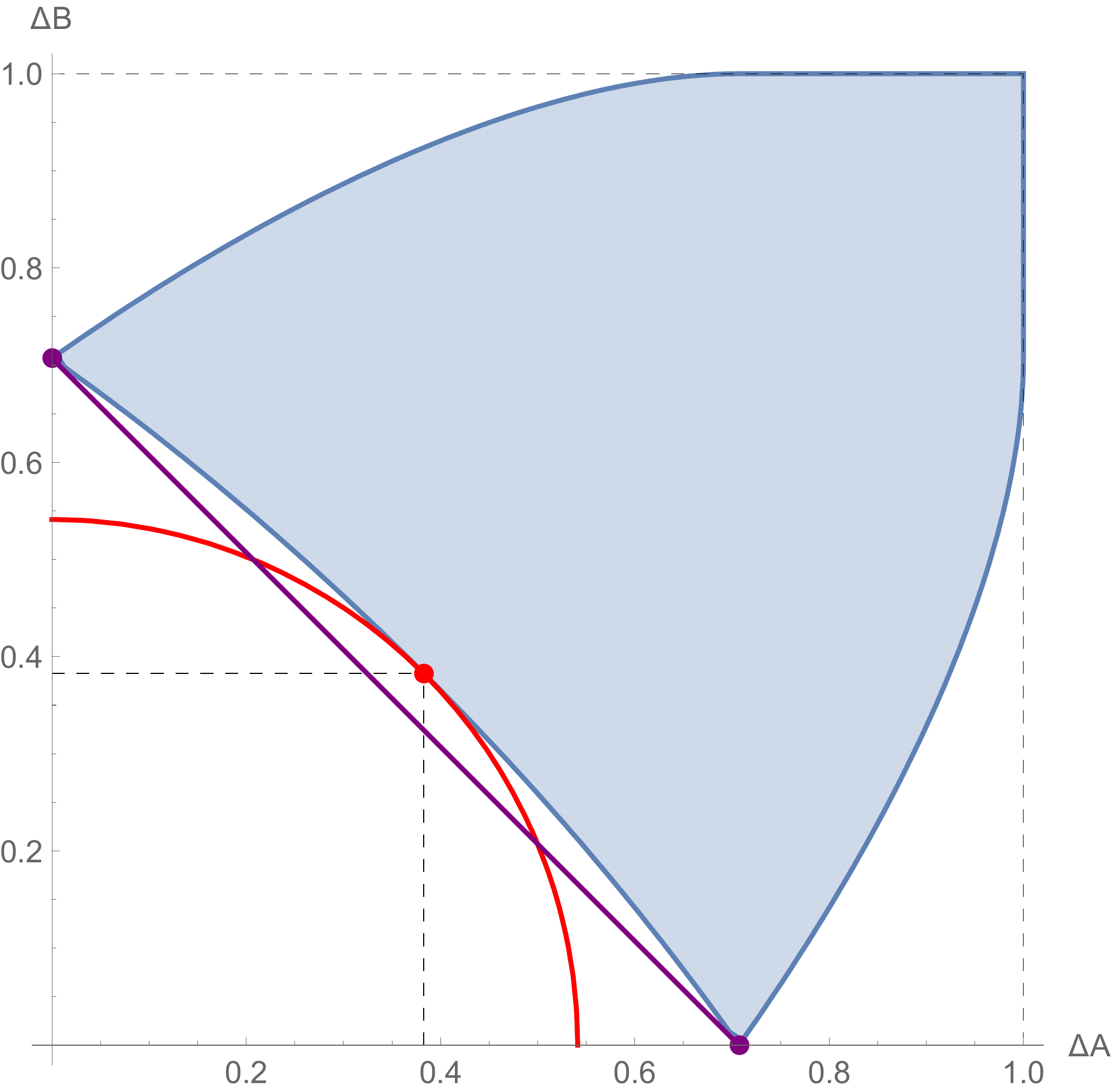}
\end{minipage}}
\caption{The uncertainty region $\cU_{\Delta\bsA,\Delta\bsB}$ in
Lemma~\ref{lem1} and those points where the minimizations in
Theorem~\ref{th1} are achieved.}\label{theta}
\end{figure}

We now turn to the uncertainty region for a triple
$(\bsA,\bsB,\bsC)$ of qubit observables
$\bsA=a_0\I+\bsa\cdot\boldsymbol{\sigma}$,
$\bsB=b_0\I+\bsb\cdot\boldsymbol{\sigma}$ and
$\bsC=c_0\I+\bsc\cdot\boldsymbol{\sigma}$ with $(a_0,\bsa)$,
$(b_0,\bsb)$, $(c_0,\bsc)\in\real^4$, and $\set{\bsa,\bsb,\bsc}$
being linearly independent. Denote
$\bsu_{\epsilon_b,\epsilon_c}(x,y,z)=((\abs{\bsa}^2-x^2)^{1/2},\epsilon_b(\abs{\bsb}^2-y^2)^{1/2},
\epsilon_c(\abs{\bsc}^2-z^2)^{1/2})$ and
\begin{eqnarray*}
\bsT_{\bsa,\bsb,\bsc}=\Pa{\begin{array}{ccc}
                               \Inner{\bsa}{\bsa} & \Inner{\bsa}{\bsb} & \Inner{\bsa}{\bsc} \\
                               \Inner{\bsb}{\bsa} & \Inner{\bsb}{\bsb} & \Inner{\bsb}{\bsc} \\
                               \Inner{\bsc}{\bsa} & \Inner{\bsc}{\bsb} & \Inner{\bsc}{\bsc}
                             \end{array}}.
\end{eqnarray*}
Let $\gamma$, $\beta$ and $\alpha$ be the angles between $\bsa$ and
$\bsb$, $\bsa$ and $\bsc$, $\bsb$ and $\bsc$, respectively, where
$\alpha,\beta,\gamma\in(0,\pi)$. Set $\phi(t_1,t_2,t_3)=\cos
(t_1)-\cos (t_2)\cos (t_3)$. We have the following result:

\begin{lem}\label{lem2}
The uncertainty region
$\cU_{\Delta\bsA,\Delta\bsB,\Delta\bsC}=\Set{(\Delta_\rho\bsA,\Delta_\rho\bsB,\Delta_\rho\bsC)\in\real^3_{\geqslant0}:
\rho\in\rD(\complex^2)}$ is determined by the following inequality:
\begin{eqnarray}
\bsu_{\epsilon_b,\epsilon_c}(x,y,z)\bsT^{-1}_{\bsa,\bsb,\bsc}\bsu^\t_{\epsilon_b,\epsilon_c}(x,y,z)\leqslant1,
\end{eqnarray}
where $\epsilon_b,\epsilon_c\in\set{\pm1}$ are independent of each
other, which is given by the union of solutions of the following
inequalities if $\abs{\bsa}=\abs{\bsb}=\abs{\bsc}=1$:
\begin{eqnarray}\label{eq:3DUR}
&&\abs{\phi(\gamma,\alpha,\beta)\sqrt{1-x^2}
+\epsilon\phi(\alpha,\beta,\gamma)\sqrt{1-z^2}}\sqrt{1-y^2}+
\epsilon\phi(\beta,\gamma,\alpha)\sqrt{(1-z^2)(1-x^2)}\notag\\
&&+\frac12\Br{\sin^2(\alpha)x^2 + \sin^2(\beta)y^2 +
\sin^2(\gamma)z^2}\geqslant1-\cos(\alpha)\cos(\beta)\cos(\gamma)
\end{eqnarray}
under the conditions $\alpha<\beta+\gamma$, $\beta<\gamma+\alpha$,
$\gamma<\alpha+\beta$ and $\alpha+\beta+\gamma<2\pi$, where
$\epsilon\equiv\epsilon_c\in\set{\pm1}$.
\end{lem}
In particular, e.g., for $\alpha=\beta=\gamma=\frac\pi2$, the
uncertainty region is just
$\set{(x,y,z)\in\real^3_{\geqslant0}:x^2+y^2+z^2\geqslant2}\cap[0,1]^3$.
Clearly, our equations include the results in \cite{Busch2019} as
special cases, see Fig.~\ref{fig:3DUR}.
\begin{figure}[ht]
\subfigure[The uncertainty region for $\cU(\pi/2,\pi/2,\pi/2)$]
{\begin{minipage}[b]{.45\linewidth}
\includegraphics[width=1\textwidth]{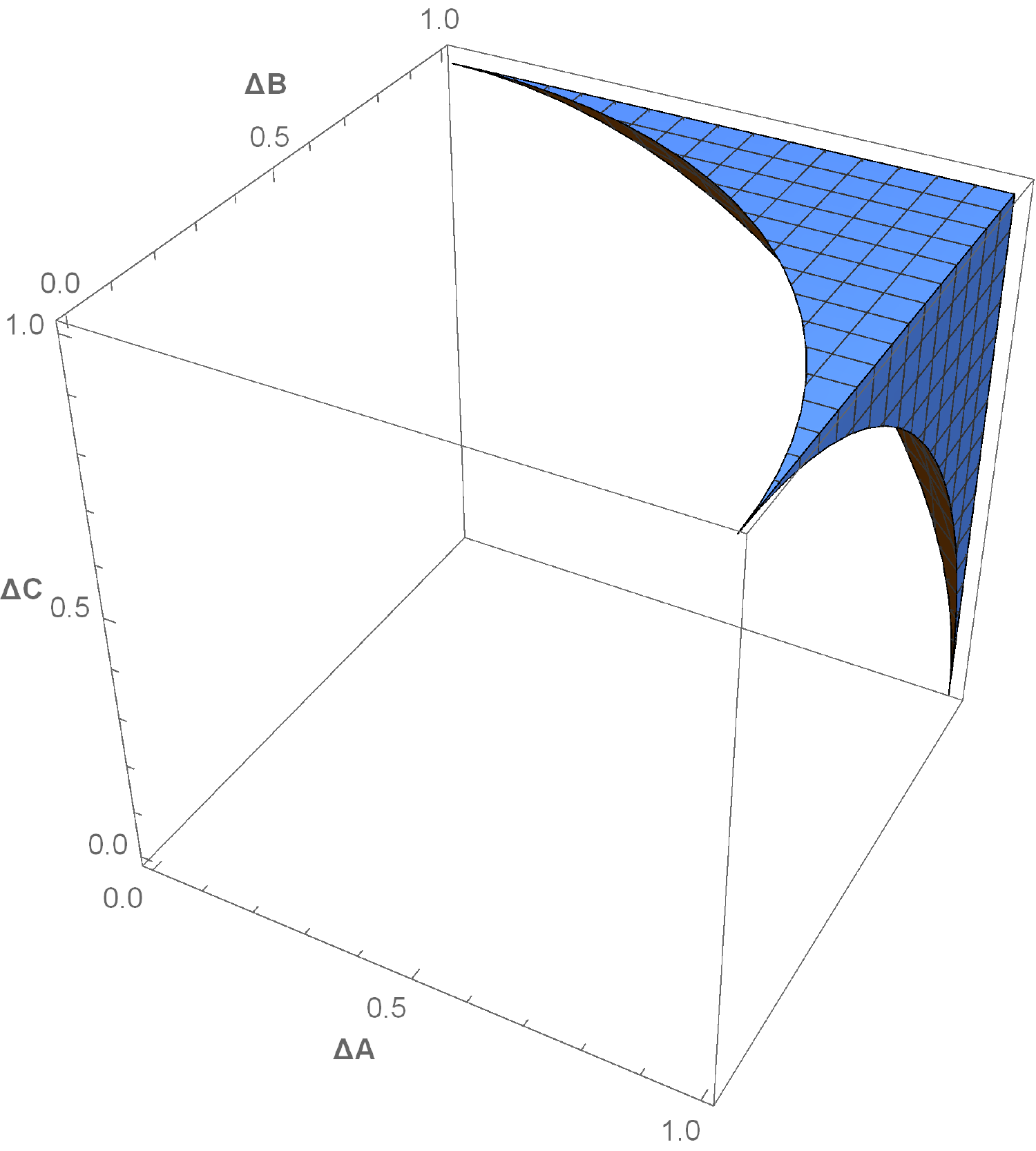}
\end{minipage}}
\subfigure[The uncertainty region for $\cU(\pi/4,\pi/4,\pi/4)$]
{\begin{minipage}[b]{.45\linewidth}
\includegraphics[width=1\textwidth]{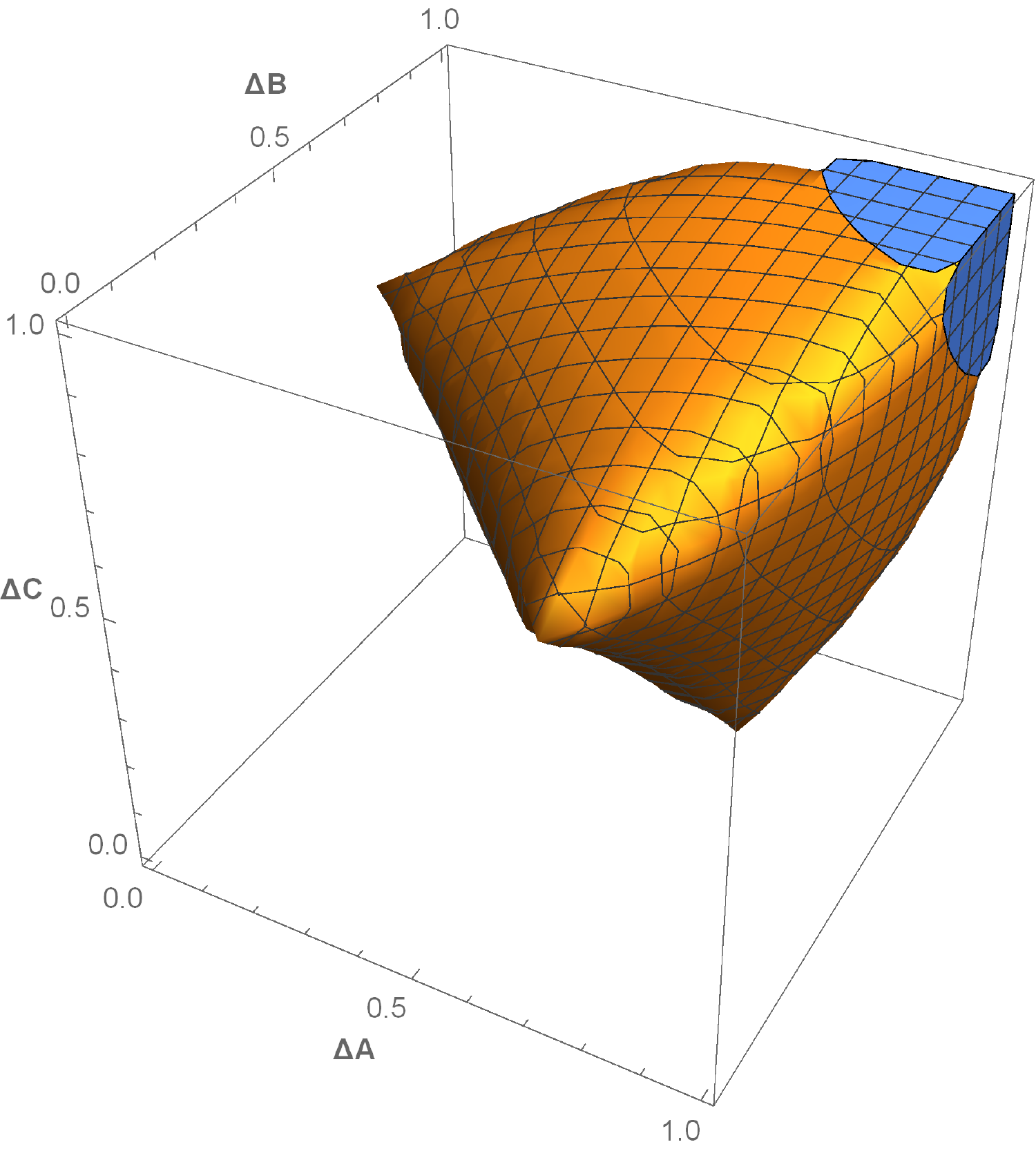}
\end{minipage}}
\subfigure[The uncertainty region for $\cU(5\pi/8,5\pi/8,5\pi/8)$]
{\begin{minipage}[b]{.45\linewidth}
\includegraphics[width=1\textwidth]{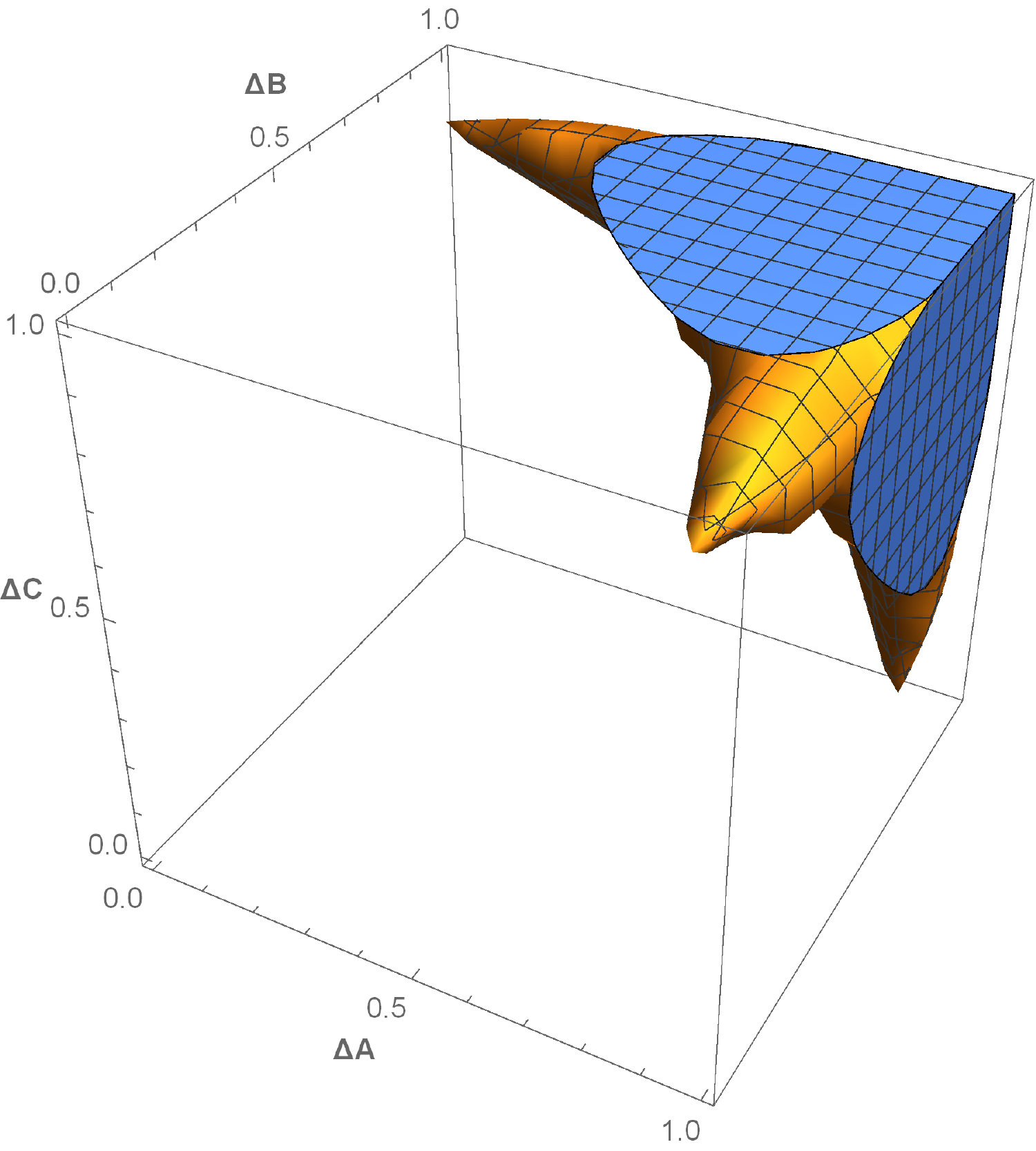}
\end{minipage}}
\subfigure[The uncertainty region for $\cU(\pi/2,\pi/3,\pi/4)$]
{\begin{minipage}[b]{.45\linewidth}
\includegraphics[width=1\textwidth]{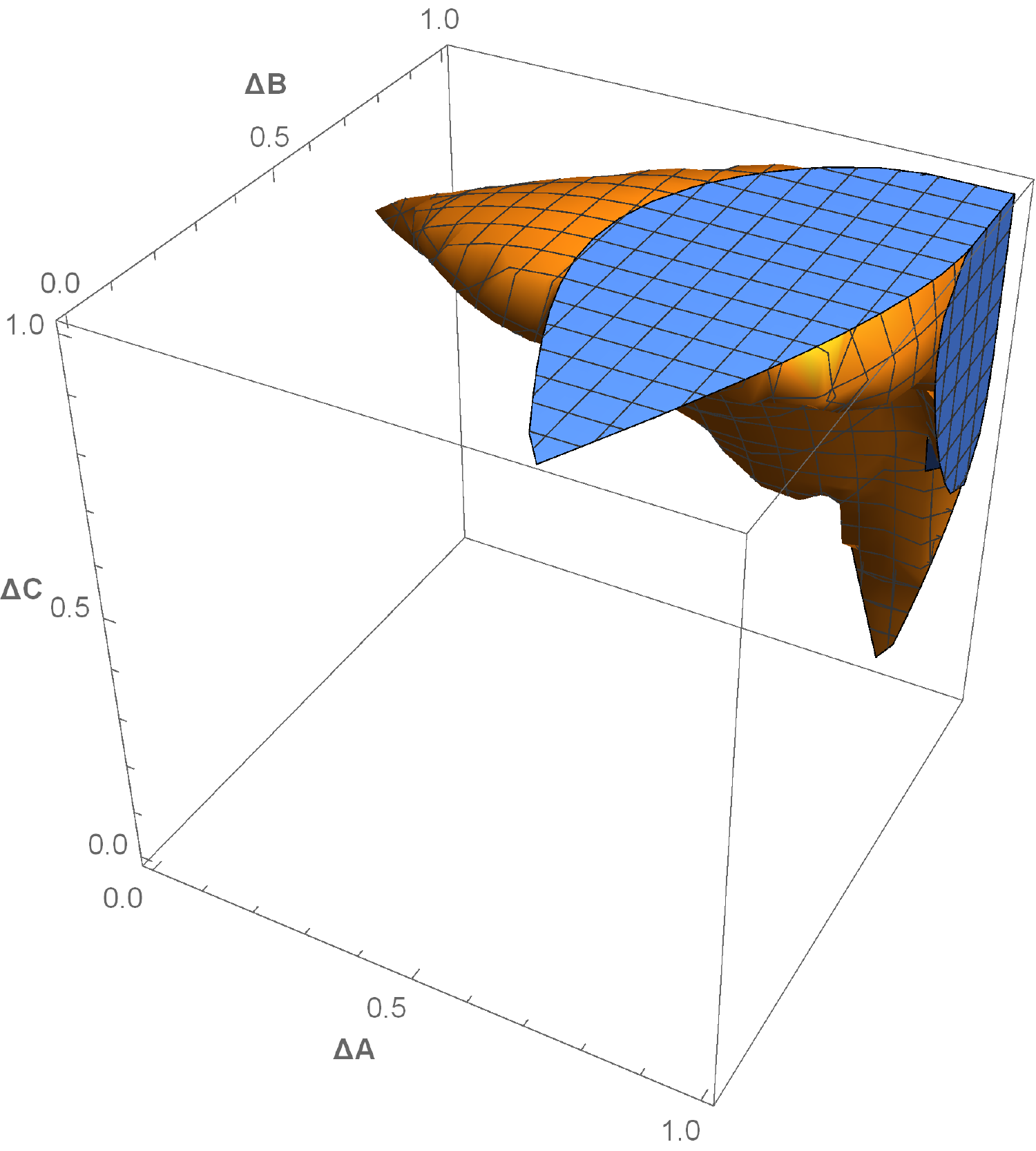}
\end{minipage}}
\caption{The uncertainty regions
$\cU_{\Delta\bsA,\Delta\bsB,\Delta\bsC}\equiv\cU(\alpha,\beta,\gamma)$
for a triple of qubit observables
$\bsA=a_0\I+\bsa\cdot\boldsymbol{\sigma}$,
$\bsB=b_0\I+\bsb\cdot\boldsymbol{\sigma}$ and
$\bsC=c_0\I+\bsc\cdot\boldsymbol{\sigma}$, where
$\abs{\bsa}=\abs{\bsb}=\abs{\bsc}=1$,
$\Inner{\bsa}{\bsb}=\cos\gamma$, $\Inner{\bsa}{\bsc}=\cos\beta$ and
$\Inner{\bsb}{\bsc}=\cos\alpha$.} \label{fig:3DUR}
\end{figure}

Analogously, the state-independent lower bound for
$(\Delta_\rho\bsA)^2+(\Delta_\rho\bsB)^2+(\Delta_\rho\bsC)^2$ can be
obtained from \eqref{eq:3DUR}.

\begin{thrm}\label{th2}
The variances of the observables $\bsA$,
$\bsB$ and $\bsC$ satisfy
\begin{eqnarray}\label{eq:zhanglin2021}
(\Delta_\rho\bsA)^2+(\Delta_\rho\bsB)^2+(\Delta_\rho\bsC)^2&\geqslant&
\min\set{x^2+y^2+z^2:(x,y,z)\in\cU_{\Delta\bsA,\Delta\bsB,\Delta\bsC}}\notag\\
&=&\Tr{\bsT_{\bsa,\bsb,\bsc}}-\lambda_{\max}(\bsT_{\bsa,\bsb,\bsc}),
\end{eqnarray}
where $\lambda_{\max}(\bsT_{\bsa,\bsb,\bsc})$ stands for the maximal
eigenvalue of $\bsT_{\bsa,\bsb,\bsc}$.
\end{thrm}

We analyze the state-independent lower bound of
\eqref{eq:zhanglin2021} in two cases:

(i) If $\alpha=\beta=\gamma=\theta\in[0,\pi/2]$, then the three
eigenvalues of $\bsT_{\bsa,\bsb,\bsc}$ are given by
$\set{1-\cos\theta,1-\cos\theta,1+2\cos\theta}$, namely, the maximal
eigenvalue is $\lambda_{\max}(\bsT_{\bsa,\bsb,\bsc})=1+2\cos\theta$.
Therefore, the uncertainty relation \eqref{eq:zhanglin2021} simply
becomes
$(\Delta_\rho\bsA)^2+(\Delta_\rho\bsB)^2+(\Delta_\rho\bsC)^2\geqslant
2(1-\cos\theta)$;

(ii) If $\beta=\gamma=\frac\pi2$ and $\alpha\in[0,\pi/2]$, the three
eigenvalues of $\bsT_{\bsa,\bsb,\bsc}$ are
$\set{1,1-\cos\alpha,1+\cos\alpha}$. The maximal eigenvalue is
$\lambda_{\max}(\bsT_{\bsa,\bsb,\bsc})=1+\cos\alpha$. In this case
the uncertainty relation \eqref{eq:zhanglin2021} becomes
$(\Delta_\rho\bsA)^2+(\Delta_\rho\bsB)^2+(\Delta_\rho\bsC)^2\geqslant2-\cos\alpha$.

We consider now the uncertainty regions for multiple qubit
observables. For an $n$-tuple of qubit observables
$(\bsA_1,\ldots,\bsA_n)$, where
$\bsA_k=a^{(k)}_0\I+\bsa_k\cdot\boldsymbol{\sigma}$ with
$(a^{(k)}_0, \bsa_k)\in\real^4$, $k=1,\ldots,n$, denote
$\bsT_{\bsa_1,\ldots,\bsa_n}=(\inner{\bsa_i}{\bsa_j})$. Note that
$\set{\bsa_1,\bsa_2,...,\bsa_n}$ has at most three vectors that are
linearly independent. Without loss of generality, we assume
$\set{\bsa_1,\bsa_2,\bsa_3}$ is linearly independent. The rest
vectors can be linearly expressed by $\set{\bsa_1,\bsa_2,\bsa_3}$,
$\bsa_l=\kappa_{l1}\bsa_1+\kappa_{l2}\bsa_2+\kappa_{l3}\bsa_3$, for
some coefficients $\kappa_{lj}$, $l=4,\cdots,n$, $j=1,2,3$. Set
$\bsu_{\epsilon_1,\epsilon_2,\epsilon_3}(x_1,x_2,x_3)=\Pa{\epsilon_1\sqrt{\abs{\bsa_1}^2-x^2_1},\epsilon_2\sqrt{\abs{\bsa_2}^2-x^2_2},
\epsilon_3\sqrt{\abs{\bsa_3}^2-x^2_3}}$.

\begin{lem}\label{lem3}
The uncertainty region
$\cU_{\Delta\bsA_1,\ldots,\Delta\bsA_n}$ of an $n$-tuple of qubit
observables $(\bsA_1,\ldots,\bsA_n)$ is determined by
$(x_1,\ldots,x_n)\in\cU_{\Delta\bsA_1,\ldots,\Delta\bsA_n}$
satisfying
\begin{eqnarray*}
\begin{cases}
\bsu_{\epsilon_1,\epsilon_2,\epsilon_3}(x_1,x_2,x_3)\bsT^{-1}_{\bsa_1,\bsa_2,\bsa_3}\bsu^\t_{\epsilon_1,\epsilon_2,\epsilon_3}(x_1,x_2,x_3)\leqslant1\\[2mm]
\epsilon_l\sqrt{\abs{\bsa_l}^2-x^2_l}=\sum^3_{j=1}\kappa_{lj}\epsilon_j\sqrt{\abs{\bsa_j}^2-x^2_j}
\end{cases}
\end{eqnarray*}
$\forall l=4,\ldots,n$, where $\epsilon_k\in\set{\pm1}$,
$x_k\in[0,\abs{\bsa_k}]$, $k=1,\ldots,n$.
\end{lem}

Correspondingly, we have the following result:

\begin{thrm}\label{th3}
The variances of the observables
$\bsA_k(k=1,\ldots,n)$ satisfy
\begin{eqnarray}
\sum^n_{k=1}(\Delta_\rho\bsA_k)^2\geqslant
\min_{(x_1,\ldots,x_n)\in\cU_{\Delta\bsA_1,\ldots,\Delta\bsA_n}}\sum^n_{k=1}x^2_k=\Tr{\bsT_{\bsa_1,\ldots,\bsa_n}}-\lambda_{\max}(\bsT_{\bsa_1,\ldots,\bsa_n}),
\end{eqnarray}
where $\lambda_{\max}(\bsT_{\bsa_1,\ldots,\bsa_n})$ stands for the
maximal eigenvalue of $\bsT_{\bsa_1,\ldots,\bsa_n}$.
\end{thrm}

As applications of our state-independent uncertainty relations, we
consider the entanglement detection. As shown in
\cite{Hofman2003pra,Guhne2009pra,Schwonnek2017prl}, every
state-independent uncertainty relation gives rise to a nonlinear
entanglement witness. We consider the tripartite scenario: three
parties, Alice, Bob and Charlie perform local measurements
$\bsA_i,\bsB_i$ and $\bsC_i$, where $i=1,2,3$, on an unknown
tripartite quantum state $\rho_{ABC}\equiv\rho$, acting on
$\complex^2\ot\complex^2\ot\complex^2$, respectively. Their goal is
to decide if $\rho$ is fully separable or not. They measure the
composite observables $\bsM_i$ $(i=1,2,3)$ given by
\begin{eqnarray*}
\bsM_i=\bsA_i\ot\I_{BC}+\I_A\ot\bsB_i\ot\I_C + \I_{AB}\ot\bsC_i.
\end{eqnarray*}

Note that
\begin{eqnarray*}
(\Delta_\rho\bsM_i)^2&=&2(\Tr{\bsA_i\ot\bsB_i\rho_{AB}}-\Tr{\bsA_i\rho_A}\Tr{\bsB_i\rho_B})\\
&&+2(\Tr{\bsB_i\ot\bsC_i\rho_{BC}}-\Tr{\bsB_i\rho_B}\Tr{\bsC_i\rho_C})\\
&&+2(\Tr{\bsA_i\ot\bsC_i\rho_{AC}}-\Tr{\bsA_i\rho_A}\Tr{\bsC_i\rho_C})\\
&&+(\Delta_{\rho_A}\bsA_i)^2+(\Delta_{\rho_B}\bsB_i)^2+(\Delta_{\rho_C}\bsC_i)^2.
\end{eqnarray*}
If $\rho$ is of the form, $\rho=\rho_A\ot\rho_B\ot\rho_C$, then
$\Delta^2_\rho\bsM_i=\Delta^2_{\rho_A}\bsA_i+\Delta^2_{\rho_B}\bsB_i+\Delta^2_{\rho_C}\bsC_i$.
Thus for fully separable states $\rho=\sum_k p_k\rho_k$, where
$\rho_k=\rho_{k,A}\ot\rho_{k,B}\ot\rho_{k,C}$, we get
\begin{eqnarray*}
&&\Delta^2_\rho\bsM_1+\Delta^2_\rho\bsM_2+\Delta^2_\rho\bsM_3\\
&&\geqslant\sum_k p_k
\Pa{\Delta^2_{\rho_k}\bsM_1+\Delta^2_{\rho_k}\bsM_2+\Delta^2_{\rho_k}\bsM_3}\\
&&=\sum_k
p_k\Pa{\Delta^2_{\rho_{k,A}}\bsA_1+\Delta^2_{\rho_{k,A}}\bsA_2+\Delta^2_{\rho_{k,A}}\bsA_3}\\
&&~~~+\sum_k
p_k\Pa{\Delta^2_{\rho_{k,B}}\bsB_1+\Delta^2_{\rho_{k,B}}\bsB_2+\Delta^2_{\rho_{k,B}}\bsB_3}\\
&&~~~+\sum_k
p_k\Pa{\Delta^2_{\rho_{k,C}}\bsC_1+\Delta^2_{\rho_{k,C}}\bsC_2+\Delta^2_{\rho_{k,C}}\bsC_3}.
\end{eqnarray*}
Denote
$m^{(3)}_{X}=\min\set{x^2+y^2+z^2:(x,y,z)\in\cU_{\Delta\bsX_1,\Delta\bsX_2,\Delta\bsX_3}}$,
$X=A,B,C$, the optimal uncertainty bounds given by Theorem~\ref{th2}
for the observable triples $(\bsA_1,\bsA_2,\bsA_3)$,
$(\bsB_1,\bsB_2,\bsB_3)$ and $(\bsC_1,\bsC_2,\bsC_3)$, respectively.
We have

\begin{thrm}\label{th4}
If a tripartite state $\rho$ is fully
separable, then
\begin{eqnarray}\label{eq:entnc3nd}
&&\Delta^2_\rho\bsM_1+\Delta^2_\rho\bsM_2+\Delta^2_\rho\bsM_3\geqslant
m^{(3)}_{A}+m^{(3)}_{B}+m^{(3)}_{C}.
\end{eqnarray}
\end{thrm}

From Theorem~\ref{th4}, one has that if \eqref{eq:entnc3nd} is
violated, $\rho$ must be entangled. In addition, denote
$$
m^{(3)}_M = \min_{\rho\in\rD(\complex^2)} \Pa{\Delta^2_\rho\bsM_1
+\Delta^2_\rho\bsM_2+\Delta^2_\rho\bsM_3},
$$
the uncertainty bound for the observable triple
$(\bsM_1,\bsM_2,\bsM_3)$. We have that if $\rho$ is not fully
separable, then
\begin{eqnarray}\nonumber
m^{(3)}_A+m^{(3)}_B+m^{(3)}_C>\Delta^2_\rho\bsM_1
+\Delta^2_\rho\bsM_2+\Delta^2_\rho\bsM_3\geqslant& m^{(3)}_M.
\end{eqnarray}

Instead of three measurements, we may also consider two measurements
with three observables each. Denote
$m^{(2)}_{X}=\min\set{x^2+y^2:(x,y)\in\cU_{\Delta\bsX_1,\Delta\bsX_2}}$
$(X=A,B,C)$ the optimal uncertainty bounds for the observables pairs
$(\bsA_1,\bsA_2)$, $(\bsB_1,\bsB_2)$ and $(\bsC_1,\bsC_2)$,
respectively, given by Theorem~\ref{th1}. Similarly, we get
\begin{eqnarray}\label{eq:entnc2nd}
\Delta^2_\rho\bsM_1+\Delta^2_\rho\bsM_2\geqslant
m^{(2)}_A+m^{(2)}_B+m^{(2)}_C.
\end{eqnarray}
If Eq.~\eqref{eq:entnc2nd} is violated, $\rho$ must be entangled.
Let $m^{(2)}_M$ be the uncertainty bound for the observable pair
$(\bsM_1,\bsM_2)$. We can draw a new criterion that $\rho$ is
entangled for
\begin{eqnarray}\nonumber
m^{(2)}_A+m^{(2)}_B+m^{(2)}_C>\Delta^2_\rho\bsM_1+\Delta^2_\rho\bsM_2\geqslant
m^{(2)}_M.
\end{eqnarray}

In \cite{Schwonnek2017prl}, the authors considered the entanglement
detect for bipartite scenario with two measurements and obtained
that
\begin{eqnarray}\label{pk}
\Delta^2_\rho\bsM_1+\Delta^2_\rho\bsM_2 \geqslant
m^{(2)}_A+m^{(2)}_B.
\end{eqnarray}
From Theorem~\ref{th4} we can also consider the entanglement detect
for bipartite systems with three measurements,
\begin{eqnarray}\nonumber
\bsM_i=\bsA_i\ot\I+\I\ot\bsB_i,\quad i=1,2,3.
\end{eqnarray}
We have
\begin{eqnarray}\label{ppp}
\Delta^2_\rho\bsM_1+\Delta^2_\rho\bsM_2+\Delta^2_\rho\bsM_3
\geqslant m^{(3)}_A+m^{(3)}_B,
\end{eqnarray}
where $m^{(3)}_A$ and $m^{(3)}_B$ are the ones given in Theorem~2.
The criterion (\ref{ppp}) is a new one that is different from
(\ref{pk}) given in \cite{Schwonnek2017prl}.

We have investigated uncertainty relations of quantum observables in
a random quantum state, by deriving explicitly the probability
distribution densities of uncertainty for two, three and multiple
qubit observables. As the supports of these density functions, the
uncertainty regions are analytically derived. The advantage of the
probabilistic approach used in the paper is that it gives a unified
framework from which one can obtain the correlations (PDF) among
uncertainties of multiple observables and derive analytically the
uncertainty regions. Various state-independent uncertainty relations
may be derived from the uncertainty regions. Throughout this paper,
we have focused on qubit observables. Our framework may be also
applied to the case of qudit obsevables with random mixed quantum
state ensembles.

\section{Proofs of main results}\label{sec3}
\subsection{Proof of Lemma~\ref{lem1}}\label{app-lemma1}

The proof of Lemma~\ref{lem1} will be essentially recognized as a
series of Propositions~\ref{prop1} -- \ref{prop8}. Note that in the
proof of Lemma~\ref{lem1}, we directly use Propositions~\ref{prop7}
and \ref{prop8}, which are in fact based on the previous
Propositions~\ref{prop1} -- \ref{prop6}.

\begin{prop}[Harish-Chandra-Itzykson-Z\"{u}ber integral \cite{Itzykson1980}]\label{prop1}
Let $\bsA$ and $\bsB$ be $n\times n$ Hermitian matrices with
eigenvalues $\lambda_1(\bsA)< \cdots< \lambda_n(\bsA)$ and
$\lambda_1(\bsB)< \cdots< \lambda_n(\bsB),$ then
\begin{eqnarray*}
\int_{\mathsf{U}(\complex^n)} e^{z\Tr{\bsA\bsU\bsB\bsU^\dagger}}
\dif\mu_{\haar}(\bsU) =C_n\frac{\det \Pa{
e^{z\lambda_i(\bsA)\lambda_j(\bsB)}}}{z^{\binom{n}{2}}
V(\lambda(\bsA)) V(\lambda(\bsB))}\quad (\forall z\in
\complex\backslash\set{0})
\end{eqnarray*}
where $\dif\mu_{\haar}$ is the Haar measure on the unitary group
$\mathsf{U}(\complex^n)$, $C_n=\prod^n_{k=1}\Gamma(k)$, and
$V(\lambda (\bsA))=\prod_{1\leqslant i< j\leqslant
n}(\lambda_j(\bsA)- \lambda_i(\bsA))$ is the so-called Vandermonde
determinant.
\end{prop}
Any state $\rho\in \rD(\complex^n)$, the set of all quantum states
(pure or mixed)  on $\complex^n$, can be purified to a bipartite
pure state on $\complex^n\ot \complex^n$. The set of pure states on
$\complex^n\ot \complex^n$ can be represented as
$\Set{\bsU\ket{\Phi}: \bsU\in \U(\complex^n\ot \complex^n)}$ with
$\ket{\Phi} \in \complex^n\ot \complex^n$ any fixed pure state and
$\U(\complex^n\ot \complex^n)$ the full unitary group on
$\complex^n\ot \complex^n$, which is endowed with the standard Haar
measure. The induced measure $\dif\mu(\rho)$ on $\rD(\complex^n)$ is
derived from the above  Haar measure by taking partial trace over
$\complex^n$ of pure states on $\complex^n\ot \complex^n$.  By
spectral decomposition theorem, for generic $\rho\in
\rD(\complex^n)$, we have $\rho=\bsU
\diag(\lambda_1,\ldots,\lambda_n) \bsU^\dagger,$ where $0\leqslant
\lambda_1<\cdots<\lambda_n\leqslant \lambda_n, \sum _{j=1}^n\lambda
_j=1,$ and $\bsU\in\U(\complex^n)$. Let
$\boldsymbol{\lambda}=(\lambda_1,\lambda_2, \cdots, \lambda_n)$ and
$V(\boldsymbol{\lambda})=\prod_{1\leqslant i<j\leqslant
n}(\lambda_j-\lambda_i).$ The Haar-induced probability measure
$\dif\mu(\rho)$ on $\rD(\complex^n)$ can be factorized into  the
following product measure form
\begin{eqnarray*}
\dif\mu(\rho) = \dif\nu(\boldsymbol{\lambda})\times
\dif\mu_{\haar}(\bsU),
\end{eqnarray*}
where \cite{Zyczkowski2001jpa}, for $0<\lambda_1<\cdots<\lambda_n<
1$,
\begin{eqnarray*}
\dif\nu(\boldsymbol{\lambda}) = N_n\cdot\delta \Big
(1-\sum^n_{j=1}\lambda_j \Big ) V^2
(\lambda)\prod^n_{j=1}\dif\lambda_j
\end{eqnarray*}
is the Lebesgue measure supported on the simplex
$$
P_+ :=\Set{(\lambda_1, \lambda_2, \ldots, \lambda_n): 0<
\lambda_1<\lambda_2<\cdots <\lambda_n< 1, \sum_{j=1}^n\lambda_j=1},
$$
where
$$
N_n=\frac{\Gamma(n+1)\Gamma(n^2)}{\prod^{n-1}_{j=0}\Gamma(n-j+1)\Gamma(n-j)},
$$
and $\delta (\cdot )$ is the Dirac delta function. As usual,
$$
\Inner{\delta}{f}=\int_\real \delta(x)f(x)\dif x  =f(0),
$$
and $\delta_a(x)=\delta(x-a)$, $\Inner{\delta_a}{f} = f(a)$. Denote
the zero set of a function $g:\real\to\real$ as $Z(g)=\Set{x\in
\real:g(x)=0}.$ If $g:\real\to\real$ is a function with continuous
derivative such that $Z(g)\cap Z(g')=\emptyset$, then
\cite{lz2020ijtp}
\begin{eqnarray*}
\delta(g(x)) = \sum_{x\in Z(g)} \frac1{\abs{g'(x)}}\delta_x.
\end{eqnarray*}

\begin{prop}\label{prop2}
Let $\bsA$ be a non-degenerate positive matrix with eigenvalues
$\lambda_1(\bsA)<\cdots<\lambda_n(\bsA)$ and
$\lambda(\bsA)=(\lambda_1(\bsA),\ldots,\lambda_n(\bsA))$, then
\begin{eqnarray*}
\int_{\rD(\complex^n)} e^{-\mathrm{i}\alpha\Tr{\bsA\rho}} \dif\mu
(\rho) = \frac{C_n\cdot N_n}{(-\mathrm{i}\alpha)^{\binom{n}{2}}
V(\lambda(\bsA))}\int_{P_+}  \prod^n_{j=1} \dif\lambda_j
V(\lambda)\det ( e^{-\mathrm{i}\alpha\lambda_i(\bsA)\lambda_j} ).
\end{eqnarray*}
\end{prop}

\begin{proof}
Since $\dif\mu(\rho) = \dif\nu(\boldsymbol{\lambda})\times
\dif\mu_{\haar}(\bsU)$, it follows from Proposition~\ref{prop1} that
\begin{eqnarray*}
\int_{\rD(\complex^n)} e^{-\mathrm{i}\alpha\Tr{\bsA\rho}} \dif\mu
(\rho) &=& \int_{P_+} \dif\nu(\boldsymbol{\lambda}) \int_{\U(n)}
\dif\mu_{\haar}(\bsU) e^{-\mathrm{i}\alpha\Tr{\bsA\bsU\Lambda
\bsU^\dagger}}   \notag\\
&=&C_n\int_{P_+}  \dif\nu(\boldsymbol{\lambda})
\frac{\det\Pa{e^{-\mathrm{i}\alpha\lambda_i(\bsA)\lambda_j}}}{(-\mathrm{i}\alpha)^{\binom{n}{2}}
V(\lambda(\bsA)) V(\boldsymbol{\lambda})}\\
&=&\frac{C_n}{(-\mathrm{i}\alpha)^{\binom{n}{2}}
V(\lambda(\bsA))}\int_{P_+} \dif\nu(\boldsymbol{\lambda})
\frac{\det \Pa{e^{-\mathrm{i}\alpha\lambda_i(\bsA)\lambda_j}}}{V(\boldsymbol{\lambda})} \\
&=& \frac{C_n\cdot
N_n}{(-\mathrm{i}\alpha)^{\binom{n}{2}}V(\lambda(\bsA))}\int_{P_+}
\prod^n_{j=1} \dif\lambda_j V(\boldsymbol{\lambda})\det\Pa{
e^{-\mathrm{i}\alpha\lambda_i(\bsA)\lambda_j}},
\end{eqnarray*}
which is the desired identity.
\end{proof}

Any qubit observable $\bsA$ can be parameterized as
\begin{equation*}\label{AA}
\bsA=a_0\I+\bsa\cdot\boldsymbol{\sigma},\qquad (a_0,\bsa)\in\real^4,
\end{equation*}
where $\I$ is the identity matrix on the qubit Hilbert space
$\complex^2$, and $\boldsymbol{\sigma}=(\sigma_1,\sigma_2,\sigma_3)$
is the vector of the standard Pauli matrices
\begin{eqnarray*}
\sigma_1=\Pa{\begin{array}{cc}
               0 & 1 \\
               1 & 0
             \end{array}
},\quad \sigma_2=\Pa{\begin{array}{cc}
               0 & -\mathrm{i} \\
               \mathrm{i} & 0
             \end{array}
},\quad \sigma_3=\Pa{\begin{array}{cc}
               1 & 0 \\
               0 & -1
             \end{array}
}.
\end{eqnarray*}
Without loss of generality, we assume that our qubit observable
$\bsA$ has simple eigenvalues
\begin{eqnarray*}\label{eq:2eig}
\lambda_k(\bsA)=a_0+(-1)^k\abs{\bsa}, \quad k=1,2
\end{eqnarray*}
with $\abs{\bsa}=\sqrt{a_1^2+a_2^2+a_3^2}>0$ being the length of
vector $\bsa =(a_1,a_2,a_3)\in \real^3$.

\begin{prop}\label{prop3}
For the qubit observable  $\bsA$ as above,  we have
\begin{eqnarray*}
\int_{\rD(\complex^2)} e^{-\mathrm{i}\Tr{\bsA\rho}}\dif\mu(\rho) =
3e^{-\mathrm{i}a_0}\frac{\sin(\abs{\bsa})-\abs{\bsa}\cos(\abs{\bsa})}{\abs{\bsa}^3}.
\end{eqnarray*}
\end{prop}

\begin{proof}
From Proposition~\ref{prop2}, we have
\begin{eqnarray*}
\int_{\rD(\complex^2)}e^{-\mathrm{i}\Tr{\bsA\rho}}\dif\mu(\rho) &=&
\frac{C_2\cdot N_2}{(-\mathrm{i})V(\lambda(\bsA))}\int
\prod^2_{j=1}\dif\lambda_j\delta \Big (1-\sum^2_{j=1}\lambda_j\Big )
V(\boldsymbol{\lambda})\det (e^{-\mathrm{i}\alpha\lambda_i(\bsA)\lambda_j})\\
&=&
\frac{3\mathrm{i}}{\abs{\bsa}}\int^{\frac12}_0\dif\lambda_1(1-2\lambda_1)
\abs{\begin{array}{cc}
                                e^{-\mathrm{i}(a_0-\abs{\bsa})\lambda_1} & e^{-\mathrm{i}(a_0-\abs{\bsa})(1-\lambda_1)} \\
                                e^{-\mathrm{i}(a_0+\abs{\bsa})\lambda_1} & e^{-\mathrm{i}(a_0+\abs{\bsa})(1-\lambda_1)}
                              \end{array}
}\\
&=&3e^{-\mathrm{i}a_0}\frac{\sin(\abs{\bsa})-\abs{\bsa}\cos(\abs{\bsa})}{\abs{\bsa}^3}.
\end{eqnarray*}
This completes the proof.
\end{proof}

Now we derive the probability density functions of uncertainties of
observables. We first derive the probability density of the mean
value for an observable. Let $\bsA$ be a non-degenerate Hermitian
matrix with eigenvalues $\lambda_1(\bsA)<\cdots<\lambda_n(\bsA)$ and
$\lambda(\bsA)=(\lambda_1(\bsA),\ldots,\lambda_n(\bsA))$. The
probability density function of the mean value
$\langle\bsA\rangle_\rho =\Tr{\bsA\rho}$ is defined as
\begin{eqnarray*}
f_{\langle\bsA\rangle }(r) =\int_{\rD(\complex^n)}
\delta(r-\langle\bsA\rangle_\rho)\dif\mu(\rho).
\end{eqnarray*}
By using the integral representation of the Dirac delta function,
$\delta(r)=\frac1{2\pi}\int \dif\alpha e^{\mathrm{i}r\alpha},$ we
have
\begin{eqnarray*}
f_{\langle\bsA\rangle }(r) =\frac1{2\pi}\int _{\real}   \dif\alpha
e^{\mathrm{i}r\alpha} \int _{\rD(\complex^2)}  \dif\mu(\rho)
e^{-\mathrm{i}\alpha\Tr{\bsA\rho}}.
\end{eqnarray*}
By combining Proposition~\ref{prop1} and Proposition~\ref{prop2}, we
have
\begin{eqnarray*}\label{eq:L-transform}
f_{\langle\bsA\rangle }(r) &=& \frac{C_n\cdot N_n}{2\pi
V(\lambda(\bsA))}\int_{\real}
(-\mathrm{i}\alpha)^{-\binom{n}{2}}e^{\mathrm{i}r\alpha}I_n(\alpha)\dif\alpha,
\end{eqnarray*}
where
\begin{eqnarray*}
I_n(\alpha)=\int_{P_+} \prod^n_{j=1} \dif\lambda_j
V(\boldsymbol{\lambda})\det
(e^{-\mathrm{i}\alpha\lambda_i(\bsA)\lambda_j}).
\end{eqnarray*}
Since  the integration is over the simplex $P_+$, which can be
represented as
\begin{eqnarray*}
\begin{cases}
0<\lambda_1<\frac1n\\
\lambda_k<\lambda_{k+1}<\frac{1-(\lambda_1+\cdots+\lambda_k)}{n-k},\qquad k=1,\cdots,n-2,\\
\lambda_n=1-(\lambda_1+\cdots+\lambda_{n-1})
\end{cases}
\end{eqnarray*}
using the last identity to replace $\lambda_n$ in the integrand
$V(\boldsymbol{\lambda})\det\Pa{e^{-\mathrm{i}\alpha\lambda_i(\bsA)\lambda_j}}=\det\Pa{\sum^n_{k=1}\lambda^{i-1}_ke^{-\mathrm{i}\alpha\lambda_k\lambda_j(\bsA)}}$,
the integral $I_n(\alpha)$ is reduced to
\begin{eqnarray*}
I_n(\alpha)=\int^{\frac1n}_0\dif\lambda_1\int^{\frac{1-\lambda_1}{n-1}}_{\lambda_1}\dif\lambda_2\cdots
\int^{\frac{1-(\lambda_1+\cdots+\lambda_{n-2})}2}_{\lambda_{n-2}}\dif\lambda_{n-1}\det
\Big (
\sum^n_{k=1}\lambda^{i-1}_ke^{-\mathrm{i}\alpha\lambda_k\lambda_j(\bsA)}\Big
).
\end{eqnarray*}

\begin{prop}\label{prop4}
For a given qubit observable $\bsA$ with simple spectrum
$\lambda(\bsA)=(\lambda_1(\bsA),\lambda_2(\bsA)) $ with
$\lambda_1(\bsA)<\lambda_2(\bsA)),$ the probability distribution
density of $\langle\bsA\rangle_\rho=\Tr{\bsA\rho}$, where $\rho$ is
resulted from partially tracing over a subsystem $\complex^2$ of a
Haar-distributed random pure state on $\complex^2\ot\complex^2$, is
given by
\begin{eqnarray*}
f_{\langle\bsA\rangle}(r)=\frac{3!}{V^3
(\lambda(\bsA))}(r-\lambda_1(\bsA))(\lambda_2(\bsA)-r)\Pa{H(r-\lambda_1(\bsA))-H(r-\lambda_2(\bsA))},
\end{eqnarray*}
where $H$ is the  Heaviside function. The support of
$f_{\langle\bsA\rangle}(r)$ is the closed interval
$[\lambda_1(\bsA),\lambda_2(\bsA)]$.
\end{prop}

\begin{proof}
In particular,  when $n=2$, we have $C_2=1$ and $N_2=6$, therefore
\begin{eqnarray*}
I_2(\alpha)&=& \int^{\frac12}_0\dif\lambda_1 (1-2\lambda_1)
\abs{\begin{array}{cc}
                                e^{-\mathrm{i}\alpha\lambda_1(\bsA)\lambda_1} & e^{-\mathrm{i}\alpha\lambda_1(\bsA)(1-\lambda_1)} \\
                                e^{-\mathrm{i}\alpha\lambda_2(\bsA)\lambda_1} & e^{-\mathrm{i}\alpha\lambda_2(\bsA)(1-\lambda_1)}
                              \end{array}
}\notag\\
&=&\frac1{(\lambda_2(\bsA)-\lambda_1(\bsA))^2
\alpha^2}\sum^2_{k=1}e^{-\mathrm{i}\lambda_k(\bsA)\alpha}
(\mathrm{i}(\lambda_2(\bsA) -\lambda_1(\bsA))\alpha+(-1)^k2 ),
\end{eqnarray*}
and
\begin{eqnarray*}
f_{\langle\bsA\rangle }(r) = \frac{3!}{2\pi V^3(\lambda(\bsA))}\int
_{\real} \dif\alpha
\alpha^{-3}e^{\mathrm{i}r\alpha}\sum^2_{k=1}e^{-\mathrm{i}\lambda_k(\bsA)\alpha}
((\lambda_1(\bsA) -\lambda_2(\bsA))\alpha+(-1)^k2\mathrm{i}) .
\end{eqnarray*}
Let
$\varphi(\alpha)=\alpha^{-3}e^{\mathrm{i}r\alpha}\sum^2_{k=1}e^{-\mathrm{i}\lambda_k(\bsA)\alpha}
((\lambda_1(\bsA), -\lambda_2(\bsA))\alpha+(-1)^k2\mathrm{i})$ and
$H(\cdot )$ be the  Heaviside function, then
\begin{eqnarray*}
\int_{\real} \varphi(\alpha)\dif\alpha=\int^\infty_0
(\varphi(\alpha)+\varphi(-\alpha))\dif\alpha =
2\pi(\lambda_2(\bsA)-r)(r-\lambda_1(\bsA))\Pa{H(r-\lambda_1(\bsA))-H(r-\lambda_2(\bsA))},
\end{eqnarray*}
and we come to the result.
\end{proof}

\begin{prop}\label{prop5}
For any  qubit observable $\bsA=a_0\I+\bsa\cdot\boldsymbol{\sigma},\
(a_0,\bsa)\in\real^4,$ the probability distribution density of the
uncertainty $\Delta_\rho \bsA$, where $\rho$ is resulted from
partially tracing over a subsystem $\complex^2$ of a
Haar-distributed random pure state on $\complex^2\ot\complex^2$, is
given by
\begin{eqnarray*}
 f_{\Delta\bsA}(x) = \frac{3x^3}{2\abs{\bsa}^3\sqrt{\abs{\bsa}^2-x^2}}.
\end{eqnarray*}
\end{prop}

\begin{proof}
From $\delta(r^2-r^2_0)=\frac1{2\abs{r_0}} (
\delta(r-r_0)+\delta(r+r_0) ) $ we conclude that
\begin{eqnarray*}
\delta(x^2-(\Delta_\rho \bsA )^2) = \frac1{2x} (
\delta(x+\Delta_\rho \bsA)+\delta(x-\Delta_\rho \bsA))
=\frac1{2x}\delta(x-\Delta_\rho \bsA), \qquad x\geqslant 0.
\end{eqnarray*}
For any complex $2\times 2$ matrix $\bsA$,
\begin{eqnarray*}
\bsA^2=\Tr{\bsA}\bsA-\det(\bsA)\I ,\qquad \det(\bsA) = \frac{ ({\rm
Tr} \bsA )^2-\Tr{\bsA^2}}2,
\end{eqnarray*}
and thus
\begin{eqnarray*}
\delta\Pa{x^2-(\Delta _\rho \bsA)^2 } &=& \delta\Pa{x^2+\det(\bsA) -
\Tr{\bsA}\langle\bsA\rangle_\rho + \langle\bsA\rangle_\rho^2}.
\end{eqnarray*}
Consequently,
\begin{eqnarray*}
f_{\Delta\bsA}(x) & =& \int_{\rD(\complex^2)}\dif\mu(\rho)\delta(x-\Delta_\rho \bsA) \\
 &=&  2x\int_{\rD(\complex^2)} \dif\mu(\rho) \delta\Pa{x^2- (\Delta _\rho
\bsA)^2 }\\
&=&2x\int_{\real} \dif r\delta\Pa{\Pa{x^2+\det(\bsA)} - \Tr{\bsA}r
+r^2}\int _{\rD(\complex^2)} \dif\mu(\rho) \delta(r-\langle\bsA\rangle_\rho)\\
&=& 2x\int_{\real} \dif
rf_{\langle\bsA\rangle}(r)\delta\Pa{x^2+\det(\bsA) - \Tr{\bsA}r
+r^2} \\
 &=& 2x\int^{\lambda_2(\bsA)}_{\lambda_1(\bsA)}\dif
rf_{\langle\bsA\rangle}(r)
\delta\Pa{x^2+\lambda_1(\bsA)\lambda_2(\bsA) -
(\lambda_1(\bsA)+\lambda_2(\bsA))r +r^2}\\
&=&\frac{12x}{V^3(\lambda(\bsA))
}\int^{\lambda_2(\bsA)}_{\lambda_1(\bsA)}\dif
r(r-\lambda_1(\bsA))(\lambda_2(\bsA)-r) \delta\Pa{x^2 -
(r-\lambda_1(\bsA))(\lambda_2(\bsA)-r)},
\end{eqnarray*}
where we used Proposition~\ref{prop4} in the last equality. For
$\bsA=a_0\I+\bsa\cdot\boldsymbol{\sigma},$ we have
$V(\lambda(\bsA))=2\abs{\bsa}$. For any fixed $x$, let
$g_x(r)=x^2-(r-\lambda_1(\bsA))(\lambda_2(\bsA)-r),$ then $g'_x(r)
=\partial_rg_x(r)=2r-\lambda_1(\bsA)-\lambda_2(\bsA)$. For fixed
$x$, the equation $g_x(r)=0$ has two distinct roots
\begin{eqnarray*}
r_\pm(x)=\frac{\lambda_1(\bsA)+\lambda_2(\bsA)\pm\sqrt{V^2
(\lambda(\bsA))-4x^2}}2=a_0\pm\sqrt {\abs{\bsa}^2-x^2}
\end{eqnarray*}
in $[\lambda_1(\bsA),\lambda_2(\bsA)]$ if and only if $x\in[0,
V(\lambda(\bsA))/2)$. In this case
\begin{eqnarray*}
\delta\Pa{g_x(r)}=\frac1{\abs{g'_x(r_+(x))}}\delta_{r_+(x)}+\frac1{
\abs{g'_x(r_-(x))}}\delta_{r_-(x)},
\end{eqnarray*}
which implies  that
\begin{eqnarray*}
f_{\Delta \bsA}(x) &=&
\frac{12x}{V^3 (\lambda(\bsA))}\Pa{\frac{(r_+(x) -\lambda_1(\bsA))(\lambda_2(\bsA)- r_+(x))}{2r_+(x) -\lambda_1(\bsA)-\lambda_2(\bsA)}+\frac{(r_-(x) -\lambda_1(\bsA))(\lambda_2(\bsA)- r_-(x))}{\lambda_1(\bsA)+\lambda_2(\bsA)-2 r_-(x)}}\\
&=& \frac{24x^3}{V^3(\lambda(\bsA)) \sqrt{V(\lambda(\bsA))^2-4x^2}}=
\frac{3x^3}{2\abs{\bsa}^3\sqrt{\abs{\bsa}^2-x^2}}.
\end{eqnarray*}
This completes the proof.
\end{proof}

\begin{prop}\label{prop6}
Let $J_0(z)= \frac1\pi\int^\pi_0\cos(z\cos\theta)\dif\theta$ be the
Bessel function of first kind. Then we have the following identity:
\begin{eqnarray*}
\int^\infty_0 \frac{\sin q-q\cos q}{q^2}J_0(\lambda q)\dif
q=\begin{cases}\sqrt{1-\lambda^2},&\text{if
}\abs{\lambda} < 1;\\
0,&\text{if }\abs{\lambda} \geqslant 1.
\end{cases}
\end{eqnarray*}
\end{prop}

\begin{proof}
Denote
\begin{eqnarray*}
\Phi(\lambda) = \int^\infty_0 \frac{\sin q-q\cos q}{q^2}J_0(\lambda
q)\dif q\quad(\forall\lambda\in\real).
\end{eqnarray*}
Clearly,  $\Phi(\lambda)$ is even and
\begin{eqnarray*}
\Phi(0) = \int^\infty_0 \frac{\sin q-q\cos q}{q^2}\dif \tau =
-\frac{\sin q}{q}\Big|^\infty_0 = 1.
\end{eqnarray*}
Without loss of generality, we assume that $\lambda>0,$ then
\begin{eqnarray*}
\Phi(\lambda) &=& -\int^\infty_0\dif \Pa{\frac{\sin q}{q}}
J_0(\lambda q) \\
&=& -J_0(\lambda q)\frac{\sin q}{q}\Big|^\infty_0+ \int^\infty_0
\frac{\sin q}{q}\dif(J_0(\lambda q))\\
&=& 1 - \lambda\int^\infty_0 \frac{\sin q}{q}J_1(\lambda q)\dif q,
\end{eqnarray*}
where $J_1(z) = \frac1\pi\int^\pi_0\dif\theta
\cos\theta\sin(z\cos\theta)$. Noting that
\begin{eqnarray*}
\int^\infty_0\frac{\sin q}{q}J_1(\lambda q)\dif q =
\begin{cases}
\frac{1-\sqrt{1-\lambda ^2}}{\lambda} ,& \lambda\in(0,1)  \\
\frac1\lambda, & \lambda\in [1, +\infty)
\end{cases}
\end{eqnarray*}
we get
\begin{eqnarray*}
\Phi(\lambda)=\begin{cases}\sqrt{1-\lambda^2},&\text{if
}\abs{\lambda} < 1; \\
0,&\text{if }\abs{\lambda}\geqslant1.
\end{cases}
\end{eqnarray*}
This completes the proof.
\end{proof}

Recall that a support of a function $f$, defined on the domain
$D(f)$, is defined by
\begin{eqnarray*}
\supp(f)=\overline{\set{x\in D(f):f(x)\neq0}}.
\end{eqnarray*}
That is, the closure of the subset of $D(f)$ in which $f$ does not
vanish. From this we see that
$\supp(f_{\Delta\bsA})=[0,\abs{\bsa}]$.

Before we study the probability distribution density
\begin{eqnarray*}
f_{\Delta\bsA,\Delta\bsB}(x,y) = \int
_{\rD(\complex^2)}\delta(x-\Delta_\rho \bsA)\delta(y-\Delta_\rho
\bsB)\dif\mu (\rho)
\end{eqnarray*}
of the uncertainties $(\Delta_\rho\bsA,\Delta_\rho\bsB)$ for a pair
of qubit observables $\bsA, \bsB$, we first consider the joint
probability distribution density
\begin{eqnarray*}
f_{\langle\bsA\rangle,\langle\bsB\rangle}(r,s)
=\int_{\rD(\complex^2)}\delta(r-\langle\bsA\rangle_\rho)\delta(s-\langle\bsB\rangle_\rho)\dif\mu(\rho)
\end{eqnarray*}
of the mean values
$(\langle\bsA\rangle_\rho,\langle\bsB\rangle_\rho)$, where $\rho$ is
resulted from partially tracing over a subsystem $\complex^2$ of a
Haar-distributed random pure state on $\complex^2\ot\complex^2.$

\begin{prop}\label{prop7}
Let  $\bsA =a_0\I+\bsa\cdot\boldsymbol{\sigma}, \ \bsB
=b_0\I+\bsb\cdot\boldsymbol{\sigma} , \ (a_0,\bsa),(b_0,\bsb)\in
\real^4 $ be a pair of qubit observables. Let
\begin{eqnarray*}
\bsT_{\bsa,\bsb}=\Pa{\begin{array}{cc}
                        \Inner{\bsa}{\bsa} & \Inner{\bsa}{\bsb} \\
                        \Inner{\bsb}{\bsa} & \Inner{\bsb}{\bsb}
                      \end{array}
}.
\end{eqnarray*}

(i) If $\set{\bsa,\bsb}$ is linearly independent, then
\begin{eqnarray*}\label{eq:meanAB}
f_{\langle\bsA\rangle,\langle\bsB\rangle}(r,s)=\frac3{2\pi\sqrt{\det(\bsT_{\bsa,\bsb})}}\sqrt{1-\omega^2_{\bsA,\bsB}(r,s)}H(1-\omega_{\bsA,\bsB}(r,s)),
\end{eqnarray*}
where
$\omega_{\bsA,\bsB}(r,s)=\sqrt{(r-a_0,s-b_0)\bsT^{-1}_{\bsa,\bsb}(r-a_0,s-b_0)^\t}$.

(ii) If $\set{\bsa,\bsb}$ is linearly dependent, without loss of
generality, we assume that $\bsb=\kappa\cdot\bsa$ for some nonzero
$\kappa$, then
\begin{eqnarray*}
f_{\langle\bsA\rangle,\langle\bsB\rangle}(r,s)=
\delta((s-b_0)-\kappa(r-a_0))f_{\langle\bsA\rangle}(r),
\end{eqnarray*}
where $f_{\langle\bsA\rangle}(r)$ is from Proposition~\ref{prop4}.

\end{prop}

\begin{proof}
By using integral representation of delta function twice, we get
\begin{eqnarray*}
f_{\langle\bsA\rangle,\langle\bsB\rangle}(r,s)=\frac1{(2\pi)^2}\int_{\real^2}\dif\alpha\dif\beta
e^{\mathrm{i}(r\alpha+s\beta)}\int_{\rD(\complex^2)}\dif\mu(\rho)e^{-\mathrm{i}\Tr{(\alpha\bsA+\beta\bsB)\rho}},
\end{eqnarray*}
and by Proposition~\ref{prop3}, we have get
\begin{eqnarray*}
\int_{\rD(\complex^2)}e^{-\mathrm{i}\Tr{(\alpha\bsA+\beta\bsB)\rho}}\dif\mu(\rho)
=
3e^{-\mathrm{i}(a_0\alpha+b_0\beta)}\frac{\sin(\abs{\alpha\bsa+\beta\bsb})-\abs{\alpha\bsa+\beta\bsb}\cos(\abs{\alpha\bsa+\beta\bsb})}{\abs{\alpha\bsa+\beta\bsb}^3}.
\end{eqnarray*}
Thus,
\begin{eqnarray*}
f_{\langle\bsA\rangle,\langle\bsB\rangle}(r,s)=\frac3{(2\pi)^2}\int_{\real^2}\dif\alpha\dif\beta
e^{\mathrm{i}((r-a_0)\alpha+(s-b_0)\beta)}\frac{\sin(\abs{\alpha\bsa+\beta\bsb})-\abs{\alpha\bsa+\beta\bsb}\cos(\abs{\alpha\bsa+\beta\bsb})}{\abs{\alpha\bsa+\beta\bsb}^3}.
\end{eqnarray*}

(i) Noting that $\set{\bsa,\bsb}$ is linearly independent if and
only if $\bsT_{\bsa,\bsb}$ is invertible, it follows that
$\abs{\alpha\bsa+\beta\bsb}$ can be rewritten as $
\abs{\alpha\bsa+\beta\bsb}=\sqrt{\tilde\alpha^2+\tilde\beta^2} $
with $(\tilde\alpha , \tilde\beta)^\t=
\bsT^{\frac12}_{\bsa,\bsb}(\alpha , \beta)^\t$. Let $(\tilde
r,\tilde s)=(r,s)\bsT^{-\frac12}_{\bsa,\bsb},$ then
$\dif\tilde\alpha\dif\tilde\beta=\sqrt{\det(\bsT_{\bsa,\bsb})}\dif\alpha\dif\beta,$
or equivalently,
$\dif\alpha\dif\beta=\frac1{\sqrt{\det(\bsT_{\bsa,\bsb})}}\dif\tilde\alpha\dif\tilde\beta$.
By change of variables
$(\alpha,\beta)\to(\tilde\alpha,\tilde\beta)$, we have
\begin{eqnarray*}
f_{\langle\bsA\rangle,\langle\bsB\rangle}(r,s)=\frac3{(2\pi)^2\sqrt{\det(\bsT_{\bsa,\bsb})}}\int_{\real^2}\dif\tilde\alpha\dif\tilde\beta
e^{\mathrm{i}((\tilde r-\tilde a_0)\tilde\alpha+(\tilde s-\tilde
b_0)\tilde\beta)}\frac{\sin\Pa{\sqrt{\tilde\alpha^2+\tilde
\beta^2}}-\sqrt{\tilde\alpha^2+\tilde
\beta^2}\cos\Pa{\sqrt{\tilde\alpha^2+\tilde
\beta^2}}}{\Pa{\sqrt{\tilde\alpha^2+\tilde \beta^2}}^3}.
\end{eqnarray*}
Furthermore, using the polar coordinate  $\tilde\alpha =
q\cos\theta,\tilde\beta=q\sin\theta$,  $q\geqslant0,
\theta\in[0,2\pi],$  then  $\omega_{\bsA,\bsB}(r,s)=\sqrt{(\tilde
r-\tilde a_0)^2+(\tilde s-\tilde b_0)^2}.$ Noting the fact that
\begin{eqnarray*}
\int^{2\pi}_0e^{\mathrm{i}(u\cos\theta+v\sin\theta)}\dif\theta =2\pi
J_0(\sqrt{u^2+v^2})\quad (\forall u,v\in\real),
\end{eqnarray*}
where $J_0(z)= \frac1\pi\int^\pi_0\cos(z\cos\theta)\dif\theta$ is
the Bessel function of first kind, we have
\begin{eqnarray*}
f_{\langle\bsA\rangle,\langle\bsB\rangle}(r,s)
&=&\frac3{(2\pi)^2\sqrt{\det(\bsT_{\bsa,\bsb})}}\int^\infty_0\dif
q\frac{\sin q-q\cos q}{q^2}\int^{2\pi}_0\dif\theta
e^{\mathrm{i}q[(\tilde r-\tilde a_0)\cos\theta+(\tilde s-\tilde
b_0)\sin\theta]}\\
&=&\frac3{2\pi\sqrt{\det(\bsT_{\bsa,\bsb})}}\int^\infty_0\dif
q\frac{\sin q-q\cos q}{q^2}J_0\Pa{q\sqrt{(\tilde r-\tilde
a_0)^2+(\tilde s-\tilde b_0)^2}} \\
&=&\frac3{2\pi\sqrt{\det(\bsT_{\bsa,\bsb})}}\int^\infty_0\frac{\sin
q-q\cos q}{q^2}J_0\Pa{q\cdot\omega_{\bsA,\bsB}(r,s)}\dif q \\
&=&\frac3{2\pi\sqrt{\det(\bsT_{\bsa,\bsb})}}\sqrt{1-\omega^2_{\bsA,\bsB}(r,s)}H(1-\omega_{\bsA,\bsB}(r,s)).
\end{eqnarray*}
Here in the last equality, we used Proposition~\ref{prop6}.

(ii) If $\set{\bsa,\bsb}$ is linearly dependent, without loss of
generality, we assume that $\bsb=\kappa\cdot\bsa$ for some
$\kappa\neq0$. Performing change of variables
$(\alpha,\beta)\to(\alpha',\beta')$ where
$\alpha'=\alpha+\kappa\beta$ and $\beta'=\beta,$ the Jacobian is
given by
\begin{eqnarray*}
\det\Pa{\frac{\partial(\alpha',\beta')}{\partial(\alpha,\beta)}}=\abs{\begin{array}{cc}
                                                                        1 & \kappa \\
                                                                        0 &
                                                                        1
                                                                      \end{array}
}=1\neq0.
\end{eqnarray*}
Now
\begin{eqnarray*}
f_{\langle\bsA\rangle,\langle\bsB\rangle}(r,s)&=&
\frac3{4\pi^2}\int_{\real^2} \dif \alpha'\dif \beta'
e^{\mathrm{i}((r-a_0)(\alpha'-\kappa\beta')+(s-b_0)\beta')}\frac{\sin(\abs{\bsa}\abs{\alpha'})-\abs{\bsa}\abs{\alpha'}\cos(\abs{\bsa}\abs{\alpha'})}{\abs{\bsa}^3\abs{\alpha'}^3}\\
&=&\frac1{2\pi}\int_\real
e^{\mathrm{i}((s-b_0)-\kappa(r-a_0))\beta'}\dif\beta'\times\frac3{2\pi}\int_\real\dif\alpha'e^{\mathrm{i}(r-a_0)\alpha'}\frac{\sin(\abs{\bsa}\abs{\alpha'})-\abs{\bsa}\abs{\alpha'}\cos(\abs{\bsa}\abs{\alpha'})}{\abs{\bsa}^3\abs{\alpha'}^3}\\
&=&\delta\Pa{(s-b_0)-\kappa(r-a_0)}\times\frac3{2\pi}\int_\real\dif\alpha'e^{\mathrm{i}(r-a_0)\alpha'}\frac{\sin(\abs{\bsa}\abs{\alpha'})-\abs{\bsa}\abs{\alpha'}\cos(\abs{\bsa}\abs{\alpha'})}{(\abs{\bsa}\abs{\alpha'})^3}
\end{eqnarray*}
where
\begin{eqnarray*}
&&\frac3{2\pi}\int_\real\dif\alpha'e^{\mathrm{i}(r-a_0)\alpha'}\frac{\sin(\abs{\bsa}\abs{\alpha'})-\abs{\bsa}\abs{\alpha'}\cos(\abs{\bsa}\abs{\alpha'})}{(\abs{\bsa}\abs{\alpha'})^3}\\
&&=\frac3{2\pi}\int^\infty_0\dif\alpha' \Pa{e^{\mathrm{i}(r-a_0)\alpha'}+e^{-\mathrm{i}(r-a_0)\alpha'}}\frac{\sin(\abs{\bsa}\abs{\alpha'})-\abs{\bsa}\abs{\alpha'}\cos(\abs{\bsa}\abs{\alpha'})}{(\abs{\bsa}\abs{\alpha'})^3}\\
&&=\frac3{2\pi \abs{\bsa}}\int^\infty_0
\Pa{e^{\mathrm{i}\frac{r-a_0}{\abs{\bsa}}q}+e^{-\mathrm{i}\frac{r-a_0}{\abs{\bsa}
}q}}\frac{\sin q-q\cos q}{q^3}\dif q\\
&&= \frac6{2\pi \abs{\bsa}}\int^\infty_0\dif
q\cos\Pa{\frac{r-a_0}{\abs{\bsa}}q}\frac{\sin q-q\cos q}{q^3}
\end{eqnarray*}
Due to the fact that
\begin{eqnarray*}
\int^\infty_0\dif q\cos(pq)\frac{\sin q-q\cos
q}{q^3}=\begin{cases}\frac\pi4(1-p^2),&\text{if
}\abs{p}\leqslant1\\
0,&\text{if }\abs{p}\geqslant1
\end{cases}
\end{eqnarray*}
we obtain
\begin{eqnarray*}
\frac3{2\pi}\int _{\real}
\dif\alpha'e^{\mathrm{i}(r-a_0)\alpha'}\frac{\sin(\abs{\bsa}\abs{\alpha'})-\abs{\bsa}\abs{\alpha'}\cos(\abs{\bsa}\abs{\alpha'})}{(\abs{\bsa}\abs{\alpha'})^3}
&=&\begin{cases}\frac{3(r-(a_0-\abs{\bsa}))((a_0+\abs{\bsa})-r)}{4
\abs{\bsa}^3},&\text{if
}\frac{\abs{r-a_0}}{\abs{\bsa}}\leqslant1 \\
0,&\text{if }\frac{\abs{r-a_0}}{\abs{\bsa}}\geqslant1
\end{cases}
\end{eqnarray*}
and  use the fact that
$\abs{\bsa}=\frac{\lambda_2(\bsA)-\lambda_1(\bsA)}2$ and
$a_0=\frac12\Tr{\bsA}=\frac{\lambda_1(\bsA)+\lambda_2(\bsA)}2$, we
have
\begin{eqnarray*}
&&\frac3{2\pi}\int_{\real} \dif\alpha'e^{\mathrm{i}(r-a_0)\alpha'}\frac{\sin(\abs{\bsa}\abs{\alpha'})-\abs{\bsa}\abs{\alpha'}\cos(\abs{\bsa}\abs{\alpha'})}{(\abs{\bsa}\abs{\alpha'})^3}\\
&&=\frac{3!}{V^3
(\lambda(\bsA))}(r-\lambda_1(\bsA))(\lambda_2(\bsA)-r)(H(r-\lambda_1(\bsA))-H(r-\lambda_2(\bsA)))=f_{\langle\bsA\rangle}(r).
\end{eqnarray*}
Therefore
\begin{eqnarray*}
f_{\langle\bsA\rangle,\langle\bsB\rangle}(r,s)=
\delta((s-b_0)-\kappa(r-a_0))f_{\langle\bsA\rangle}(r),
\end{eqnarray*}
where $f_{\langle\bsA\rangle}(r)$ is from Proposition~\ref{prop4},
which is the desired result.
\end{proof}

\begin{prop}\label{prop8}
The joint probability distribution density of the uncertainties
$(\Delta_\rho \bsA,\Delta_\rho \bsB)$ for a pair of qubit
observables $\label{AB} \bsA =a_0\I+\bsa\cdot\boldsymbol{\sigma}, \
\bsB =b_0\I+\bsb\cdot\boldsymbol{\sigma} , \
(a_0,\bsa),(b_0,\bsb)\in \real^4,$ where $\set{\bsa,\bsb}$ is
linearly independent, and $\rho$ is resulted from partially tracing
a subsystem over a Haar-distributed random pure state on
$\complex^2\ot\complex^2$, is given by
\begin{eqnarray*}
f_{\Delta\bsA,\Delta\bsB}(x,y) = \frac{2xy\sum_{j\in\set{\pm}}
f_{\langle\bsA\rangle,\langle\bsB\rangle}(r_+(x)
,s_j(y))}{\sqrt{(\abs{\bsa}^2-x^2)(\abs{\bsb}^2-y^2)}},
\end{eqnarray*}
where $r_\pm(x)=a_0\pm\sqrt{\abs{\bsa}^2-x^2}, s_\pm(y)
=b_0\pm\sqrt{\abs{\bsb}^2-y^2}$, and
$f_{\langle\bsA\rangle,\langle\bsB\rangle}(\cdot,\cdot)$ is given in
Proposition~\ref{prop7}.
\end{prop}

\begin{proof}
Noting that
\begin{eqnarray*}
\delta(x^2-(\Delta _\rho\bsA )^2) =
\delta(x^2-(r-\lambda_1(\bsA))(\lambda_2(\bsA)-r))) =\delta(g_x(r)),
\end{eqnarray*}
where $g_x(r)=x^2-(r-\lambda_1(\bsA))(\lambda_2(\bsA)-r)$.
Similarly, $\delta(y^2-(\Delta _\rho\bsB)^2 ) = \delta(h_y(s)),$
where $h_y(s)=y^2-(s-\lambda_1(\bsB))(\lambda_2(\bsB)-s).$
Consequently,
\begin{eqnarray*}
f_{\Delta\bsA,\Delta\bsB}(x,y) = 4xy\int _{\rD(\complex^2)}
\dif\mu(\rho) \delta(x^2-(\Delta _\rho\bsA)^2 )\delta(y^2-(\Delta
_\rho\bsB)^2 ) = 4 xy\int _{\real^2}\dif r\dif
sf_{\langle\bsA\rangle,\langle\bsB\rangle}(r,s)\delta(g_x(r))\delta(h_y(s)),
\end{eqnarray*}
where $f_{\langle\bsA\rangle,\langle\bsB\rangle}(r,s)$ is determined
by Proposition~\ref{prop7}. Noting that
\begin{eqnarray*}
\delta\Pa{g_x(r)}=\frac1{\abs{g'_x(r_+(x))}}\delta_{r_+(x)}+\frac1{\abs{g'_x(
r_-(x))}}\delta_{r_-(x)},\quad
\delta\Pa{h_y(s)}=\frac1{\abs{h'_y(s_+(y))}}\delta_{s_+(y)}+\frac1{\abs{h'_y(
s_-(y))}}\delta_{s_-(y)},
\end{eqnarray*}
we obtain
\begin{eqnarray*}
\delta(g_x(r))\delta(h_y(s)) =
\frac{\delta_{(r_+,s_+)}+\delta_{(r_+,s_-)}+\delta_{(r_-,s_+)}+\delta_{(r_-,s_-)}}{4\sqrt{(\abs{\bsa}^2-x^2)(\abs{\bsb}^2-y^2)}}.
\end{eqnarray*}
Based on this observation, we get
\begin{eqnarray*}
f_{\Delta\bsA,\Delta\bsB}(x,y) =
\frac{xy}{\sqrt{(\abs{\bsa}^2-x^2)(\abs{\bsb}^2-y^2)}}\sum_{i,j\in\set{\pm}}f_{\langle\bsA\rangle,\langle\bsB\rangle}(r_i,s_j).
\end{eqnarray*}
It is easily checked that $\omega_{\bsA,\bsB}(r_+ , s_+
)=\omega_{\bsA,\bsB}(r_- ,s_-), \ \omega_{\bsA,\bsB}(r_+ , s_-
)=\omega_{\bsA,\bsB}(r_- , s_+)$, therefore
\begin{eqnarray*}
\sum_{i,j\in\set{\pm}}f_{\langle\bsA\rangle,\langle\bsB\rangle}(r_i,s_j)=2\sum_{j\in\set{\pm}}f_{\langle\bsA\rangle,\langle\bsB\rangle}(r_+
,s_j),
\end{eqnarray*}
and
\begin{eqnarray*}
f_{\Delta\bsA,\Delta\bsB}(x,y) =
\frac{2xy\sum_{j\in\set{\pm}}f_{\langle\bsA\rangle,\langle\bsB\rangle}(r_+
,s_j)}{\sqrt{(\abs{\bsa}^2-x^2)(\abs{\bsb}^2-y^2)}},
\end{eqnarray*}
which is  the desired result.
\end{proof}

\begin{proof}[Proof of Lemma~\ref{lem1}]
With Proposition~\ref{prop8}, we now make an analysis of the support
of $f_{\Delta\bsA,\Delta\bsB}$. In fact, due to the relation between
$f_{\Delta\bsA,\Delta\bsB}$ and
$f_{\langle\bsA\rangle,\langle\bsB\rangle}$, the support of
$f_{\Delta\bsA,\Delta\bsB}$ can be identified by the support of
$f_{\langle\bsA\rangle,\langle\bsB\rangle}$ which can be seen from
Proposition~\ref{prop7} (i),
\begin{eqnarray*}
\supp(f_{\langle\bsA\rangle,\langle\bsB\rangle}) =
\Set{(r,s)\in\real^2:\omega_{\bsA,\bsB}(r,s)\leqslant1}.
\end{eqnarray*}
Note that $f_{\Delta\bsA,\Delta\bsB}(x,y)$ is defined on the first
quadrant $\real^2_{\geqslant0}$, if $xy>0$, then
$f_{\Delta\bsA,\Delta\bsB}(x,y)=0$ if and only if
$$
\sum_{i,j\in\set{\pm}}f_{\langle\bsA\rangle,\langle\bsB\rangle}(r_i(x),s_j(y))=0,
$$
i.e., $f_{\langle\bsA\rangle,\langle\bsB\rangle}(r_i(x),s_j(y))=0$
because $f_{\langle\bsA\rangle,\langle\bsB\rangle}$ is a
non-negative function. This means that all four points
$(r_\pm(x),s_\pm(y))$ are not in the support of
$f_{\langle\bsA\rangle,\langle\bsB\rangle}$. Therefore the
uncertainty region (i.e.,the support of $f_{\Delta\bsA,\Delta\bsB}$)
of $\bsA$ and $\bsB$ is given by the following set:
\begin{eqnarray*}
\cU_{\Delta\bsA,\Delta\bsB}=\supp(f_{\Delta\bsA,\Delta\bsB})=D^+_{\bsa,\bsb}\cup
D^-_{\bsa,\bsb},
\end{eqnarray*}
where, via $\bsu_\epsilon(x,y)=\Pa{\sqrt{\abs{\bsa}^2-x^2},
\epsilon\sqrt{\abs{\bsb}^2-y^2}}$,
\begin{eqnarray*}
D^{\epsilon}_{\bsa,\bsb}:=\Set{(x,y)\in\real^2_{\geqslant0}\cap([0,\abs{\bsa}]\times[0,\abs{\bsb}]):
\bsu_\epsilon(x,y)\bsT_{\bsa,\bsb}^{-1}\bsu^\t_\epsilon(x,y)\leqslant1},\quad\epsilon\in\set{\pm}.
\end{eqnarray*}
This is what we want. Furthermore, we have {\small\begin{eqnarray*}
D^{\epsilon}_{\bsa,\bsb}=\Set{(x,y)\in\real^2_{\geqslant0}\cap([0,\abs{\bsa}]\times[0,\abs{\bsb}]):
\abs{\bsb}^2x^2+\abs{\bsa}^2y^2+2\epsilon\Inner{\bsa}{\bsb}\sqrt{(\abs{\bsa}^2-x^2)(\abs{\bsb}^2-y^2)}\geqslant
\abs{\bsa}^2\abs{\bsb}^2+\Inner{\bsa}{\bsb}^2}.
\end{eqnarray*}}
Therefore the uncertainty region
$\cU_{\Delta\bsA,\Delta\bsB}=\Set{(\Delta_\rho \bsA,\Delta_\rho
\bsB)\in\real^2_{\geqslant0}: \rho\in\density{\complex^2}}$ of
$\bsA$ and $\bsB$ is determined by the following inequality:
\begin{eqnarray*}
\abs{\bsb}^2x^2+\abs{\bsa}^2y^2+2\abs{\Inner{\bsa}{\bsb}}
\sqrt{(\abs{\bsa}^2-x^2)(\abs{\bsb}^2-y^2)}\geqslant
\abs{\bsa}^2\abs{\bsb}^2+\Inner{\bsa}{\bsb}^2,
\end{eqnarray*}
where $x\in[0,\abs{\bsa}]$ and $y\in[0,\abs{\bsb}]$.
\end{proof}

\subsection{Proof of formula \eqref{vol2}}\label{app-vol2}

Since the $\cU(\theta)$ is defined by
\begin{eqnarray*}
\cU(\theta):=\Set{(x,y)\in[0,1]^2:x^2+y^2+2\abs{\cos\theta}\sqrt{(1-x^2)(1-y^2)}\geqslant1+\cos^2\theta}.
\end{eqnarray*}
The volume (i.e., the area for 2D domain) of the uncertainty region
$\cU(\theta)(\theta\in[0,\pi/2])$ is calculated as follows
\begin{eqnarray*}
\vol(\cU(\theta))=\frac12(\pi-3\theta) \sin\theta-\cos\theta+1.
\end{eqnarray*}
Indeed, if $\theta\in[0,\pi/4]$, then $\cU(\theta)$ becomes
\begin{eqnarray*}
0\leqslant &x&\leqslant\sin\theta, -x
\cos\theta+\sin\theta\sqrt{1-x^2}\leqslant y\leqslant x
\cos\theta+\sin\theta\sqrt{1-x^2};\\
\sin\theta\leqslant &x&\leqslant \cos\theta, x
\cos\theta-\sin\theta\sqrt{1-x^2}\leqslant y\leqslant x
\cos\theta+\sin\theta\sqrt{1-x^2};\\
\cos\theta\leqslant &x&\leqslant1,
x\cos\theta-\sin\theta\sqrt{1-x^2}\leqslant y\leqslant1.
\end{eqnarray*}
Hence
\begin{eqnarray*}
\vol(\cU(\theta))&=&\int^{\sin\theta}_0\dif x\int^{\sqrt{1-x^2}\sin
\theta+x \cos \theta}_{\sqrt{1-x^2}\sin \theta-x \cos\theta}\dif y
+\int^{\cos\theta}_{\sin \theta}\int^{x \cos\theta+\sqrt{1-x^2}
\sin\theta}_{x \cos\theta-\sqrt{1-x^2} \sin\theta}\dif
y+\int^1_{\cos\theta}\dif x\int_{x\cos\theta-\sqrt{1-x^2} \sin
\theta}^1\dif y\\
&=&\frac12(\pi-3\theta) \sin\theta-\cos\theta+1.
\end{eqnarray*}
If $\theta\in[\pi/4,\pi/2]$, then $\cU(\theta)$ becomes
\begin{eqnarray*}
0\leqslant &x&\leqslant\cos\theta, -x
\cos\theta+\sin\theta\sqrt{1-x^2}\leqslant y\leqslant x
\cos\theta+\sin\theta\sqrt{1-x^2};\\
\cos\theta\leqslant &x&\leqslant \sin\theta, -x
\cos\theta+\sin\theta\sqrt{1-x^2}\leqslant y\leqslant 1;\\
\sin\theta\leqslant &x&\leqslant1,
x\cos\theta-\sin\theta\sqrt{1-x^2}\leqslant y\leqslant1,
\end{eqnarray*}
implying that
\begin{eqnarray*}
\vol(\cU(\theta))&=&\int^{\cos\theta}_0\dif x\int^{\sqrt{1-x^2}\sin
\theta+x \cos\theta}_{\sqrt{1-x^2}\sin \theta-x \cos\theta}\dif y
+\int^{\sin\theta}_{\cos\theta}\dif x\int^1_{\sqrt{1-x^2} \sin
\theta-x\cos\theta}\dif y+\int^1_{\sin\theta}\dif x\int^1_{x
\cos\theta-\sqrt{1-x^2}
\sin\theta}\dif y\\
&=&\frac12(\pi-3\theta) \sin\theta-\cos\theta+1.
\end{eqnarray*}
In summary, we get the desired result. We remark here that for
general lengths $\abs{\bsa}$ and $\abs{\bsb}$, we immediately get
that
$\vol(\cU_{\Delta\bsA,\Delta\bsB})=\abs{\bsa}\abs{\bsb}\vol(\cU(\theta))$.

As some representatives, in Fig.~\ref{fig:2DUR&PDF}, we plot the
PDFs (probability density functions) over the their respective
uncertainty regions $\cU(\theta)$ for
$\theta\in\set{\frac\pi8,\frac\pi4,\frac{3\pi}8}$, and the
probability density functions on $\cU(\theta)$.
\begin{figure}[ht]
\subfigure[The uncertainty region $\cU(\pi/8)$]
{\begin{minipage}[b]{.49\linewidth}
\includegraphics[width=.75\textwidth]{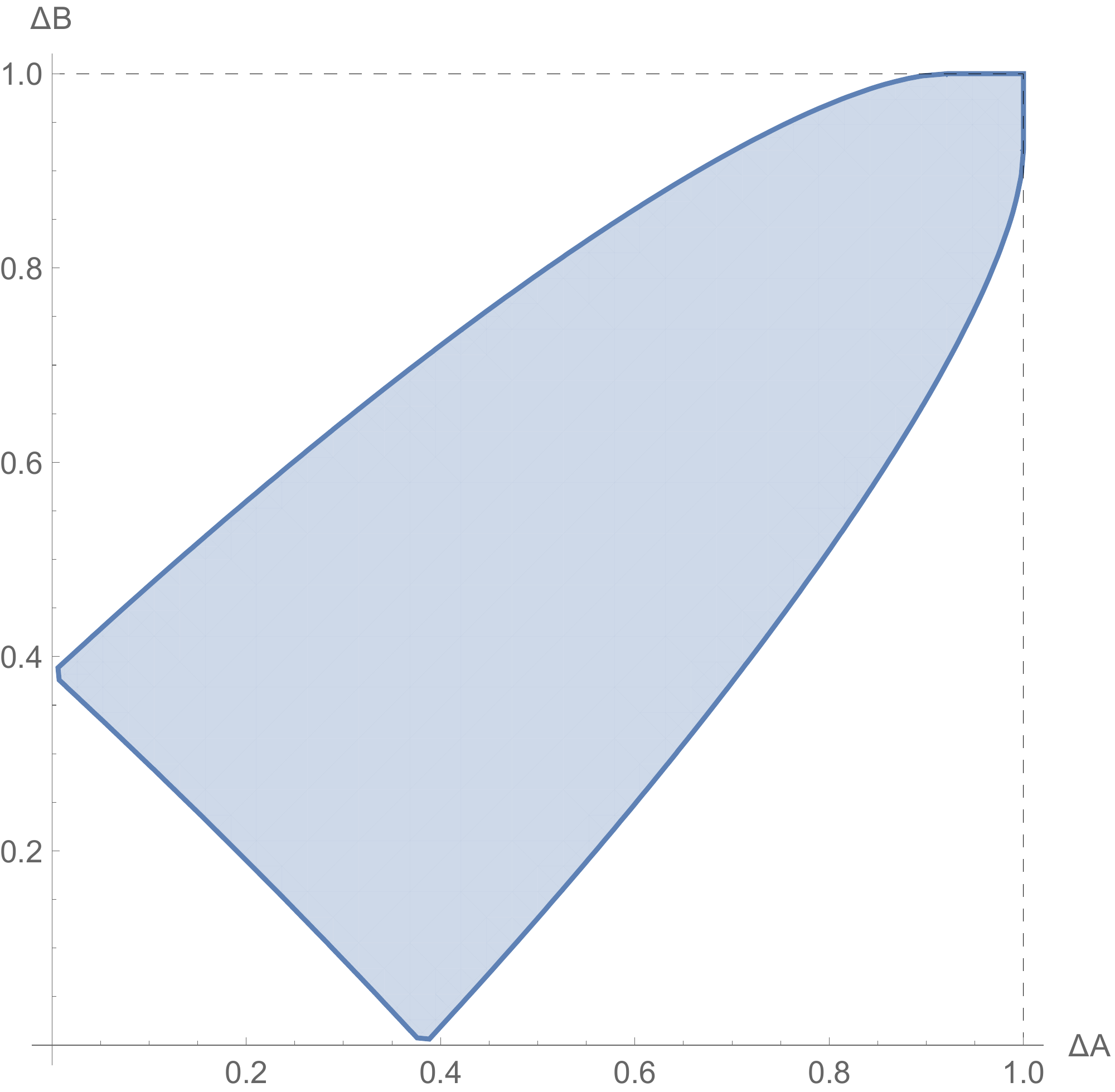}
\end{minipage}}
\subfigure[The pdf on $\cU(\pi/8)$]
{\begin{minipage}[b]{.49\linewidth}
\includegraphics[width=1\textwidth]{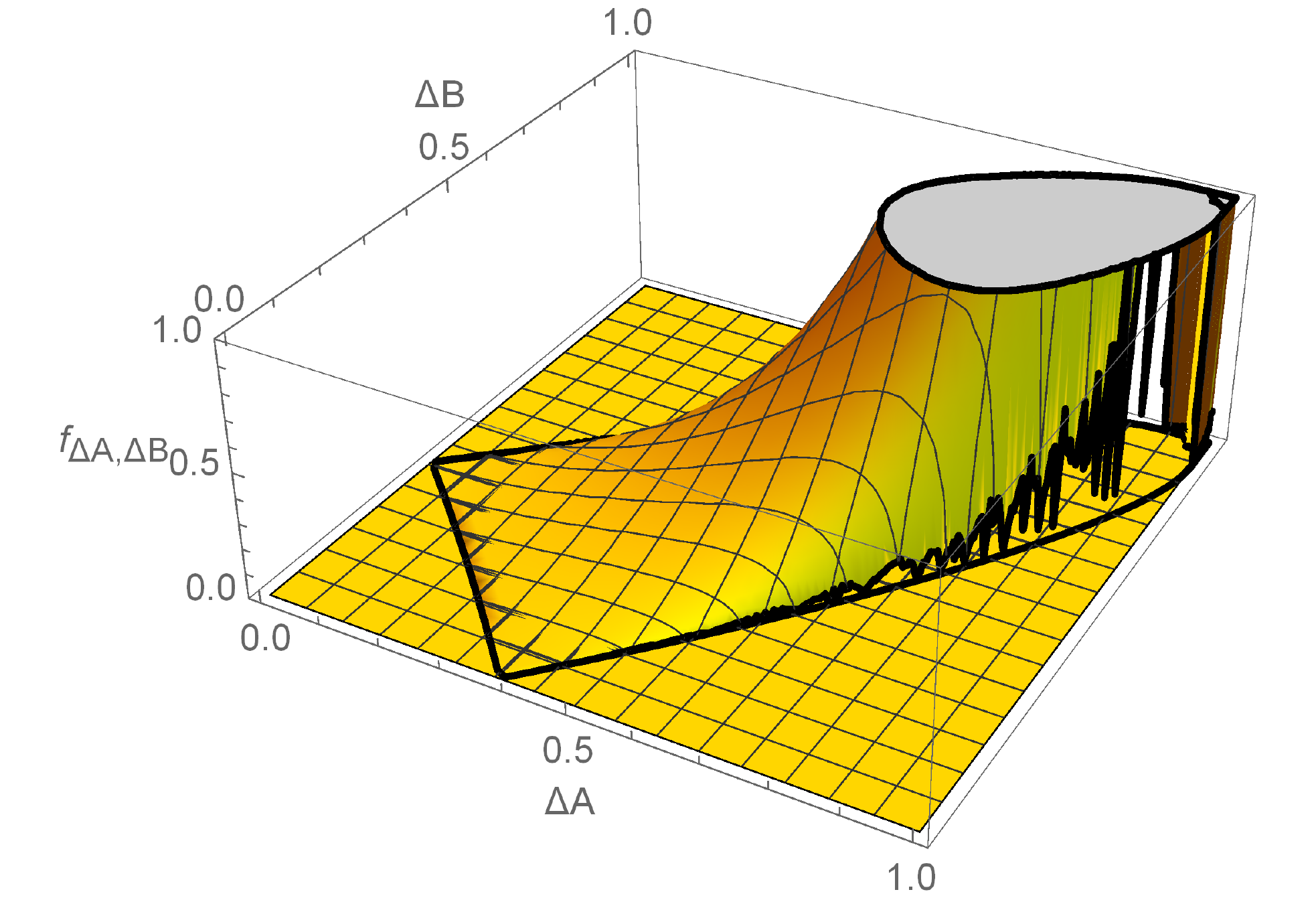}
\end{minipage}}
\subfigure[The uncertainty region $\cU(\pi/4)$]
{\begin{minipage}[b]{.49\linewidth}
\includegraphics[width=.75\textwidth]{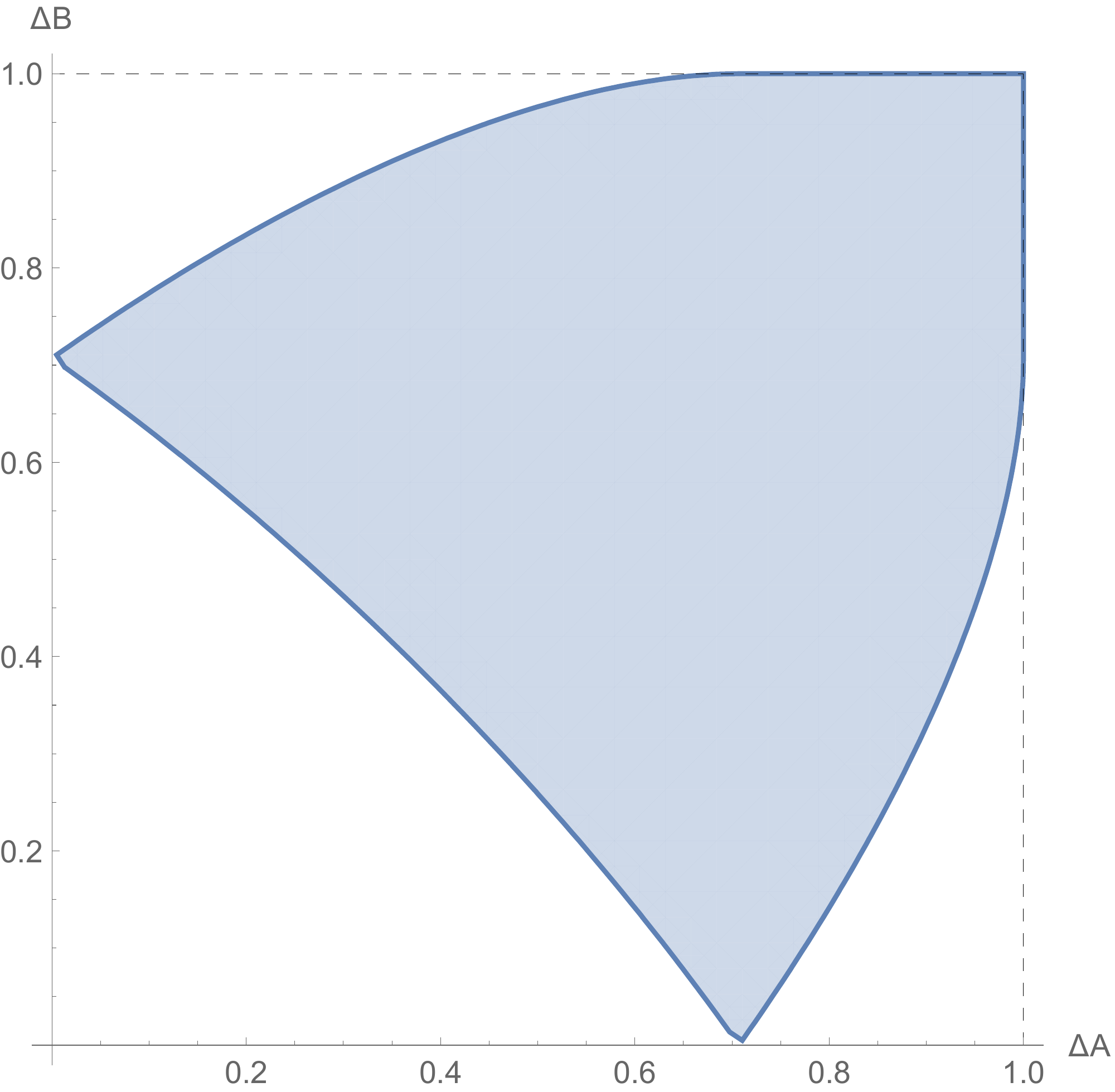}
\end{minipage}}
\subfigure[The pdf on $\cU(\pi/4)$]
{\begin{minipage}[b]{.49\linewidth}
\includegraphics[width=1\textwidth]{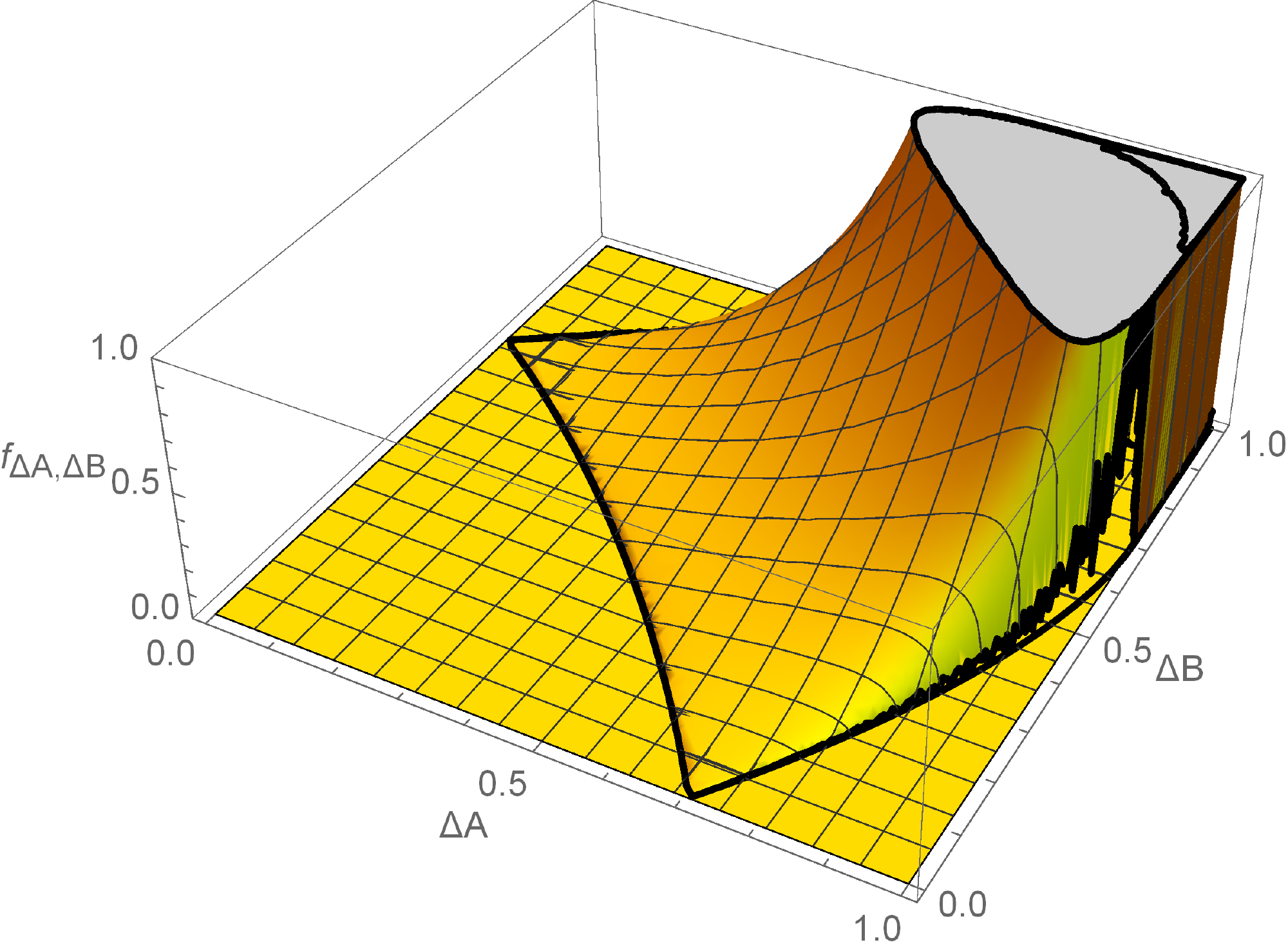}
\end{minipage}}
\subfigure[The uncertainty region $\cU(\pi/2)$]
{\begin{minipage}[b]{.49\linewidth}
\includegraphics[width=.75\textwidth]{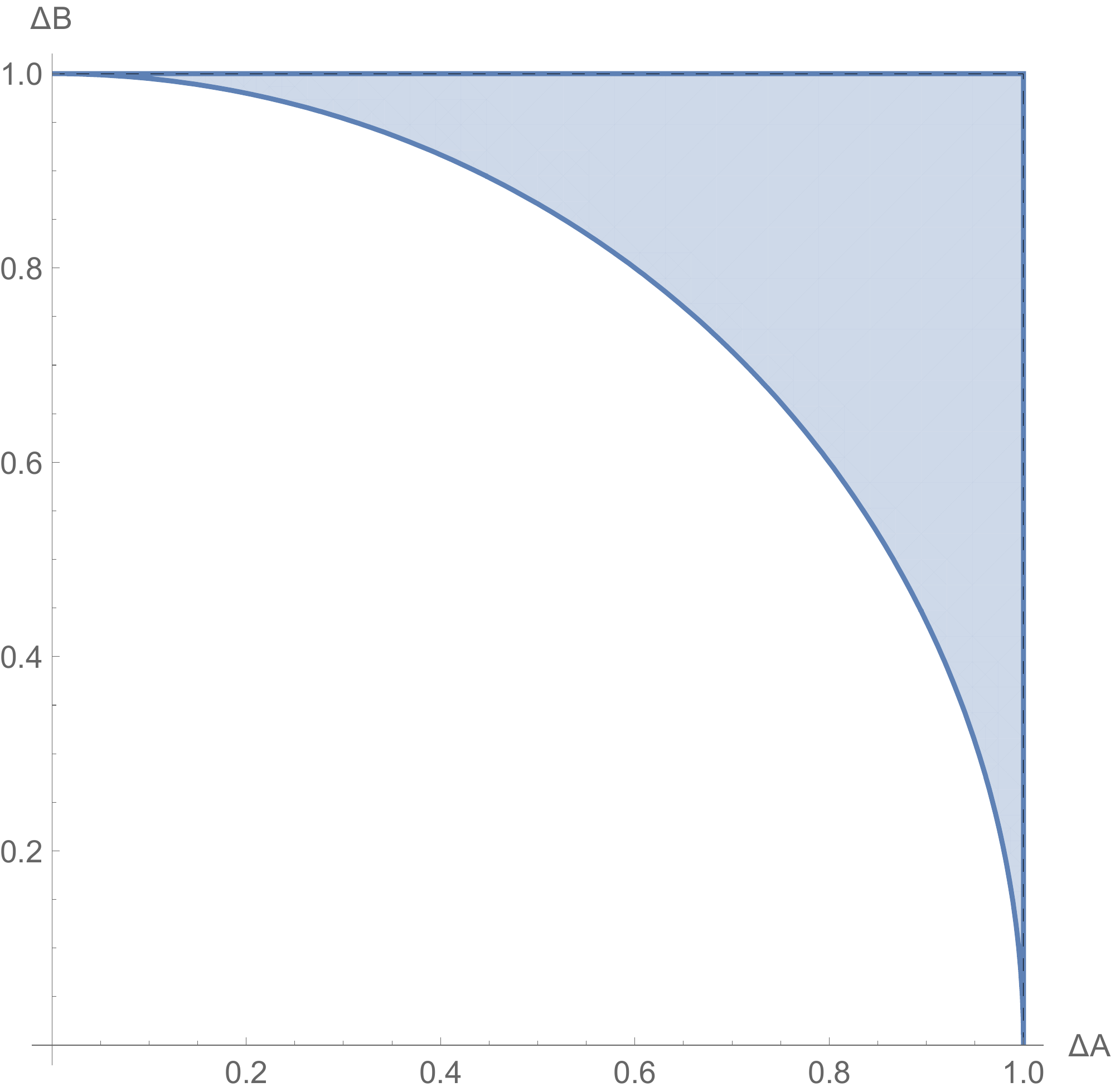}
\end{minipage}}
\subfigure[The pdf on $\cU(\pi/2)$]
{\begin{minipage}[b]{.49\linewidth}
\includegraphics[width=1\textwidth]{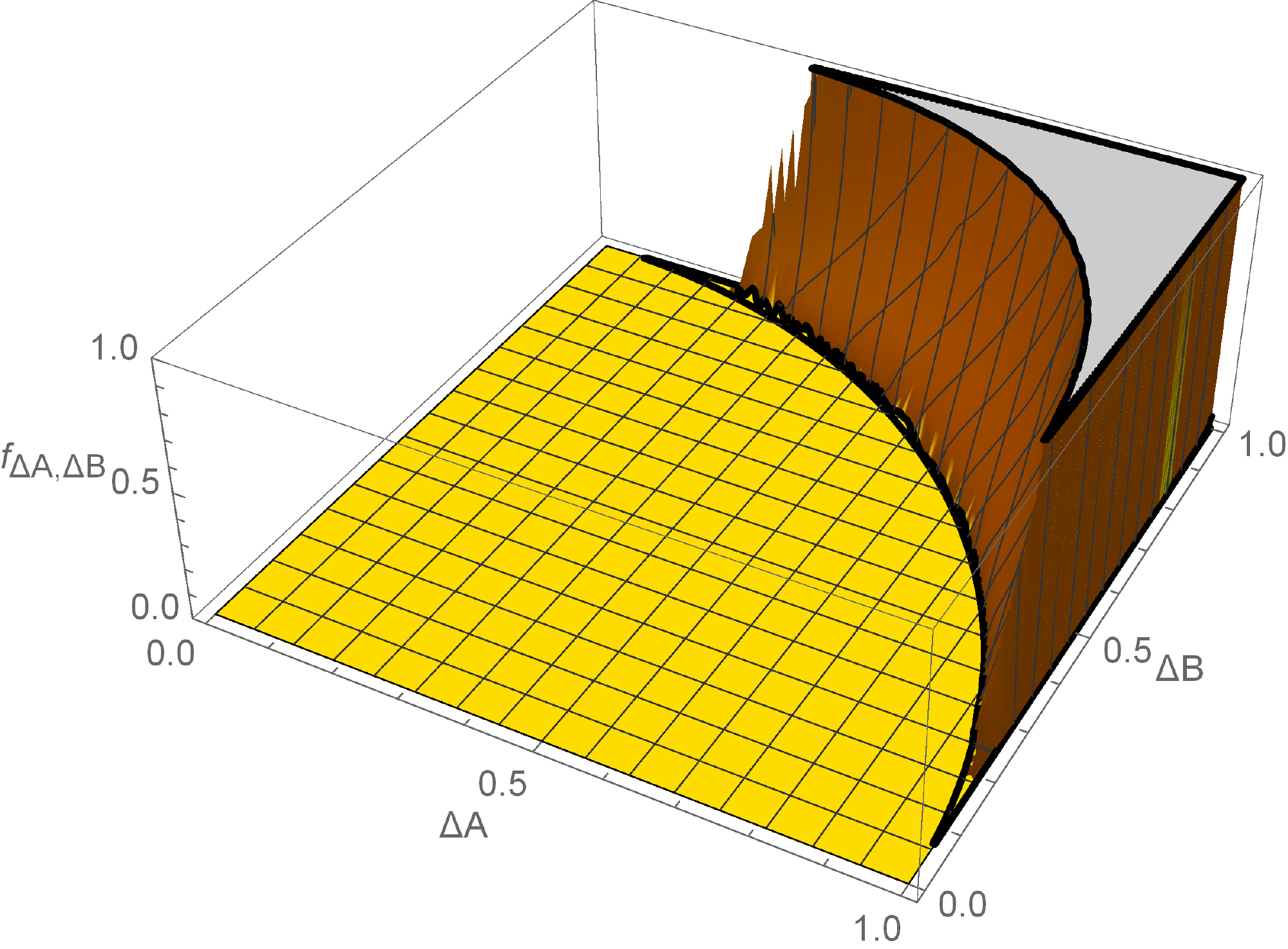}
\end{minipage}}
\caption{Plots of the uncertainty region $\cU(\theta)$ and the pdf
on it for a pair of qubit observables
$\bsA=a_0\I+\bsa\cdot\boldsymbol{\sigma}$ and
$\bsB=b_0\I+\bsb\cdot\boldsymbol{\sigma}$, where
$\abs{\bsa}=\abs{\bsb}=1$ and
$\Inner{\bsa}{\bsb}=\cos(\theta)$.}\label{fig:2DUR&PDF}
\end{figure}

\subsection{Proof of Theorem~\ref{th1}}\label{app-theorem1}

Based on the proof of Lemma~\ref{lem1}, we now present the proof of
Theorem~\ref{th1}.

\begin{proof}[Proof of Theorem~\ref{th1}]
Let $\sqrt{\abs{\bsa}^2-x^2}=X,\sqrt{\abs{\bsb}^2-y^2}=Y$ in
Lemma~\ref{lem1}, where $X\in[0,\abs{\bsa}],Y\in[0,\abs{\bsb}]$ due
to the fact that $x\in[0,\abs{\bsa}],y\in[0,\abs{\bsb}]$. Thus we
get that
\begin{eqnarray*}
\min
\Set{x^2+y^2:(x,y)\in\cU_{\Delta\bsA,\Delta\bsB}}=\Tr{\bsT_{\bsa,\bsb}}-\max\Set{
X^2+Y^2: (X,\pm Y)\bsT^{-1}_{\bsa,\bsb}(X,\pm Y)^\t\leqslant1}.
\end{eqnarray*}
It is easily seen that the objection function $x^2+y^2$, where
$(x,y)\in\cU_{\Delta\bsA,\Delta\bsB}$, attains its minimal value on
the boundary curve $\partial\cU_{\Delta\bsA,\Delta\bsB}$ of the
uncertainty region $\cU_{\Delta\bsA,\Delta\bsB}$; this also
corresponds to the objection function $X^2+Y^2$ attains its maximal
value on the boundary curve $(X,\pm Y)\bsT^{-1}_{\bsa,\bsb}(X,\pm
Y)^\t=1$.

Denote by $\lambda_k(\bsT_{\bsa,\bsb})(k=1,2)$ the two eigenvalues
of the matrix $\bsT_{\bsa,\bsb}$. By Spectral Decomposition Theorem,
we get that there exists orthogonal $\bsO\in\O(2)$ such that
\begin{eqnarray*}
\bsT_{\bsa,\bsb} =
\bsO\diag(\lambda_1(\bsT_{\bsa,\bsb}),\lambda_2(\bsT_{\bsa,\bsb}))\bsO^\t.
\end{eqnarray*}
Now let $(X',Y')^\t=\bsO(X,\pm Y)^\t$. Then
\begin{eqnarray*}
(X,\pm Y)\bsT^{-1}_{\bsa,\bsb}(X,\pm Y)^\t
=(X',Y')\diag(\lambda^{-1}_1(\bsT_{\bsa,\bsb}),\lambda^{-1}_2(\bsT_{\bsa,\bsb}))(X',Y')^\t=
\frac{{X'}^2}{\lambda_1(\bsT_{\bsa,\bsb})}+\frac{{Y'}^2}{\lambda_2(\bsT_{\bsa,\bsb})}.
\end{eqnarray*}
Because these rotations do not change the length of vectors, we get
that
\begin{eqnarray*}
\max \set{X^2+Y^2:(X,\pm Y)\bsT^{-1}_{\bsa,\bsb}(X,\pm Y)^\t=1}
&=&\max \Set{X^2+(\pm Y)^2:(X,\pm Y)\bsT^{-1}_{\bsa,\bsb}(X,\pm
Y)^\t=1}\\
&=&\max
\Set{{X'}^2+{Y'}^2:\frac{{X'}^2}{\lambda_1(\bsT_{\bsa,\bsb})}+\frac{{Y'}^2}{\lambda_2(\bsT_{\bsa,\bsb})}=1}.
\end{eqnarray*}
Again, the above optimization problem becomes
\begin{eqnarray*}
\min
\Set{x^2+y^2:(x,y)\in\partial\cU_{\Delta\bsA,\Delta\bsB}}=\Tr{\bsT_{\bsa,\bsb}}-\max\Set{{X'}^2+{Y'}^2:\frac{{X'}^2}{\lambda_1(\bsT_{\bsa,\bsb})}+\frac{{Y'}^2}{\lambda_2(\bsT_{\bsa,\bsb})}=1}.
\end{eqnarray*}
Therefore
\begin{eqnarray*}
\min
\Set{x^2+y^2:(x,y)\in\cU_{\Delta\bsA,\Delta\bsB}}=\Tr{\bsT_{\bsa,\bsb}}-\lambda_{\max}(\bsT_{\bsa,\bsb})
\end{eqnarray*}
implying that
\begin{eqnarray*}
(\Delta_\rho\bsA)^2+(\Delta_\rho\bsB)^2\geqslant\lambda_{\min}(\bsT_{\bsa,\bsb})
=
\frac12\Pa{\abs{\bsa}^2+\abs{\bsb}^2-\sqrt{(\abs{\bsa}^2-\abs{\bsb}^2)^2+4\Inner{\bsa}{\bsb}^2}}.
\end{eqnarray*}
Next we show the second inequality concerning standard deviations.
We recall that the boundary curve of uncertainty region
$\cU_{\Delta\bsA,\Delta\bsB}$ is given by
\begin{eqnarray*}
\abs{\bsb}^2x^2+\abs{\bsa}^2y^2+2\abs{\Inner{\bsa}{\bsb}}
\sqrt{(\abs{\bsa}^2-x^2)(\abs{\bsb}^2-y^2)}=
\abs{\bsa}^2\abs{\bsb}^2+\Inner{\bsa}{\bsb}^2.
\end{eqnarray*}
Let $\theta$ be the angle between $\bsa$ and $\bsb$. It is easily
seen that such boundary curve intersects two points, respectively,
with two axis, whose coordinates are $(\abs{\bsa}\sin\theta,0)$ and
$(0,\abs{\bsb}\sin\theta)$. It suffices to consider the part of such
boundary curve over the interval $[0,\abs{\bsa}\sin\theta]$, whose
equation is, indeed, given by
\begin{eqnarray*}
y = \frac{\abs{\bsb}}{\abs{\bsa}}\Pa{-\cos\theta x
+\sin\theta\sqrt{\abs{\bsa}^2-x^2}}.
\end{eqnarray*}
Then the optimization problem
$\min\Set{x+y:(x,y)\in\cU_{\Delta\bsA,\Delta\bsB}}$ is equivalent to
that choosing minimal positive real number $R$ such that the
straight line, described by the equation $x+y=R$, and the part of
the boundary curve over the closed interval
$[0,\abs{\bsa}\sin\theta]$, intersects only one point. Clearly $y =
\frac{\abs{\bsb}}{\abs{\bsa}}\Pa{-\cos\theta x
+\sin\theta\sqrt{\abs{\bsa}^2-x^2}}$, where
$x\in[0,\abs{\bsa}\sin\theta]$, is a concave function, which means
the part of the boundary curve on $[0,\abs{\bsa}\sin\theta]$ should
be above the straight line $x+y=R$. This situation appears if and
only if the straight line $x+y=R$ gets through one point whose
horizontal or vertical coordinate is
$\min(\abs{\bsa}\sin\theta,\abs{\bsb}\sin\theta)=\min(\abs{\bsa},\abs{\bsb})\sin\theta$.
Therefore $R=\min(\abs{\bsa},\abs{\bsb})\sin\theta$, which can be
also rewritten as
\begin{eqnarray*}
R=\min(\abs{\bsa},\abs{\bsb})\sin\theta
=\frac{\min(\abs{\bsa},\abs{\bsb})\max(\abs{\bsa},\abs{\bsb})}{\max(\abs{\bsa},\abs{\bsb})}\sin\theta=
\frac{\abs{\bsa\times\bsb}}{\max(\abs{\bsa},\abs{\bsb})}.
\end{eqnarray*}
Finally we obtain that
\begin{eqnarray*}
\Delta_\rho\bsA+\Delta_\rho\bsB
\geqslant\min\Set{x+y:(x,y)\in\cU_{\Delta\bsA,\Delta\bsB}}=\frac{\abs{\bsa\times\bsb}}{\max(\abs{\bsa},\abs{\bsb})},
\end{eqnarray*}
This completes the proof.
\end{proof}

\subsection{Proof of Lemma~\ref{lem2}}\label{app-lemma2}

The proof of Lemma~\ref{lem2} will be also recognized as
Propositions~\ref{prop9} and \ref{prop10}. In order to present the
proof of Lemma~\ref{lem2}, we next derive the joint probability
distribution density
\begin{eqnarray*}
f_{\Delta\bsA,\Delta\bsB,\Delta\bsC}(x,y,z) =
\int_{\rD(\complex^2)}\dif\mu (\rho)\delta(x-\Delta_\rho
\bsA)\delta(y-\Delta_\rho \bsB)\delta(z-\Delta_\rho \bsC)
\end{eqnarray*}
of the uncertainties
$(\Delta_\rho\bsA,\Delta_\rho\bsB,\Delta_\rho\bsC)$ of the three
qubit observables $\bsA, \bsB, \bsC$. For this purpose, we first
derive the joint probability distribution density
\begin{eqnarray*}
f_{\langle\bsA\rangle,\langle\bsB\rangle,\langle\bsC\rangle}(r,s,t)
=\int_{\rD(\complex^2)}
\dif\mu(\rho)\delta(r-\langle\bsA\rangle_\rho)\delta(s-\langle\bsB\rangle_\rho)\delta(t-\langle\bsC\rangle_\rho)
\end{eqnarray*}
of the mean value
$(\langle\bsA\rangle_\rho,\langle\bsB\rangle_\rho,\langle\bsC\rangle_\rho)$
of $\bsA, \bsB, \bsC, $ where $\rho$ is resulted from partially
tracing a subsystem over a Haar-distributed random pure state on
$\complex^2\ot\complex^2.$

\begin{prop}\label{prop9}
For three qubit observables $ \bsA
=a_0\I+\bsa\cdot\boldsymbol{\sigma}, \ \bsB
=b_0\I+\bsb\cdot\boldsymbol{\sigma}, \ \bsC
=c_0\I+\bsc\cdot\boldsymbol{\sigma} \
(a_0,\bsa),(b_0,\bsb),(c_0,\bsc)\in \real^4, $ let \begin{eqnarray*}
\bsT_{\bsa,\bsb,\bsc}=\Pa{\begin{array}{ccc}
                        \Inner{\bsa}{\bsa} & \Inner{\bsa}{\bsb} & \Inner{\bsa}{\bsc} \\
                        \Inner{\bsb}{\bsa} & \Inner{\bsb}{\bsb} & \Inner{\bsb}{\bsc} \\
                        \Inner{\bsc}{\bsa} & \Inner{\bsc}{\bsb} & \Inner{\bsc}{\bsc}
                      \end{array}
}
\end{eqnarray*}
(i) If $\rank(\bsT_{\bsa,\bsb,\bsc})=3$, i.e.,
$\set{\bsa,\bsb,\bsc}$ is linearly independent, then
\begin{eqnarray*}\label{eq:meanABC}
f_{\langle \bsA\rangle ,\langle \bsB\rangle , \langle \bsC\rangle
}(r,s,t)
=\frac{3(1+\sign(1-\omega_{\bsA,\bsB,\bsC}(r,s,t)))}{8\pi\sqrt{\det(\bsT_{\bsa,\bsb,\bsc})}}=
\begin{cases}
\frac3{4\pi\sqrt{\det(\bsT_{\bsa,\bsb,\bsc})}},&\text{if }\omega_{\bsA,\bsB,\bsC}(r,s,t)<1 \\
\frac3{8\pi\sqrt{\det(\bsT_{\bsa,\bsb,\bsc})}},&\text{if }\omega_{\bsA,\bsB,\bsC}(r,s,t)=1 \\
0,&\text{if }\omega_{\bsA,\bsB,\bsC}(r,s,t)>1
\end{cases}
\end{eqnarray*}
where
$\omega_{\bsA,\bsB,\bsC}(r,s,t)=\sqrt{(r-a_0,s-b_0,t-c_0)\bsT^{-1}_{\bsa,\bsb,\bsc}(r-a_0,s-b_0,t-c_0)^\t}$.

(ii) If $\rank(\bsT_{\bsa,\bsb,\bsc})=2,$ without loss of
generality, we assume that $\set{\bsa,\bsb}$ are linearly
independent and $\bsc=\kappa_{\bsa}\cdot\bsa+\kappa_{\bsb}\cdot\bsb$
for some $\kappa_{\bsa}$ and $\kappa_{\bsb}$ with
$\kappa_{\bsa}\kappa_{\bsb}\neq0$, then
\begin{eqnarray*}
f_{\langle\bsA\rangle,\langle\bsB\rangle,\langle\bsC\rangle}(r,s,t)=
\delta((t-c_0)-\kappa_{\bsa}(r-a_0)-\kappa_{\bsb}(s-b_0))f_{\langle\bsA\rangle,\langle\bsB\rangle}(r,s).
\end{eqnarray*}
where $f_{\langle\bsA\rangle,\langle\bsB\rangle}(r,s)$ is from
Proposition~\ref{prop7}.

(iii) If $\rank(\bsT_{\bsa,\bsb,\bsc})=1$, without loss of
generality, we may assume that
$\bsb=\kappa_{\bsb\bsa}\cdot\bsa,\bsc=\kappa_{\bsc\bsa}\cdot\bsa$
for some $\kappa_{\bsb\bsa}$ and $\kappa_{\bsc\bsa}$ with
$\kappa_{\bsb\bsa}\kappa_{\bsc\bsa}\neq0$, then
\begin{eqnarray*}
f_{\langle\bsA\rangle,\langle\bsB\rangle,\langle\bsC\rangle}(r,s,t)
=\delta((s-b_0)-\kappa_{\bsb\bsa}(r-a_0))\delta((t-c_0)-\kappa_{\bsc\bsa}(r-a_0))
f_{\langle\bsA\rangle}(r),
\end{eqnarray*}
where $f_{\langle\bsA\rangle}(r)$ is from Proposition~\ref{prop4}.
\end{prop}

\begin{proof}
(i) If $\rank(\bsT_{\bsa,\bsb,\bsc})=3$, then
$\bsT_{\bsa,\bsb,\bsc}$ is invertible. Using the integral
representation of delta function, we have
\begin{eqnarray*}\label{eq:fABC}
f_{\langle\bsA\rangle,\langle\bsB\rangle,\langle\bsC\rangle}(r,s,t)
= \frac1{8\pi^3}\int_{\real^3} \dif \alpha\dif \beta\dif\gamma
e^{\mathrm{i}(r\alpha+s\beta+t\gamma)} \int_{\rD(\complex^2)}
\dif\mu(\rho) e^{-\mathrm{i}\Tr{\bsS^{(\alpha,\beta,\gamma)}\rho}}.
\end{eqnarray*}
where $\dif\alpha\dif\beta\dif \gamma$ is the Lebesgue volume
element in $\real^3$,
$\bsS^{(\alpha,\beta,\gamma)}=\alpha\bsA+\beta\bsB+\gamma\bsC$. By
Proposition~\ref{prop3}, we have
\begin{eqnarray*}
\int_{\rD(\complex^2)}
\dif\mu(\rho)e^{-\mathrm{i}\Tr{\bsS^{(\alpha,\beta,\gamma)}\rho}}
=3e^{-\mathrm{i}(a_0\alpha+b_0\beta+c_0\gamma)}
\frac{\sin\Pa{\abs{\alpha\bsa+\beta\bsb+\gamma\bsc}}-\abs{\alpha\bsa+\beta\bsb+\gamma\bsc}\cos\Pa{\abs{\alpha\bsa+\beta\bsb+\gamma\bsc}}}{\abs{\alpha\bsa+\beta\bsb+\gamma\bsc}^3}.
\end{eqnarray*}
Consequently,
\begin{eqnarray*}
& & f_{\langle\bsA\rangle,\langle\bsB\rangle,\langle\bsC\rangle}(r,s,t)\\
&&= \frac3{8\pi^3}\int_{\real^3} \dif \alpha\dif \beta \dif\gamma
e^{\mathrm{i}\Pa{(r-a_0)\alpha+(s-b_0)\beta+(t-c_0)\gamma}}
\frac{\sin\Pa{\abs{\alpha\bsa+\beta\bsb+\gamma\bsc}}-\abs{\alpha\bsa+\beta\bsb+\gamma\bsc}\cos\Pa{\abs{\alpha\bsa+\beta\bsb+\gamma\bsc}}}{\abs{\alpha\bsa+\beta\bsb+\gamma\bsc}^3}.
\end{eqnarray*}
Let
$(\tilde\alpha,\tilde\beta,\tilde\gamma)^\t=\bsT^{\frac12}_{\bsa,\bsb,\bsc}(\alpha,\beta,\gamma)^\t$,
then $\dif\alpha\dif\beta\dif\gamma =
\det^{-\frac12}\Pa{\bsT_{\bsa,\bsb,\bsc}}\dif\tilde\alpha\dif\tilde\beta\dif\tilde\gamma,$
and
\begin{eqnarray*}
\int _{\rD(\complex^2)}\dif\mu(\rho)
e^{-\mathrm{i}\Tr{\bsS^{(\alpha,\beta,\gamma)}\rho}} =
3e^{-\mathrm{i}(\tilde a_0\tilde\alpha+\tilde b_0\tilde\beta+\tilde
c_0\tilde\gamma)}\frac{\sin\Pa{\sqrt{\tilde\alpha^2+\tilde\beta^2+\tilde\gamma^2}}-\sqrt{\tilde\alpha^2+\tilde\beta^2+\tilde\gamma^2}\cos\Pa{\sqrt{\tilde\alpha^2+\tilde\beta^2+\tilde\gamma^2}}}{\Pa{\tilde\alpha^2+\tilde\beta^2+\tilde\gamma^2}^{\frac32}}.
\end{eqnarray*}
Let $(\tilde r,\tilde s,\tilde
t)=(r,s,t)\bsT_{\bsa,\bsb,\bsc}^{-\frac12}$ and $(\tilde a_0,\tilde
b_0,\tilde c_0)=(a_0,b_0,c_0)\bsT_{\bsa,\bsb,\bsc}^{-\frac12},$ then
\begin{eqnarray*}
f_{\langle\bsA\rangle,\langle\bsB\rangle,\langle\bsC\rangle}(r,s,t)
&=& \frac3{8\pi^3\sqrt{\det(\bsT_{\bsa,\bsb,\bsc})}}\int _{\real^3}
\dif \tilde\alpha\dif \tilde\beta \dif\tilde \gamma
e^{\mathrm{i}\Pa{(\tilde r-\tilde a_0)\tilde\alpha+(\tilde s-\tilde
b_0)\tilde\beta+(\tilde t-\tilde
c_0)\tilde\gamma}}\\
&&\times\frac{\sin\Pa{\sqrt{\tilde\alpha^2+\tilde\beta^2+\tilde\gamma^2}}-\sqrt{\tilde\alpha^2+\tilde\beta^2+\tilde\gamma^2}\cos\Pa{\sqrt{\tilde\alpha^2+\tilde\beta^2+\tilde\gamma^2}}}{\Pa{\tilde\alpha^2+\tilde\beta^2+\tilde\gamma^2}^{\frac32}}.
\end{eqnarray*}
Denote $\tilde\bsz_0=(\tilde r-\tilde a_0,\tilde s-\tilde b_0,\tilde
t-\tilde c_0)$ and
$\tilde\bsz=(\tilde\alpha,\tilde\beta,\tilde\gamma),$ the last
integral can be rewritten as
\begin{eqnarray*}
&&\int _{\real^3} \dif\tilde\bsz
e^{\mathrm{i}\Inner{\tilde\bsz_0}{\tilde\bsz}}\frac{\sin\abs{\tilde\bsz}-\abs{\tilde\bsz}\cos\abs{\tilde\bsz}}{\abs{\tilde\bsz}^3}
=\int^\infty_0\dif q q^2\int _{\real^3}\dif
\tilde\bsu\delta(1-\abs{\tilde\bsu})
e^{\mathrm{i}q\Inner{\tilde\bsz_0}{\tilde\bsu}}\frac{\sin q-q\cos q}{q^3}\\
&&=\int^\infty_0\dif q\frac{\sin q-q\cos q}q\int_{\real^3} \dif
\tilde\bsu\delta(1-\abs{\tilde\bsu})
e^{\mathrm{i}q\Inner{\tilde\bsz_0}{\tilde\bsu}},
\end{eqnarray*}
where $\dif\tilde \bsz=\dif\tilde\alpha\dif\tilde\beta\dif\tilde
\gamma$ is the Lebesgue volume element in $\real^3.$
 From
\begin{eqnarray*}
\int _{\real^3} \dif \tilde\bsu\delta(1-\abs{\tilde\bsu})
e^{\mathrm{i}q\Inner{\tilde\bsz_0}{\tilde\bsu}} = 4\pi\int
\dif\mu_{\haar}(\tilde\bsu)e^{\mathrm{i}q\Inner{\tilde\bsz_0}{\tilde\bsu}}
= 4\pi\times\frac12\int^1_{-1}\dif \tau
e^{\mathrm{i}q\abs{\tilde\bsz_0}\tau} =
4\pi\frac{\sin(q\abs{\tilde\bsz_0})}{q\abs{\tilde\bsz_0}},
\end{eqnarray*}
where, in the second equality, we used the probability density
function of inner product of two random unit vectors, a result has
been already obtained in \cite{LZ2019PRE}, we obtain
\begin{eqnarray*}
\int _{\real^3} \dif  \tilde\bsz
e^{\mathrm{i}\Inner{\tilde\bsz_0}{\tilde\bsz}}\frac{\sin\abs{\tilde\bsz}-\abs{\tilde\bsz}\cos\abs{\tilde\bsz}}{\abs{\tilde\bsz}^3}
&=&\frac{4\pi}{\abs{\tilde\bsz_0}}\int^\infty_0\dif
q\frac{\sin(\abs{\tilde\bsz_0}q)(\sin q-q\cos q)}{q^2}\\
&=&\pi^2\frac{\abs{\tilde\bsz_0}-\abs{\abs{\tilde\bsz_0}-1}-\sign(\abs{\tilde\bsz_0}-1)}{\abs{\tilde\bsz_0}}.
\end{eqnarray*}
Therefore
\begin{eqnarray*}
f_{\langle\bsA\rangle,\langle\bsB\rangle,\langle\bsC\rangle}(r,s,t)
&=&
\frac{3\Pa{\abs{\tilde\bsz_0}-\abs{\abs{\tilde\bsz_0}-1}-\sign(\abs{\tilde\bsz_0}-1)}}{8\pi\sqrt{\det(\bsT_{\bsa,\bsb,\bsc})}\abs{\tilde\bsz_0}}.
\end{eqnarray*}
Noting that $ \abs{\tilde\bsz_0}=\sqrt{(\tilde r-\tilde
a_0)^2+(\tilde s-\tilde b_0)^2+(\tilde t-\tilde
c_0)^2}=\omega_{\bsA,\bsB,\bsC}(r,s,t), $ we finally get
\begin{eqnarray*}
f_{\langle \bsA\rangle ,\langle \bsB\rangle , \langle \bsC\rangle
}(r,s,t)
=\frac{3(1+\sign(1-\omega_{\bsA,\bsB,\bsC}(r,s,t)))}{8\pi\sqrt{\det(\bsT_{\bsa,\bsb,\bsc})}}.
\end{eqnarray*}

(ii) In this case,
\begin{eqnarray*}
f_{\langle\bsA\rangle,\langle\bsB\rangle,\langle\bsC\rangle}(r,s,t)
&=& \frac3{(2\pi)^3}\int_{\real^3} \dif \alpha\dif\beta\dif\gamma
e^{\mathrm{i}\Pa{(r-a_0)\alpha+(s-b_0)\beta+(t-c_0)\gamma}}\notag\\
&&\times\Pa{\frac{\sin\abs{(\alpha+\gamma
\kappa_{\bsa})\bsa+(\beta+\gamma
\kappa_{\bsb})\bsb}}{\abs{(\alpha+\gamma
\kappa_{\bsa})\bsa+(\beta+\gamma
\kappa_{\bsb})\bsb}^3}-\frac{\cos\abs{(\alpha+\gamma
\kappa_{\bsa})\bsa+(\beta+\gamma
\kappa_{\bsb})\bsb}}{\abs{(\alpha+\gamma
\kappa_{\bsa})\bsa+(\beta+\gamma \kappa_{\bsb})\bsb}^2}}.
\end{eqnarray*}
Let $(\alpha',\beta',\gamma')=(\alpha+\gamma
\kappa_{\bsa},\beta+\gamma \kappa_{\bsb},\gamma),$ then the Jacobian
of $(\alpha,\beta,\gamma)\to(\alpha',\beta',\gamma')$ is given by
\begin{eqnarray*}
\det\Pa{\frac{\partial(\alpha',\beta',\gamma')}{\partial(\alpha,\beta,\gamma)}}
= \abs{\begin{array}{ccc}
         1 & 0 & \kappa_{\bsa} \\
         0 & 1 & \kappa_{\bsb} \\
         0 & 0 & 1
       \end{array}
}=1\neq0.
\end{eqnarray*}
Thus we have
\begin{eqnarray*}
f_{\langle\bsA\rangle,\langle\bsB\rangle,\langle\bsC\rangle}(r,s,t)
&=& \frac1{2\pi}\int _{\real }
e^{\mathrm{i}\gamma'\Pa{(t-c_0)-\kappa_{\bsa}(r-a_0)-\kappa_{\bsb}(s-b_0)}}
\dif\gamma'
\notag\\
&&\times\frac3{(2\pi)^2}\int _{\real^2}\dif
\alpha'\dif\beta'e^{\mathrm{i}\Pa{(r-a_0)\alpha'+(s-b_0)\beta'}}\frac{\sin\abs{\alpha'\bsa+\beta'\bsb}-\abs{\alpha'\bsa+\beta'\bsb}
\cos\abs{\alpha'\bsa+\beta'\bsb}}{\abs{\alpha'\bsa+\beta'\bsb}^3}\notag\\
&=&\delta\Pa{(t-c_0)-\kappa_{\bsa}(r-a_0)-\kappa_{\bsb}(s-b_0)}f_{\langle\bsA\rangle,\langle\bsB\rangle}(r,s),
\end{eqnarray*}
where $f_{\langle\bsA\rangle,\langle\bsB\rangle}(r,s)$ is from
Proposition~\ref{prop7}.

(iii) In this case, we have
\begin{eqnarray*}
f_{\langle\bsA\rangle,\langle\bsB\rangle,\langle\bsC\rangle}(r,s,t)
&=& \frac3{(2\pi)^3}\int _{\real^3} \dif \alpha\dif\beta\dif\gamma
e^{\mathrm{i}\Pa{(r-a_0)\alpha+(s-b_0)\beta+(t-c_0)\gamma}}\notag\\
&&\times\frac{\sin\Pa{\abs{\bsa}\abs{\alpha+\kappa_{\bsb\bsa}\beta+\kappa_{\bsc\bsa}\gamma}}-\abs{\bsa}\abs{\alpha+\kappa_{\bsb\bsa}\beta+\kappa_{\bsc\bsa}\gamma}
\cos\Pa{\abs{\bsa}\abs{\alpha+\kappa_{\bsb\bsa}\beta+\kappa_{\bsc\bsa}\gamma}}}{\abs{\bsa}^3\abs{\alpha+\kappa_{\bsb\bsa}\beta+\kappa_{\bsc\bsa}\gamma}^3}.
\end{eqnarray*}
Let $(\alpha',\beta',\gamma')=(\alpha+
\kappa_{\bsb\bsa}\beta+\kappa_{\bsc\bsa}\gamma,\beta,\gamma), $ then
the Jacobian of the transformation
$(\alpha,\beta,\gamma)\to(\alpha',\beta',\gamma')$ is given by
\begin{eqnarray*}
\det\Pa{\frac{\partial(\alpha',\beta',\gamma')}{\partial(\alpha,\beta,\gamma)}}
= \abs{\begin{array}{ccc}
         1 & \kappa_{\bsb\bsa} & \kappa_{\bsc\bsa} \\
         0 & 1 & 0 \\
         0 & 0 & 1
       \end{array}
}=1\neq0,
\end{eqnarray*}
and we have
\begin{eqnarray*}
& &f_{\langle\bsA\rangle,\langle\bsB\rangle,\langle\bsC\rangle}(r,s,t) \\
&&= \frac3{(2\pi)^3}\int _{\real^3} \dif
\alpha'\dif\beta'\dif\gamma'
e^{\mathrm{i}\Pa{(r-a_0)(\alpha'-\kappa_{\bsb\bsa}\beta'-\kappa_{\bsc\bsa}\gamma')+(s-b_0)\beta'+(t-c_0)\gamma'}}
\frac{\sin\Pa{\abs{\bsa}\abs{\alpha'}}-\abs{\bsa}\abs{\alpha'}
\cos\Pa{\abs{\bsa}\abs{\alpha'}}}{\abs{\bsa}^3\abs{\alpha'}^3}\notag\\
&&=\frac1{2\pi}\int _{\real }
e^{\mathrm{i}[(s-b_0)-\kappa_{\bsb}(r-a_0)]\beta'}\dif\beta'\times\frac1{2\pi}\int_\real
e^{\mathrm{i}[(t-c_0)-\kappa_{\bsc}(r-a_0)]\gamma'}\dif\gamma'\\
&&~~~~~~\times \frac3{2\pi}\int_{\real}
\dif\alpha'e^{\mathrm{i}(r-a_0)\alpha'}\frac{\sin\Pa{\abs{\bsa}\abs{\alpha'}}-\abs{\bsa}\abs{\alpha'}
\cos\Pa{\abs{\bsa}\abs{\alpha'}}}{\abs{\bsa}^3\abs{\alpha'}^3} \\
&&=\delta\Pa{(s-b_0)-\kappa_{\bsb\bsa}(r-a_0)}\delta\Pa{(t-c_0)-\kappa_{\bsc\bsa}(r-a_0)}f_{\langle\bsA\rangle}(r),
\end{eqnarray*}
where $f_{\langle\bsA\rangle}(r)$ is from Proposition~\ref{prop4}.
\end{proof}

\begin{prop}\label{prop10}
The joint probability distribution density of $(\Delta_\rho
\bsA,\Delta_\rho\bsB,\Delta_\rho \bsC)$ for a triple of qubit
observables $ \bsA =a_0\I+\bsa\cdot\boldsymbol{\sigma}, \ \bsB
=b_0\I+\bsb\cdot\boldsymbol{\sigma}, \ \bsC
=c_0\I+\bsc\cdot\boldsymbol{\sigma}, \
(a_0,\bsa),(b_0,\bsb),(c_0,\bsc)\in \real^4$, where
$\set{\bsa,\bsb,\bsc}$ is linearly independent, and $\rho$ is
resulted from partially tracing a subsystem over a Haar-distributed
random pure state on $\complex^2\ot\complex^2$, is given by
\begin{eqnarray*}
f_{\Delta\bsA,\Delta\bsB,\Delta\bsC}(x,y,z)
=\frac{2xyz}{\sqrt{(a^2-x^2)(b^2-y^2)(c^2-z^2)}}\sum_{j,k\in\set{\pm}}
f_{\langle \bsA\rangle ,\langle \bsB\rangle,\langle \bsC\rangle}(
x_+ ,y_j,z_k).
\end{eqnarray*}
Here $f_{\langle \bsA\rangle,\langle \bsB\rangle, \langle
\bsC\rangle}(r,s,t)$ is the joint probability distribution density
of the mean values  $(\langle \bsA\rangle_\rho,\langle
\bsB\rangle_\rho,\langle \bsC\rangle_\rho)$, determined by
Proposition~\ref{prop9}, and
\begin{eqnarray*}
x_\pm =a_0\pm\sqrt{\abs{\bsa}^2-x^2},\qquad y_\pm
=b_0\pm\sqrt{\abs{\bsb}^2-y^2},\qquad z_\pm
=c_0\pm\sqrt{\abs{\bsc}^2-z^2}.
\end{eqnarray*}
\end{prop}

\begin{proof}
Noting that
\begin{eqnarray*}
f_{\Delta\bsA,\Delta\bsB,\Delta\bsC}(x,y,z)&=&\int
_{\rD(\complex^2)}\dif\mu(\rho)\delta(x-\Delta_\rho
\bsA)\delta(y-\Delta_\rho
\bsB)\delta(z-\Delta_\rho \bsC) \notag\\
&=& 8xyz\int _{\rD(\complex^2)}\dif\mu(\rho)\delta\Pa{x^2-(\Delta
_\rho \bsA)^2 }\delta\Pa{y^2-( \Delta _\rho \bsB)^2
}\delta\Pa{z^2-(\Delta _\rho \bsC)^2 },
\end{eqnarray*}
and
\begin{eqnarray*}
\delta\Pa{x^2-(\Delta _\rho \bsA)^2 }\delta\Pa{y^2-(\Delta _\rho
\bsB)^2 }\delta\Pa{z^2-(\Delta_\rho \bsC)^2
}=\delta(g_x(r))\delta(h_y(s))\delta(l_z(t)),
\end{eqnarray*}
where $g_x(r)=x^2-(r-\lambda_1(\bsA))(\lambda_2(\bsA)-r),\
h_y(s)=y^2-(s-\lambda_1(\bsB))(\lambda_2(\bsB)-s),\
l_z(t)=z^2-(t-\lambda_1(\bsC))(\lambda_2(\bsC)-t),$ we have
\begin{eqnarray*}
f_{\Delta\bsA,\Delta\bsB,\Delta\bsC}(x,y,z)&=&8xyz\int _{\real^3}
\dif r\dif s\dif t
\delta(g_x(r))\delta(h_y(s))\delta(l_z(t))f_{\langle\bsA\rangle,\langle\bsB\rangle,\langle\bsC\rangle}(r,s,t) \\ &=&\frac{xyz}{\sqrt{(\abs{\bsa}^2-x^2)(\abs{\bsb}^2-y^2)(\abs{\bsc}^2-z^2)}}\sum_{i,j,k\in\set{\pm}}\Inner{\delta_{(r_i(x), s_j(y),t_k(z))}}{f_{\langle\bsA\rangle,\langle\bsB\rangle,\langle\bsC\rangle}} \\
& =&
\frac{xyz\sum_{i,j,k\in\set{\pm}}f_{\langle\bsA\rangle,\langle\bsB\rangle,\langle\bsC\rangle}( r_i(x), s_j(y) ,t_k(z))}{\sqrt{(\abs{\bsa}^2-x^2)(\abs{\bsb}^2-y^2)(\abs{\bsc}^2-z^2)}} \\
&=&
\frac{2xyz\sum_{j,k\in\set{\pm}}f_{\langle\bsA\rangle,\langle\bsB\rangle,\langle\bsC\rangle}(
r_+,
s_j,t_k)}{\sqrt{(\abs{\bsa}^2-x^2)(\abs{\bsb}^2-y^2)(\abs{\bsc}^2-z^2)}},
\end{eqnarray*}
where we have used the fact that
\begin{eqnarray*}
& & f_{\langle\bsA\rangle,\langle\bsB\rangle,\langle\bsC\rangle}(r_+, s_+, t_+ )=f_{\langle\bsA\rangle,\langle\bsB\rangle,\langle\bsC\rangle}(r_-, s_-, t_- ),\quad f_{\langle\bsA\rangle,\langle\bsB\rangle,\langle\bsC\rangle}(r_+, s_+, t_-)=f_{\langle\bsA\rangle,\langle\bsB\rangle,\langle\bsC\rangle}(r_-, s_-, t_+),\\
& &
f_{\langle\bsA\rangle,\langle\bsB\rangle,\langle\bsC\rangle}(r_+,
s_-,
t_+)=f_{\langle\bsA\rangle,\langle\bsB\rangle,\langle\bsC\rangle}(r_-,
s_+, t_-),\quad
f_{\langle\bsA\rangle,\langle\bsB\rangle,\langle\bsC\rangle}(r_+,
s_-,
t_-)=f_{\langle\bsA\rangle,\langle\bsB\rangle,\langle\bsC\rangle}(
r_-, s_+, t_+).
\end{eqnarray*}
This completes the proof.
\end{proof}

We now turn to the uncertainty region for a triple
$(\bsA,\bsB,\bsC)$ of qubit observables
$\bsA=a_0\I+\bsa\cdot\boldsymbol{\sigma}$,
$\bsB=b_0\I+\bsb\cdot\boldsymbol{\sigma}$ and
$\bsC=c_0\I+\bsc\cdot\boldsymbol{\sigma}$ with $(a_0,\bsa)$,
$(b_0,\bsb)$, $(c_0,\bsc)\in\real^4$, and $\set{\bsa,\bsb,\bsc}$
being linearly independent. Denote
$\bsu_{\epsilon_b,\epsilon_c}(x,y,z)=((\abs{\bsa}^2-x^2)^{1/2},\epsilon_b(\abs{\bsb}^2-y^2)^{1/2},
\epsilon_c(\abs{\bsc}^2-z^2)^{1/2})$, where
$\epsilon_b,\epsilon_c\in\set{\pm1}$, and
\begin{eqnarray*}
\bsT_{\bsa,\bsb,\bsc}=\Pa{\begin{array}{ccc}
                               \Inner{\bsa}{\bsa} & \Inner{\bsa}{\bsb} & \Inner{\bsa}{\bsc} \\
                               \Inner{\bsb}{\bsa} & \Inner{\bsb}{\bsb} & \Inner{\bsb}{\bsc} \\
                               \Inner{\bsc}{\bsa} & \Inner{\bsc}{\bsb} & \Inner{\bsc}{\bsc}
                             \end{array}}.
\end{eqnarray*}
Let $\gamma$, $\beta$ and $\alpha$ be the angles between $\bsa$ and
$\bsb$, $\bsa$ and $\bsc$, $\bsb$ and $\bsc$, respectively, where
$\alpha,\beta,\gamma\in(0,\pi)$. Set $\phi(t_1,t_2,t_3)=\cos
(t_1)-\cos (t_2)\cos (t_3)$.

\begin{proof}[Proof of Lemma~\ref{lem2}]
In fact, due to the relation between
$f_{\Delta\bsA,\Delta\bsB,\Delta\bsC}$ and
$f_{\langle\bsA\rangle,\langle\bsB\rangle,\langle\bsC\rangle}$, in
Proposition~\ref{prop10}, the support of
$f_{\Delta\bsA,\Delta\bsB,\Delta\bsC}$ can be identified by the
support of
$f_{\langle\bsA\rangle,\langle\bsB\rangle,\langle\bsC\rangle}$ which
can be seen from Proposition~\ref{prop9} (i),
\begin{eqnarray*}
\supp(f_{\langle\bsA\rangle,\langle\bsB\rangle,\langle\bsC\rangle})
= \Set{(r,s,t)\in\real^2:\omega_{\bsA,\bsB,\bsC}(r,s,t)\leqslant1}.
\end{eqnarray*}
With the support of
$f_{\langle\bsA\rangle,\langle\bsB\rangle,\langle\bsC\rangle}$, now
we can make an analysis of the support of
$f_{\Delta\bsA,\Delta\bsB,\Delta\bsC}$. Note that
$f_{\Delta\bsA,\Delta\bsB,\Delta\bsC}(x,y,z)$ is defined on the
first quadrant $\real^3_{\geqslant0}$, if $xyz>0$, then
$f_{\Delta\bsA,\Delta\bsB,\Delta\bsC}(x,y,z)=0$ if and only if
$$
\sum_{i,j,k\in\set{\pm}}f_{\langle\bsA\rangle,\langle\bsB\rangle,\langle\bsC\rangle}(r_i(x),s_j(y),t_k(z))=0,
$$
i.e.,
$f_{\langle\bsA\rangle,\langle\bsB\rangle,\langle\bsC\rangle}(r_i(x),s_j(y),t_k(z))=0$
because
$f_{\langle\bsA\rangle,\langle\bsB\rangle,\langle\bsC\rangle}$ is a
non-negative function. This means that all eight points
$(r_\pm(x),s_\pm(y),t_\pm(z))$ are not in the support of
$f_{\langle\bsA\rangle,\langle\bsB\rangle,\langle\bsC\rangle}$.
Therefore the uncertainty region (i.e.,the support of
$f_{\Delta\bsA,\Delta\bsB,\Delta\bsC}$) of $\bsA,\bsB$, and $\bsC$
is given by the following set:
\begin{eqnarray*}
\cU_{\Delta\bsA,\Delta\bsB,\Delta\bsC}=\supp(f_{\Delta\bsA,\Delta\bsB,\Delta\bsC})=
D^{++}_{\bsa,\bsb,\bsc}\cup D^{+-}_{\bsa,\bsb,\bsc}\cup
D^{-+}_{\bsa,\bsb,\bsc}\cup D^{--}_{\bsa,\bsb,\bsc},
\end{eqnarray*}
where, via
$\bsu_{\epsilon_b,\epsilon_c}(x,y,z)=\Pa{\sqrt{\abs{\bsa}^2-x^2},\epsilon_b\sqrt{\abs{\bsb}^2-y^2},
\epsilon_c\sqrt{\abs{\bsc}^2-z^2}}$,
\begin{eqnarray*}
D^{\epsilon_b\epsilon_c}_{\bsa,\bsb,\bsc}=\Set{(x,y,z)\in\real^3_{\geqslant0}:
\bsu_{\epsilon_b,\epsilon_c}(x,y,z)\bsT^{-1}_{\bsa,\bsb,\bsc}\bsu^\t_{\epsilon_b,\epsilon_c}(x,y,z)\leqslant1}
\end{eqnarray*}
for $\epsilon_b,\epsilon_c\in\set{\pm}$. Now the inequality
\begin{eqnarray*}
\bsu_{\epsilon_b,\epsilon_c}(x,y,z)\bsT^{-1}_{\bsa,\bsb,\bsc}\bsu^\t_{\epsilon_b,\epsilon_c}(x,y,z)\leqslant1
\end{eqnarray*}
is reduced into the following:
\begin{eqnarray*}
&&\Big[\det(\bsT_{\bsb,\bsc})x^2+\det(\bsT_{\bsa,\bsc})y^2+\det(\bsT_{\bsa,\bsb})z^2+2(\Inner{\bsa}{\bsb}\Inner{\bsa}{\bsc}\Inner{\bsb}{\bsc}-\abs{\bsa}^2\abs{\bsb}^2\abs{\bsc}^2)\\
&&~~~+2\epsilon_b(\abs{\bsc}^2\Inner{\bsa}{\bsb}-\Inner{\bsa}{\bsc}\Inner{\bsb}{\bsc})\sqrt{(\abs{\bsa}^2-x^2)(\abs{\bsb}^2-y^2)}\\
&&~~~+2\epsilon_b\epsilon_c(\abs{\bsa}^2\Inner{\bsb}{\bsc}-\Inner{\bsa}{\bsb}\Inner{\bsa}{\bsc})\sqrt{(\abs{\bsb}^2-y^2)(\abs{\bsc}^2-z^2)}\\
&&~~~+2\epsilon_c(\abs{\bsb}^2\Inner{\bsa}{\bsc}-\Inner{\bsa}{\bsb}\Inner{\bsb}{\bsc})\sqrt{(\abs{\bsc}^2-z^2)(\abs{\bsa}^2-x^2)}\Big]\\
&&\times\sign
\Pa{\abs{\bsa}^2\Inner{\bsb}{\bsc}^2+\abs{\bsb}^2\Inner{\bsa}{\bsc}^2+c^2\Inner{\bsa}{\bsb}^2-\abs{\bsa}^2\abs{\bsb}^2\abs{\bsc}^2-2\Inner{\bsa}{\bsb}\Inner{\bsa}{\bsc}\Inner{\bsb}{\bsc}}\\
&&\leqslant0.
\end{eqnarray*}
Denote the angle between $\bsa$ and $\bsb$ by $\gamma$; the angle
between $\bsa$ and $\bsc$ by $\beta$; the angle between $\bsb$ and
$\bsc$ by $\alpha$. Thus
\begin{eqnarray*}
\Inner{\bsa}{\bsb} = \abs{\bsa}\abs{\bsb}\cos(\gamma),\quad
\Inner{\bsa}{\bsc} =\abs{\bsa}\abs{\bsc}\cos(\beta),\quad
\Inner{\bsb}{\bsc} = \abs{\bsb}\abs{\bsc}\cos(\alpha).
\end{eqnarray*}
We also write $\phi(t_1,t_2,t_3):=\cos (t_1)-\cos (t_2)\cos (t_3)$.
By scaling transformations:
\begin{eqnarray*}
(x,y,z)\longmapsto (ax,by,cz)
\end{eqnarray*}
without loss of generality, we assume
$\abs{\bsa}=\abs{\bsb}=\abs{\bsc}=1$, then the above inequalities is
equivalent to:
\begin{eqnarray*}
&&\big[\sin^2(\alpha)x^2 + \sin^2(\beta)y^2 + \sin^2(\gamma)z^2 \\
&&~~~+2\epsilon_b\phi(\gamma,\alpha,\beta)\sqrt{(1-x^2)(1-y^2)}+2\epsilon_b\epsilon_c\phi(\alpha,\beta,\gamma)\sqrt{(1-y^2)(1-z^2)}\\
&&~~~+2\epsilon_c\phi(\beta,\gamma,\alpha)\sqrt{(1-z^2)(1-x^2)}+2(\cos(\alpha)\cos(\beta)\cos(\gamma)-1)\big]\\
&&\times\sign\Pa{\cos^2(\alpha)+\cos^2(\beta)+\cos^2(\gamma)-2\cos(\alpha)\cos(\beta)\cos(\gamma)-1}\\
&&\leqslant0.
\end{eqnarray*}
Note that three angles $\alpha,\beta,\gamma\in(0,\pi)$ should be
such that
$$
\det(\bsT_{\bsa,\bsb,\bsc})=1-\cos^2(\alpha)-\cos^2(\beta)-\cos^2(\gamma)+2\cos(\alpha)\cos(\beta)\cos(\gamma)>0
$$
due to the fact that $\set{\bsa,\bsb,\bsc}$ is linearly independent.
Because
$$
-\det(\bsT_{\bsa,\bsb,\bsc})=\cos^2(\alpha)+\cos^2(\beta)+\cos^2(\gamma)-2\cos(\alpha)\cos(\beta)\cos(\gamma)-1=0
$$
if and only if $\cos(\gamma)=\cos(\alpha+\beta)$ or
$\cos(\gamma)=\cos(\alpha-\beta)$; if and only if
$\cos(\beta)=\cos(\alpha+\gamma)$ or
$\cos(\beta)=\cos(\alpha-\gamma)$; if and only if
$\cos(\alpha)=\cos(\beta+\gamma)$ or
$\cos(\alpha)=\cos(\beta-\gamma)$. Thus the necessary and sufficient
condition of $\det(\bsT_{\bsa,\bsb,\bsc})=0$ is either one of the
following statements:
\begin{enumerate}
\item[(i)] $\alpha=\beta+\gamma$;
\item[(ii)] $\beta=\gamma+\beta$;
\item[(iii)] $\gamma=\alpha+\beta$;
\item[(iv)] $\alpha+\beta+\gamma=2\pi$.
\end{enumerate}
In other words, $\det(\bsT_{\bsa,\bsb,\bsc})>0$ if and only if the
following four inequalities should be satisfied
\begin{eqnarray*}
\alpha,\beta,\gamma\in(0,\pi):\quad\alpha<\beta+\gamma,\quad
\beta<\gamma+\alpha, \quad\gamma<\alpha+\beta,\quad
\alpha+\beta+\gamma<2\pi.
\end{eqnarray*}
Under these conditions, the above inequalities is reduced to
\begin{eqnarray*}
&&2\epsilon_b\phi(\gamma,\alpha,\beta)\sqrt{(1-x^2)(1-y^2)}+2\epsilon_b\epsilon_c\phi(\alpha,\beta,\gamma)\sqrt{(1-y^2)(1-z^2)}+
2\epsilon_c\phi(\beta,\gamma,\alpha)\sqrt{(1-z^2)(1-x^2)}\notag\\
&&+\sin^2(\alpha)x^2 + \sin^2(\beta)y^2 +
\sin^2(\gamma)z^2\geqslant2(1-\cos(\alpha)\cos(\beta)\cos(\gamma)).
\end{eqnarray*}
In order to get a more compact form of the above formula, we find
that we can use the absolute function to get rid of $\epsilon_b$
with keeping $\epsilon_c\equiv\epsilon\in\set{\pm1}$ due to the
independence of $\epsilon_b$ and $\epsilon_c$. Thus they can also be
rewritten as the following form by multiplying $\frac12$ on both
sides:
\begin{eqnarray*}
&&\abs{\phi(\gamma,\alpha,\beta)\sqrt{1-x^2}
+\epsilon\phi(\alpha,\beta,\gamma)\sqrt{1-z^2}}\sqrt{1-y^2}+
\epsilon\phi(\beta,\gamma,\alpha)\sqrt{(1-z^2)(1-x^2)}\notag\\
&&+\frac12\Br{\sin^2(\alpha)x^2 + \sin^2(\beta)y^2 +
\sin^2(\gamma)z^2}\geqslant1-\cos(\alpha)\cos(\beta)\cos(\gamma)\quad
(\epsilon\in\set{\pm1}),
\end{eqnarray*}
where $x,y,z\in[0,1]$.
\end{proof}

\subsection{Proof of Theorem~\ref{th2}}\label{app-theorem2}

Based on the proof of Lemma~\ref{lem2}, we present the proof of
Theorem~\ref{th2} as follows.

\begin{proof}[Proof of Theorem~\ref{th2}]
Let $\sqrt{\abs{\bsa}^2-x^2}=X,\sqrt{\abs{\bsb}^2-y^2}=Y$, and
$\sqrt{\abs{\bsc}^2-z^2}=Z$ in Lemma~\ref{lem2}, where
$X\in[0,\abs{\bsa}],Y\in[0,\abs{\bsb}]$, and $Z\in[0,\abs{\bsc}]$
due to the fact that $x\in[0,\abs{\bsa}],y\in[0,\abs{\bsb}]$, and
$z\in[0,\abs{\bsc}]$. Thus we get that
\begin{eqnarray*}
&&\min
\Set{x^2+y^2+z^2:(x,y,z)\in\cU_{\Delta\bsA,\Delta\bsB,\Delta\bsC}}\\
&&=\Tr{\bsT_{\bsa,\bsb,\bsc}}-\max\Set{ X^2+Y^2+Z^2:
(X,\epsilon_bY,\epsilon_cZ)\bsT^{-1}_{\bsa,\bsb,\bsc}(X,\epsilon_bY,\epsilon_cZ)^\t\leqslant1}(\epsilon_b,\epsilon_c\in\set{\pm1}).
\end{eqnarray*}
It is easily seen that the objection function $x^2+y^2+z^2$, where
$(x,y,z)\in\cU_{\Delta\bsA,\Delta\bsB,\Delta\bsC}$, attains its
minimal value on the boundary surface
$\partial\cU_{\Delta\bsA,\Delta\bsB,\Delta\bsC}$ of the uncertainty
region $\cU_{\Delta\bsA,\Delta\bsB,\Delta\bsC}$; this also
corresponds to the objection function $X^2+Y^2+Z^2$ attains its
maximal value on the boundary surface
$$
(X,\epsilon_bY,\epsilon_cZ)\bsT^{-1}_{\bsa,\bsb,\bsc}(X,\epsilon_bY,\epsilon_cZ)^\t=1.
$$

Denote by $\lambda_k(\bsT_{\bsa,\bsb,\bsc})(k=1,2,3)$ the three
eigenvalues of the matrix $\bsT_{\bsa,\bsb,\bsc}$. By Spectral
Decomposition Theorem, we get that there exists orthogonal
$\bsO\in\O(3)$ such that
\begin{eqnarray*}
\bsT_{\bsa,\bsb,\bsc} =
\bsO\diag(\lambda_1(\bsT_{\bsa,\bsb,\bsc}),\lambda_2(\bsT_{\bsa,\bsb,\bsc}),\lambda_3(\bsT_{\bsa,\bsb,\bsc}))\bsO^\t.
\end{eqnarray*}
Now let $(X',Y',Z')^\t=\bsO(X,\epsilon_bY,\epsilon_cZ)^\t$. Then
\begin{eqnarray*}
(X,\epsilon_bY,\epsilon_cZ)\bsT^{-1}_{\bsa,\bsb,\bsc}(X,\epsilon_bY,\epsilon_cZ)^\t
&&=(X',Y',Z')\diag(\lambda^{-1}_1(\bsT_{\bsa,\bsb,\bsc}),\lambda^{-1}_2(\bsT_{\bsa,\bsb,\bsc}),\lambda^{-1}_3(\bsT_{\bsa,\bsb,\bsc}))(X',Y',Z')^\t\\
&&=
\frac{{X'}^2}{\lambda_1(\bsT_{\bsa,\bsb,\bsc})}+\frac{{Y'}^2}{\lambda_2(\bsT_{\bsa,\bsb,\bsc})}+\frac{{Z'}^2}{\lambda_3(\bsT_{\bsa,\bsb,\bsc})}.
\end{eqnarray*}
Because these rotations do not change the length of vectors, we get
that
\begin{eqnarray*}
&&\max
\set{X^2+Y^2+Z^2:(X,\epsilon_bY,\epsilon_cZ)\bsT^{-1}_{\bsa,\bsb,\bsc}(X,\epsilon_bY,\epsilon_cZ)^\t=1}\\
&&=\max
\Set{X^2+(\epsilon_bY)^2+(\epsilon_cZ)^2:(X,\epsilon_bY,\epsilon_cZ)\bsT^{-1}_{\bsa,\bsb,\bsc}(X,\epsilon_bY,\epsilon_cZ)^\t=1}\\
&&=\max
\Set{{X'}^2+{Y'}^2+{Z'}^2:\frac{{X'}^2}{\lambda_1(\bsT_{\bsa,\bsb,\bsc})}+\frac{{Y'}^2}{\lambda_2(\bsT_{\bsa,\bsb,\bsc})}+\frac{{Z'}^2}{\lambda_3(\bsT_{\bsa,\bsb,\bsc})}=1}.
\end{eqnarray*}
Again, the above optimization problem becomes
\begin{eqnarray*}
&&\min
\Set{x^2+y^2+z^2:(x,y,z)\in\partial\cU_{\Delta\bsA,\Delta\bsB,\Delta\bsC}}\\
&&=\Tr{\bsT_{\bsa,\bsb,\bsc}}-\max\Set{{X'}^2+{Y'}^2+{Z'}^2:\frac{{X'}^2}{\lambda_1(\bsT_{\bsa,\bsb,\bsc})}+\frac{{Y'}^2}{\lambda_2(\bsT_{\bsa,\bsb,\bsc})}+\frac{{Z'}^2}{\lambda_3(\bsT_{\bsa,\bsb,\bsc})}=1}.
\end{eqnarray*}
Therefore
\begin{eqnarray*}
\min
\Set{x^2+y^2+z^2:(x,y,z)\in\cU_{\Delta\bsA,\Delta\bsB,\Delta\bsC}}=\Tr{\bsT_{\bsa,\bsb,\bsc}}-\lambda_{\max}(\bsT_{\bsa,\bsb,\bsc})
\end{eqnarray*}
implying that
\begin{eqnarray*}
(\Delta_\rho\bsA)^2+(\Delta_\rho\bsB)^2+(\Delta_\rho\bsC)^2\geqslant
\Tr{\bsT_{\bsa,\bsb,\bsc}}-\lambda_{\max}(\bsT_{\bsa,\bsb,\bsc}).
\end{eqnarray*}
We are done.
\end{proof}

\subsection{Proof of Lemma~\ref{lem3}}\label{app-lemma3}

The proof of Lemma~\ref{lem3} is based on Propositions~\ref{prop11}
and \ref{prop12}, which are described as follows.

For an $n$-tuple of qubit observables $(\bsA_1,\ldots,\bsA_n)$ where
\begin{eqnarray*}\label{eq:k}
\bsA_k=a^{(k)}_0\I+\bsa_k\cdot\boldsymbol{\sigma},\qquad(a^{(k)}_0,
\bsa_k)\in\real^4, \quad k=1, 2,\cdots,n
\end{eqnarray*}
 the eigenvalues of $\bsA_k$ are given, respectively, by
\begin{eqnarray*}
\lambda_i(\bsA_k)=a^{(k)}_0+(-1)^i\abs{\bsa_k},\quad i=1,2; k=1,
\ldots,n
\end{eqnarray*}
By the assumption that $\lambda_2(\bsA_k)>\lambda_1(\bsA_k)$ for all
$k=1,\ldots,n$, we see that $\abs{\bsa_k}>0, k=1, 2, \cdots, n$. Let
\begin{eqnarray*}\label{eq:EABN-def}
f_{\langle\bsA_1\rangle,\cdots,\langle\bsA_n\rangle}(r_1,\ldots,r_n)
= \int_{\rD(\complex^2)} \dif\mu(\rho)
\prod^n_{k=1}\delta(r_k-\langle \bsA_k\rangle_\rho)
\end{eqnarray*}
be the joint probability density of the mean values of
$(\langle\bsA_1\rangle_\rho,\cdots,\langle\bsA_n\rangle_\rho)$.
Denote
\begin{eqnarray*}
\bsT_{\bsa_1,\ldots,\bsa_n} = \Pa{\begin{array}{cccc}
                                    \Inner{\bsa_1}{\bsa_1} & \Inner{\bsa_1}{\bsa_2} & \cdots & \Inner{\bsa_1}{\bsa_n} \\
                                    \Inner{\bsa_2}{\bsa_1} & \Inner{\bsa_2}{\bsa_2} & \cdots & \Inner{\bsa_2}{\bsa_n}\\
                                    \vdots & \vdots & \ddots &
                                    \vdots\\
                                    \Inner{\bsa_n}{\bsa_1} & \Inner{\bsa_n}{\bsa_2} & \cdots & \Inner{\bsa_n}{\bsa_n}
                                  \end{array}
}.
\end{eqnarray*}

\begin{prop}\label{prop11}
Let $(\bsA_1,\ldots,\bsA_n)$ be an $n$-tuple of qubit observables
given by $ \bsA_k=a^{(k)}_0\I+\bsa_k\cdot\boldsymbol{\sigma},\
(a^{(k)}_0, \bsa_k)\in\real^4, k=1,\cdots,n,$ and let
$f_{\langle\bsA_1\rangle,\cdots,\langle\bsA_n\rangle}(r_1,\ldots,r_n)$
be the joint probability distribution density of the mean values
$(\langle\bsA_1\rangle_\rho,\ldots,\langle\bsA_n\rangle_\rho),$
where $\rho$ is resulted from partially tracing a subsystem over a
Haar-distributed random pure state on $\complex^2\ot\complex^2.$

(i) If $\rank(\bsT_{\bsa_1,\ldots,\bsa_n})=3$, then without loss of
generality, we can take nonzero coefficients $\kappa_{lj},
l=4,\cdots,n, j=1,2,3$, such that $
\bsa_l=\kappa_{l1}\bsa_1+\kappa_{l2}\bsa_2+\kappa_{l3}\bsa_3$. In
this case,
\begin{eqnarray*}
f_{\langle\bsA_1\rangle,\cdots,\langle\bsA_n\rangle}(r_1,\ldots,r_n)
&=&
f_{\langle\bsA_1\rangle,\langle\bsA_2\rangle,\langle\bsA_3\rangle}(r_1,r_2,r_3)\prod^n_{l=4}\delta
\Big ((r_l-a^{(l)}_0)-\sum^3_{j=1}\kappa_{lj}(r_j-a^{(j)}_0)\Big ),
\end{eqnarray*}
where
$f_{\langle\bsA_1\rangle,\langle\bsA_2\rangle,\langle\bsA_3\rangle}$
is determined similarly by Proposition~\ref{prop9}.

(ii) If $\rank(\bsT_{\bsa_1,\ldots,\bsa_n})=2$, then  without loss
of generality, we  take nonzero coefficients $\eta_{li},
l=2,\cdots,n, i=1,2, $ such that $
\bsa_l=\eta_{l1}\bsa_1+\eta_{l2}\bsa_2.$ In this case
\begin{eqnarray*}
f_{\langle\bsA_1\rangle,\cdots,\langle\bsA_n\rangle}(r_1,\ldots,r_n)
&=&
f_{\langle\bsA_1\rangle,\langle\bsA_2\rangle}(r_1,r_2)\prod^n_{l=3}\delta
\Big ((r_l-a^{(l)}_0)-\sum^2_{i=1}\eta_{li}(r_i-a^{(i)}_0)\Big),
\end{eqnarray*}
where $f_{\langle\bsA_1\rangle,\langle\bsA_2\rangle}$ is determined
similarly by Proposition~\ref{prop7}.

(iii) If $\rank(\bsT_{\bsa_1,\ldots,\bsa_n})=1$, then   without loss
of generality, we can take nonzero coefficient $\kappa_{l},
l=2,\cdots,n, $ such that $\bsa_l=\kappa_{l}\bsa_1.$ In this case,
\begin{eqnarray*}
f_{\langle\bsA_1\rangle,\cdots,\langle\bsA_n\rangle}(r_1,\ldots,r_n)
&=&
f_{\langle\bsA_1\rangle}(r_1)\prod^n_{l=2}\delta\Pa{(r_l-a^{(l)}_0)-\kappa_{l}(r_1-a^{(1)}_0)},
\end{eqnarray*}
where $f_{\langle\bsA_1\rangle}$ is determined similarly by
Proposition~\ref{prop4}.
\end{prop}

\begin{proof}
By using integral representation of delta function $n$ times in
$f_{\langle\bsA_1\rangle,\ldots,\langle\bsA_n\rangle}$ and
Proposition~\ref{prop3}, we get
\begin{eqnarray*}\label{eq:finalform}
f_{\langle\bsA_1\rangle,\ldots,\langle\bsA_n\rangle}(r_1,\ldots,r_n)
&=& \frac3{(2\pi)^n}\int _{\real^n} \prod^n_{k=1}\dif \alpha_k
e^{\mathrm{i}\Pa{\sum^n_{k=1}(r_k-a^{(k)}_0)\alpha_k}}\notag\\
&&\times\frac{\sin\Pa{\abs{\sum^n_{k=1}\alpha_k\bsa_k}}-\abs{\sum^n_{k=1}\alpha_k\bsa_k}
\cos\Pa{\abs{\sum^n_{k=1}\alpha_k\bsa_k}}}{\abs{\sum^n_{k=1}\alpha_k\bsa_k}^3}.
\end{eqnarray*}

(i) We have
\begin{eqnarray*}
\sum^n_{k=1}\alpha_k\bsa_k&=&(\alpha_1\bsa_1+\alpha_2\bsa_2+\alpha_3\bsa_3)
+\sum^n_{l=4}\alpha_l(\kappa_{l1}\bsa_1+\kappa_{l2}\bsa_2+\kappa_{l3}\bsa_3)\notag\\
&=& \Big ( \alpha_1+\sum^n_{l=4}\kappa_{l1}\alpha_l\Big ) \bsa_1+
\Big (\alpha_2+\sum^n_{l=4}\kappa_{l2}\alpha_l\Big )\bsa_2+\Big
(\alpha_3+\sum^n_{l=4}\kappa_{l3}\alpha_l\Big )\bsa_3.
\end{eqnarray*}
Let
\begin{eqnarray*}
(\alpha'_1,\alpha'_2,\alpha'_3)=\Big
(\alpha_1+\sum^n_{l=4}\kappa_{l1}\alpha_l,\alpha_2+\sum^n_{l=4}\kappa_{l2}\alpha_l,\alpha_3+\sum^n_{l=4}\kappa_{l3}\alpha_l\Big
), \quad (\alpha'_4,\ldots,\alpha'_n)=(\alpha_4,\ldots,\alpha_n),
\end{eqnarray*}
then the Jacobian of the transformation
$(\alpha_1,\ldots,\alpha_n)\to(\alpha'_1,\ldots,\alpha'_n)$ is given
by
\begin{eqnarray*}
\det\Pa{\frac{\partial(\alpha'_1,\ldots,\alpha'_n)}{\partial(\alpha_1,\ldots,\alpha_n)}}
= 1\neq0.
\end{eqnarray*}
Noting that
\begin{eqnarray*}
\sum^n_{j=1}(r_j-a^{(j)}_0)\alpha_j &=&
\sum^3_{j=1}(r_j-a^{(j)}_0)\Big (\alpha'_j-\sum^n_{l=4}\kappa_{lj}\alpha'_l\Big ) +\sum^n_{l=4}(r_l-a^{(l)}_0)\alpha'_l\notag\\
&=&\sum^3_{j=1}(r_j-a^{(j)}_0)\alpha'_j -\sum^n_{l=4}\Big
(\sum^3_{j=1}\kappa_{lj}(r_j-a^{(j)}_0)\Big )\alpha'_l
+\sum^n_{l=4}(r_l-a^{(l)}_0)\alpha'_l\notag\\
&=&\sum^3_{l=1}(r_l-a^{(l)}_0)\alpha'_l +\sum^n_{l=4}\Big
((r_l-a^{(l)}_0)-\sum^3_{j=1}\kappa_{lj}(r_j-a^{(j)}_0)\Big
)\alpha'_l,
\end{eqnarray*}
we have
\begin{eqnarray*}
& &f_{\langle\bsA_1\rangle,\ldots,\langle\bsA_n\rangle}(r_1,\ldots,r_n) \\
&=& \frac3{(2\pi)^3}\int _{\real^3} \prod^3_{l=1}\dif \alpha'_l
e^{\mathrm{i}\Pa{\sum^3_{l=1}(r_l-a^{(l)}_0)\alpha'_l}}\frac{\sin\Pa{\abs{\sum^3_{l=1}\alpha'_l\bsa_l}}-\abs{\sum^3_{l=1}\alpha'_l\bsa_l}
\cos\Pa{\abs{\sum^3_{l=1}\alpha'_l\bsa_l}}}{\abs{\sum^3_{l=1}\alpha'_l\bsa_l}^3}\notag\\
&&\times \prod^n_{l=4}\Pa{\frac1{2\pi}\int _{\real }\dif\alpha'_l
e^{\mathrm{i}\Pa{(r_l-a^{(l)}_0)-\sum^3_{j=1}\kappa_{lj}(r_j-a^{(j)}_0)}\alpha'_l}},
\end{eqnarray*}
therefore
\begin{eqnarray*}
f_{\langle\bsA_1\rangle,\ldots,\langle\bsA_n\rangle}(r_1,\ldots,r_n)
&=&
f_{\langle\bsA_1\rangle,\langle\bsA_2\rangle,\langle\bsA_3\rangle}(r_1,r_2,r_3)\prod^n_{l=4}\delta\Big
((r_l-a^{(l)}_0)-\sum^3_{j=1}\kappa_{lj}(r_j-a^{(j)}_0)\Big ).
\end{eqnarray*}
Items  (ii) and (iii) follow similarly.
\end{proof}

\begin{prop}\label{prop12}
The joint probability distribution density of $(\Delta_\rho
\bsA_1,\ldots,\Delta_\rho \bsA_n)$ for an $n$-triple of qubit
observables defined by Eq.~\eqref{eq:k}, where $\rho$ is resulted
from partial-tracing a subsystem over a Haar-distributed random pure
state on $\complex^2\ot\complex^2$, is given by
\begin{eqnarray*}
f_{\Delta\bsA_1,\ldots,\Delta\bsA_n}(x_1,\ldots,x_n)=\Pa{2\prod^n_{j=1}\frac{x_j}{\sqrt{\abs{\bsa_j}^2-x^2_j}}}\sum_{j_2,\ldots,j_n\in\set{\pm}}
f_{\langle \bsA_1\rangle,\ldots,\langle
\bsA_n\rangle}(r^{(1)}_+(x_1),r^{(2)}_{j_2}(x_2),\ldots,r^{(n)}_{j_n}(x_n)),
\end{eqnarray*}
where $f_{\langle \bsA_1\rangle,\ldots,\langle
\bsA_n\rangle}(\cdot,\ldots,\cdot)$ is is determined by
Proposition~\ref{prop11}, and
\begin{eqnarray*}
r^{(k)}_\pm(x_k):=a^{(k)}_0\pm\sqrt{\abs{\bsa_k}^2-x^2_k}\quad(k=1,\ldots,n).
\end{eqnarray*}
\end{prop}

\begin{proof}
The proof goes similarly for Propositions~\ref{prop5}, \ref{prop8},
and \ref{prop10}. We omitted here.
\end{proof}

\begin{proof}[Proof of Lemma~\ref{lem3}]
From Proposition~\ref{prop11}, we get the support of
$f_{\langle\bsA_1\rangle,\ldots,\langle\bsA_n\rangle}$ which is
given by, via
$$
\omega_{\bsA_1,\bsA_2,\bsA_3}(r_1,r_2,r_3)=\sqrt{\Pa{r_1-a^{(1)}_0,r_2-a^{(2)}_0,r_3-a^{(3)}_0}\bsT^{-1}_{\bsa_1,\bsa_2,\bsa_3}\Pa{r_1-a^{(1)}_0,r_2-a^{(2)}_0,r_3-a^{(3)}_0}^\t},
$$
\begin{eqnarray*}
\supp(f_{\langle\bsA_1\rangle,\ldots,\langle\bsA_n\rangle})=\Set{(r_1,\ldots,r_n)\in\real^n:\omega_{\bsA_1,\bsA_2,\bsA_3}(r_1,r_2,r_3)\leqslant1,r_l-a^{(l)}_0=\sum^3_{j=1}\kappa_{lj}(r_j-a^{(j)}_0)(\forall
l=4,\ldots,n)}.
\end{eqnarray*}
By similar analysis as in the proofs of Propositions~\ref{prop4},
\ref{prop7}, and \ref{prop9}, we obtain the support of
$f_{\Delta\bsA_1,\ldots,\Delta\bsA_n}$, i.e., the uncertainty region
$\cU_{\Delta\bsA_1,\ldots,\Delta\bsA_n}=\supp(f_{\Delta\bsA_1,\ldots,\Delta\bsA_n})$,
\begin{eqnarray*}
\supp(f_{\Delta\bsA_1,\ldots,\Delta\bsA_n}) =
\Set{(x_1,\ldots,x_n)\in\real^n_{\geqslant0}:\begin{cases}
\bsu_{\epsilon_1,\epsilon_2,\epsilon_3}(x_1,x_2,x_3)\bsT^{-1}_{\bsa_1,\bsa_2,\bsa_3}\bsu^\t_{\epsilon_1,\epsilon_2,\epsilon_3}(x_1,x_2,x_3)\leqslant1\\
\epsilon_l\sqrt{\abs{\bsa_l}^2-x^2_l}=\sum^3_{j=1}\kappa_{lj}\epsilon_j\sqrt{\abs{\bsa_j}^2-x^2_j}(\forall
l=4,\ldots,n)
\end{cases}}.
\end{eqnarray*}
Here $\epsilon_k\in\set{\pm1}$ and $x_k\in[0,\abs{\bsa_k}]$, where
$k=1,\ldots,n$.
\end{proof}

\subsection{Proof of Theorem~\ref{th3}}\label{app-theorem3}

We consider now the uncertainty regions for multiple qubit
observables. For an $n$-tuple of qubit observables
$(\bsA_1,\ldots,\bsA_n)$, where
$\bsA_k=a^{(k)}_0\I+\bsa_k\cdot\boldsymbol{\sigma}$ with
$(a^{(k)}_0, \bsa_k)\in\real^4$, $k=1,\ldots,n$, denote
$\bsT_{\bsa_1,\ldots,\bsa_n}=(\inner{\bsa_i}{\bsa_j})$. Note that
$\set{\bsa_1,\bsa_2,...,\bsa_n}$ has at most three vectors that are
linearly independent. Without loss of generality, we assume
$\set{\bsa_1,\bsa_2,\bsa_3}$ is linearly independent. The rest
vectors can be linearly expressed by $\set{\bsa_1,\bsa_2,\bsa_3}$,
$\bsa_l=\kappa_{l1}\bsa_1+\kappa_{l2}\bsa_2+\kappa_{l3}\bsa_3$, for
some coefficients $\kappa_{lj}$, $l=4,\cdots,n$, $j=1,2,3$. Set
$\bsu_{\epsilon_1,\epsilon_2,\epsilon_3}(x_1,x_2,x_3)=\Pa{\epsilon_1\sqrt{\abs{\bsa_1}^2-x^2_1},\epsilon_2\sqrt{\abs{\bsa_2}^2-x^2_2},
\epsilon_3\sqrt{\abs{\bsa_3}^2-x^2_3}}$, where
$\epsilon_k\in\set{\pm1}(k=1,2,3)$.

Based on the proof of Lemma~\ref{lem3}, we give the proof of
Theorem~\ref{th3}.

\begin{proof}[Proof of Theorem~\ref{th3}]
In Lemma~\ref{lem3}, let
$\sqrt{\abs{\bsa_k}^2-x^2_k}:=X_k\in[0,\abs{\bsa_k}](k=1,\ldots,n)$
due to the fact that $x_k\in[0,\abs{\bsa_k}]$. This implies that
\begin{eqnarray*}
&&\min_{(x_1,\ldots,x_n)\in\cU_{\Delta\bsA_1,\ldots,\Delta\bsA_n}}\sum^n_{k=1}x^2_k\\
&&= \Tr{\bsT_{\bsa_1,\ldots,\bsa_n}} -
\max\Set{\sum^n_{k=1}X^2_k:\begin{cases}(\epsilon_1 X_1,\epsilon_2
X_2,\epsilon_3 X_3)\bsT^{-1}_{\bsa_1,\bsa_2,\bsa_3}(\epsilon_1
X_1,\epsilon_2 X_2,\epsilon_3 X_3)\leqslant1\\
\epsilon_lX_l=\sum^3_{j=1}\kappa_{lj}\epsilon_jX_j(\forall
l=4,\ldots,n)
\end{cases}}\\
&&=\Tr{\bsT_{\bsa_1,\ldots,\bsa_n}} -
\max\Set{\sum^n_{k=1}Y^2_k:\begin{cases}(Y_1,Y_2,Y_3)\bsT^{-1}_{\bsa_1,\bsa_2,\bsa_3}(Y_1,Y_2,Y_3)^\t\leqslant1\\
Y_l=\sum^3_{j=1}\kappa_{lj}Y_j(\forall l=4,\ldots,n)
\end{cases}}.
\end{eqnarray*}
Next, we show that
\begin{eqnarray*}
\max\Set{\sum^n_{k=1}Y^2_k:\begin{cases}(Y_1,Y_2,Y_3)\bsT^{-1}_{\bsa_1,\bsa_2,\bsa_3}(Y_1,Y_2,Y_3)^\t\leqslant1\\
Y_l=\sum^3_{j=1}\kappa_{lj}Y_j(\forall l=4,\ldots,n)
\end{cases}} = \lambda_{\max}(\bsT_{\bsa_1,\ldots,\bsa_n}).
\end{eqnarray*}
Indeed, let $\bsP=(\bsa_1,\bsa_2,\bsa_3)$, where each $\bsa_j$ is a
column vector in $\real^3$. Due to the fact that
$\bsa_l=\sum^3_{j=1}\kappa_{lj}\bsa_j(\forall l=4,\ldots,n)$, we see
that
\begin{eqnarray*}
(\bsa_1,\bsa_2,\bsa_3)\Pa{\begin{array}{c}
      \kappa_{l1} \\
      \kappa_{l2} \\
      \kappa_{l3}
    \end{array}
} =\bsP\Pa{\begin{array}{c}
      \kappa_{l1} \\
      \kappa_{l2} \\
      \kappa_{l3}
    \end{array}
} = \bsa_l\Longleftrightarrow\Pa{\begin{array}{c}
      \kappa_{l1} \\
      \kappa_{l2} \\
      \kappa_{l3}
    \end{array}
} = \bsP^{-1}\bsa_l.
\end{eqnarray*}
Denote by $\bsy:=(Y_1,Y_2,Y_3)^\t$. Then
$Y_l=\sum^3_{j=1}\kappa_{lj}Y_j(l=4,\ldots,n)$ can be rewritten as
\begin{eqnarray*}
Y_l = \Inner{\Pa{\begin{array}{c}
                   \kappa_{l1} \\
                   \kappa_{l2} \\
                   \kappa_{l3}
                 \end{array}
}}{\Pa{\begin{array}{c}
         Y_1 \\
         Y_2 \\
         Y_3
       \end{array}
}} = \Inner{\bsP^{-1}\bsa_l}{\bsy}.
\end{eqnarray*}
Based on the above observation, we have immediately that
\begin{eqnarray*}
\sum^n_{j=1}Y^2_j &=& \sum^3_{j=1}Y^2_j + \sum^n_{l=4}Y^2_l =
\Inner{\bsy}{\bsy}+
\sum^n_{l=4}\Inner{\bsy}{\bsP^{-1}\bsa_l}\Inner{\bsP^{-1}\bsa_l}{\bsy}\\
&=&\Innerm{\bsy}{\I}{\bsy}+\Innerm{\bsy}{\sum^n_{l=4}\Pa{\bsP^{-1}\bsa_l}\Pa{\bsP^{-1}\bsa_l}^\t}{\bsy}\\
&=&\Innerm{\bsy}{\I+\sum^n_{l=4}\Pa{\bsP^{-1}\bsa_l}\Pa{\bsP^{-1}\bsa_l}^\t}{\bsy},
\end{eqnarray*}
where, via $\bsQ:=\bsP^{-1}$,
\begin{eqnarray*}
\I+\sum^n_{l=4}\Pa{\bsP^{-1}\bsa_l}\Pa{\bsP^{-1}\bsa_l}^\t
=\bsQ\Pa{\sum^n_{j=1}\bsa_j\bsa^\t_j}\bsQ^\t
\end{eqnarray*}
implying that
\begin{eqnarray*}
\sum^n_{j=1}Y^2_j =
\Innerm{\bsy}{\bsQ\Pa{\sum^n_{j=1}\bsa_j\bsa^\t_j}\bsQ^\t}{\bsy}.
\end{eqnarray*}
Note that $\bsT_{\bsa_1,\bsa_2,\bsa_3}=\bsP^\t\bsP$. From this, we
see that $\bsy\bsT^{-1}_{\bsa_1,\bsa_2,\bsa_3}\bsy^\t\leqslant1$,
which is equivalent to $\Innerm{\bsy}{\bsQ\bsQ^\t}{\bsy}\leqslant1$.
Denote $\bsv=Q^\t\bsy$. This indicates that
\begin{eqnarray*}
&&\max\Set{\sum^n_{k=1}Y^2_k:\begin{cases}(Y_1,Y_2,Y_3)\bsT^{-1}_{\bsa_1,\bsa_2,\bsa_3}(Y_1,Y_2,Y_3)^\t\leqslant1\\
Y_l=\sum^3_{j=1}\kappa_{lj}Y_j(\forall l=4,\ldots,n)
\end{cases}} \\
&&=\max_{\Inner{\bsv}{\bsv}\leqslant1}\Innerm{\bsv}{\sum^n_{j=1}\bsa_j\bsa^\t_j}{\bsv}=\lambda_{\max}\Pa{\sum^n_{j=1}\bsa_j\bsa^\t_j}\\
&&=\lambda_{\max}(\bsT_{\bsa_1,\ldots,\bsa_n}).
\end{eqnarray*}
Here the last equality is true since both $\bsM\bsM^\t$ and
$\bsM^\t\bsM$ have the same spectrum when ignoring the zero
eigenvalues for real matrix $\bsM$. Hence the same maximal
eigenvalues
$\lambda_{\max}(\bsM\bsM^\t)=\lambda_{\max}(\bsM^\t\bsM)$. Let
$\bsM=(\bsa_1,\ldots,\bsa_n)$. Then
\begin{eqnarray*}
\sum^n_{j=1}\bsa_j\bsa^\t_j = \bsM\bsM^\t,\quad
\bsT_{\bsa_1,\ldots,\bsa_n}=\bsM^\t\bsM.
\end{eqnarray*}
These can give our desired result.
\end{proof}

\subsection*{Acknowledgments}

This work is supported by the NSF of China under Grant Nos.
11971140, 61771174 and 12075159, Beijing Natural Science Foundation
(Z190005), Key Project of Beijing Municipal Commission of Education
(KZ201810028042), the Academician Innovation Platform of Hainan
Province, Academy for Multidisciplinary Studies, Capital Normal
University, and Shenzhen Institute for Quantum Science and
Engineering, Southern University of Science and Technology (Grant
No. SIQSE202001).


\end{document}

%% file: ur_preamble.tex
\newcommand{\cmp}{Comm. Math. Phys.~}

\newcommand{\jmp}{J. Math. Phys.~}
\newcommand{\jpa}{J. Phys. A~}

\newcommand{\natphy}{Nature Phys.~}

\newcommand{\njp}{New. J. Phys.~}

\newcommand{\prl}{Phys. Rev. Lett.~}
\newcommand{\pra}{Phys. Rev. A~}
\newcommand{\pre}{Phys. Rev. E~}
\newcommand{\pla}{Phys. Lett. A~}

\newcommand{\sci}{Science}

\usepackage{appendix}
\usepackage{epstopdf}
\usepackage{color}
\usepackage{subfigure}
\usepackage[T1]{fontenc}
\usepackage[sc]{mathpazo}
\usepackage{amsmath}
\usepackage{amssymb}
\usepackage{graphicx,color}
\usepackage{framed}
\usepackage{multirow}
\usepackage{enumerate}
\usepackage{amsthm}
\usepackage{amsfonts,mathrsfs}
\usepackage{geometry} 
\usepackage{eepic}
\usepackage{ifthen}
\usepackage[vcentermath]{youngtab}
\usepackage[unicode=true,pdfusetitle, bookmarks=true,bookmarksnumbered=false,bookmarksopen=false, breaklinks=false,pdfborder={0 0 0},backref=false,colorlinks=false] {hyperref}
\hypersetup{
colorlinks,linkcolor=myurlcolor,citecolor=myurlcolor,urlcolor=myurlcolor}
\definecolor{myurlcolor}{rgb}{0,0,0.7}

\usepackage{color}

\usepackage{hyperref}
\hypersetup{pdfpagemode=UseNone}

\newcommand{\tinyspace}{\mspace{1mu}}

\newcommand{\op}[1]{\operatorname{#1}}

\newcommand{\abs}[1]{\left\lvert\tinyspace #1 \tinyspace\right\rvert}

\newcommand{\norm}[1]{\left\lVert\tinyspace #1 \tinyspace\right\rVert}

\renewcommand{\det}{\operatorname{det}}
\renewcommand{\t}{{\scriptscriptstyle\mathsf{T}}}

\newcommand{\setft}[1]{\mathrm{#1}}

\newcommand{\density}[1]{\setft{D}\left(#1\right)}

\newcommand{\rank}{\op{rank}}

\newcommand{\sign}{\op{sign}}

\newcommand{\supp}{{\operatorname{supp}}}

\def\vol{\mathrm{vol}}

\def\haar{\mathrm{Haar}}

\def \dif {\mathrm{d}}
\def \diag {\mathrm{diag}}
\def \vol {\mathrm{vol}}

\def\complex{\mathbb{C}}
\def\real{\mathbb{R}}

\def\I{\mathbb{1}}

\newenvironment{mylist}[1]{\begin{list}{}{
    \setlength{\leftmargin}{#1}
    \setlength{\rightmargin}{0mm}
    \setlength{\labelsep}{2mm}
    \setlength{\labelwidth}{8mm}
    \setlength{\itemsep}{0mm}}}
    {\end{list}}

\def\ot{\otimes}

\newcommand{\inner}[2]{\langle #1 , #2\rangle}

\newcommand{\Inner}[2]{\left\langle #1 , #2\right\rangle}
\newcommand{\Innerm}[3]{\left\langle #1 \left| #2 \right| #3 \right\rangle}

\newcommand{\Pa}[1]{\left(#1\right)}

\newcommand{\Br}[1]{\left[#1\right]}
\newcommand{\set}[1]{\{#1\}}
\newcommand{\Set}[1]{\left\{#1\right\}}

\newcommand{\ket}[1]{|#1\rangle}

\DeclareMathOperator{\trace}{Tr}

\newcommand{\Ptr}[2]{\trace_{#1}\Pa{#2}}

\newcommand{\Tr}[1]{\Ptr{}{#1}}

\def\cU{\mathcal{U}}

\def\bsA{\boldsymbol{A}}\def\bsB{\boldsymbol{B}}\def\bsC{\boldsymbol{C}}

\def\bsM{\boldsymbol{M}}\def\bsO{\boldsymbol{O}}
\def\bsP{\boldsymbol{P}}\def\bsQ{\boldsymbol{Q}}\def\bsS{\boldsymbol{S}}\def\bsT{\boldsymbol{T}}
\def\bsU{\boldsymbol{U}}\def\bsX{\boldsymbol{X}}

\def\bsa{\boldsymbol{a}}\def\bsb{\boldsymbol{b}}\def\bsc{\boldsymbol{c}}

\def\bsu{\boldsymbol{u}}\def\bsv{\boldsymbol{v}}\def\bsy{\boldsymbol{y}}
\def\bsz{\boldsymbol{z}}

\def\rD{\mathrm{D}}

\def\O{\textsf{O}}

\def\U{\textsf{U}}

\newtheorem{thrm}{Theorem}[section]
\newtheorem{lem}[thrm]{Lemma}
\newtheorem{prop}[thrm]{Proposition}

\theoremstyle{definition}

\numberwithin{equation}{section}

\newcounter{questionnumber}